\documentclass{llncs_template/llncs}

\usepackage{xcolor}
\usepackage{amsmath}
\usepackage{amssymb}
\usepackage{tikz}
\usepackage{caption}
\usepackage{subcaption}
\usepackage{hyperref}

\newcommand{\powset}[1]{\mathcal{P}(#1)}
\newcommand{\powsetnop}{\mathcal{P}}

\newcommand{\allmatches}[2]{\mathcal{M}^{#1}_{#2}}

\newcommand{\results}[4]{\mathcal{R}(#1, #2, #3, #4)}
\newcommand{\resultsnop}{\mathcal{R}}
\newcommand{\resultslocal}[5]{\mathcal{R}^\Phi(#1, #2, #3, #4, #5)}

\newcommand{\resultsstripped}[2]{\mathcal{R}^\Phi_X(#1, #2)}

\usetikzlibrary{positioning}
\usetikzlibrary{arrows.meta}
\tikzset{retenode/.style = {rectangle, text centered, text height = 0.25cm, draw=black, minimum width = 1.7cm, minimum height = 0.5cm}}
\tikzset{ms_retenode/.style = {rectangle, text centered, text height = 0.25cm, draw=black, minimum width = 1.7cm, minimum height = 0.5cm, label={[xshift=0.75cm, yshift=-0.35cm]\scriptsize$\Phi$}}}

\tikzset{retenode_app/.style = {rectangle, text centered, text height = 0.25cm, draw=black, minimum width = 1.7cm, minimum height = 0.75cm}}
\tikzset{ms_retenode_app/.style = {rectangle, text centered, text height = 0.25cm, draw=black, minimum width = 1.7cm, minimum height = 0.75cm, label={[xshift=0.75cm, yshift=-0.35cm]\scriptsize$\Phi$}}}

\tikzset{vertex/.style = {circle, text centered, text height = 0.25cm, draw=black, minimum width = 1.5cm, minimum height = 1.5cm}}

\begin{document}

\title{Localized RETE for Incremental Graph Queries\thanks{This preprint has not undergone peer review or any post-submission improvements or corrections. The Version of Record of this contribution is published in Graph Transformation, and is available online at \url{https://doi.org/10.1007/978-3-031-64285-2_7}}}

\author{Matthias Barkowsky \and
Holger Giese}

\institute{Hasso Plattner Institute at the University of Potsdam\\
Prof.-Dr.-Helmert-Str. 2-3, D-14482 Potsdam, Germany\\
\email{\{matthias.barkowsky,holger.giese\}@hpi.de}}

\maketitle

\begin{abstract}
\textbf{Context:} The growing size of graph-based modeling artifacts in model-driven engineering calls for techniques that enable efficient execution of graph queries. Incremental approaches based on the RETE algorithm provide an adequate solution in many scenarios, but are generally designed to search for query results over the entire graph. However, in certain situations, a user may only be interested in query results for a subgraph, for instance when a developer is working on a large model of which only a part is loaded into their workspace. In this case, the global execution semantics can result in significant computational overhead.

\textbf{Contribution:} To mitigate the outlined shortcoming, in this paper we propose an extension of the RETE approach that enables local, yet fully incremental execution of graph queries, while still guaranteeing completeness of results with respect to the relevant subgraph.

\textbf{Results:} We empirically evaluate the presented approach via experiments inspired by a scenario from software development and an independent social network benchmark. The experimental results indicate that the proposed technique can significantly improve performance regarding memory consumption and execution time in favorable cases, but may incur a noticeable linear overhead in unfavorable cases.
\end{abstract}


\section{Introduction} \label{sec:introduction}

In model-driven engineering, models constitute important development artifacts \cite{kent2002model}. With complex development projects involving large, interconnected models, performance of automated model operations becomes a primary concern.

Incremental graph query execution based on the RETE algorithm \cite{forgy1989rete} has been demonstrated to be an adequate solution in scenarios where an evolving model is repeatedly queried for the same information \cite{szarnyas2018train}. In this context, the RETE algorithm essentially tackles the problem of querying a usually graph-based model representation via operators from relational algebra \cite{codd1970relational}. With these operators not designed to exploit the locality found in graph-based encodings, current incremental querying techniques require processing of the entire model to guarantee complete results. However, in certain situations, global query execution is not required and may be undesirable due to performance considerations.

For instance, a developer may be working on only a part of a model loaded into their workspace, with a large portion of the full model still stored on disk. A concrete example of this would be a developer using a graphical editor to modify an individual block diagram from a collection of interconnected block diagrams, which together effectively form one large model. As the developer modifies the model part in their workspace, they may want to continuously monitor how their changes impact the satisfaction of some consistency constraints via incremental model queries. In this scenario, existing techniques for incremental query execution require loading and querying the entire collection of models, even though the user is ultimately only interested in the often local effect of their own changes.

The global RETE execution semantics then results in at least three problems: First, the computation of query results for the full model, of which only a portion is relevant to the user, may incur a substantial overhead on initial query execution time. Second, incremental querying techniques are known to create and store a large number of intermediate results \cite{szarnyas2018train}, many of which can in this scenario be superfluous, causing an overhead in memory consumption. Third, query execution requires loading the entire model into memory, potentially causing an overhead in loading time and increasing the overhead in memory consumption.

These problems can be mitigated to some extent by employing local search \cite{cordella2004sub,giese2009improved,arendt2010henshin}, which better exploits the locality of the problem and lazy loading capabilities of model persistence layers \cite{cdo,daniel2017neoemf}, instead of a RETE-based technique. However, resorting to local search can result in expensive redundant search operations that are only avoided by fully incremental solutions \cite{szarnyas2018train}.

In this paper, we instead propose to tackle the outlined shortcoming of incremental querying techniques via an extension of the RETE approach. This is supported by a relaxed notion of completeness for query results that accounts for situations where the computation for the full model is unnecessary. Essentially, this enables a distinction between the full model, for which query results do not necessarily have to be complete, and a relevant model part, for which complete results are required. The extended RETE approach then anchors query execution to the relevant model part and lazily fetches additional model elements required to compute query results that meet the relaxed completeness requirement. Our approach thereby avoids potentially expensive global query execution and allows an effective integration of incremental queries with model persistence layers.

The remainder of this paper is structured as follows:
Section \ref{sec:preliminaries} provides a summary of our notion of graphs, graph queries, and the RETE mechanism for query execution.
Section \ref{sec:incremental_queries_over_persistent_models} discusses our contribution in the form of a relaxed notion of completeness for query results and an extension to the RETE querying mechanism that enables local, yet fully incremental execution of model queries.
A prototypical implementation of our approach is evaluated regarding execution time and memory consumption in Section \ref{sec:evaluation}.
Section \ref{sec:related_work} discusses related work, before Section \ref{sec:conclusion} concludes the paper and provides an outlook on future work.


\section{Preliminaries} \label{sec:preliminaries}

In the following, we briefly reiterate the definitions of graphs and graph queries and summarize the RETE approach to incremental graph querying.

\subsection{Graphs and Graph Queries} \label{sec:preliminaries_graphs}


As defined in \cite{Ehrig+2006}, a \emph{graph} is a tuple $G = (V^G, E^G, s^G, t^G)$, with $V^G$ the set of vertices, $E^G$ the set of edges, and $s^G : E^G \rightarrow V^G$ and $t^G : E^G \rightarrow V^G$ functions mapping edges to their source respectively target vertices.

A mapping from a graph $Q$ into another graph $H$ is defined via a \emph{graph morphism} $m : Q \rightarrow H$, which is characterized by two functions $m^V : V^Q \rightarrow V^H$ and $m^E : E^Q \rightarrow E^H$ such that $s^H \circ m^E = m^V \circ s^Q$ and $t^H \circ m^E = m^V \circ t^Q$.

A graph $G$ can be typed over a type graph $TG$ via a graph morphism $type^G : G \rightarrow TG$, forming the \emph{typed graph} $G^T =(G, type^G)$. A \emph{typed graph morphism} from $G^T = (G, type^G)$ into another typed graph $H^T = (H, type^H)$ with type graph $TG$ is a graph morphism $m^T : G \rightarrow H$ such that $type^H \circ m^T = type^G$.

We say that $G$ is \emph{edge-dominated} if $\forall e_{TG} \in E^{TG} : |\{e \in E^G | type(e) = e_{TG}\}| \geq max(|\{v \in V^G | type(v) = s^{TG}(e_{TG})\}|, |\{v \in V^G | type(v) = t^{TG}(e_{TG})\}|)$.


A \emph{model} is then simply given by a (typed) graph. In the remainder of the paper, we therefore use the terms graph and model interchangably.

A \emph{graph query} as considered in this paper is characterized by a (typed) query graph $Q$ and can be executed over a (typed) host graph $H$ by finding (typed) graph morphisms $m : Q \rightarrow H$. We also call these morphisms \emph{matches} and denote the set of all matches from $Q$ into $H$ by $\allmatches{Q}{H}$. Typically, a set of matches is considered a complete query result if it contains all matches in $\allmatches{Q}{H}$.

\subsection{Incremental Graph Queries with RETE}


The RETE algorithm \cite{forgy1989rete} forms the basis of mature incremental graph querying techniques \cite{varro2016}. Therefore, the query graph is recursively decomposed into simpler subqueries, which are arranged in a second graph called RETE net. In the following, we will refer to the vertices of RETE nets as \emph{(RETE) nodes}.

Each RETE node is associated with a subgraph of the query graph for which it computes matches. This computation may depend on matches computed by other RETE nodes. Such dependencies are represented by edges from the dependent node to the dependency node. For each RETE net, one of its nodes is designated as the \emph{production node}, which computes the net's overall result. A RETE net is thus given by a tuple $(N, p)$, where the \emph{RETE graph} $N$ is a graph of RETE nodes and dependency edges and $p \in V^N$ is the production node.


We describe the \emph{configuration} of a RETE net $(N, p)$ during execution by a function $\mathcal{C} : V^N \rightarrow \powset{\mathcal{M}_\Omega}$, with $\powsetnop$ denoting the power set and $\mathcal{M}_\Omega$ the set of all graph morphisms. $\mathcal{C}$ then assigns each node in $V^N$ a \emph{current result set}.

For a starting configuration $\mathcal{C}$ and host graph $H$, executing a RETE node $n \in V^N$ with associated query subgraph $Q$ yields the updated configuration $\mathcal{C}' = execute(n, N, H, \mathcal{C})$, with

\[
\mathcal{C}'(n') =
	\begin{cases}
	\results{n}{N}{H}{\mathcal{C}}	& \quad \text{if } $n' = n$\\
	\mathcal{C}(n')			& \quad \text{otherwise},
	\end{cases}
\]

where $\resultsnop$ is a function defining the \emph{target result set} of a RETE node $n$  in the RETE graph $N$ for $H$ and $\mathcal{C}$ such that $\results{n}{N}{H}{\mathcal{C}} \subseteq \allmatches{Q}{H}$.

We say that $\mathcal{C}$ is \emph{consistent} for a RETE node $n \in V^N$ and host graph $H$ iff $\mathcal{C}(n) = \results{n}{N}{H}{\mathcal{C}}$. If $\mathcal{C}$ is consistent for all $n \in V^N$, we say that $\mathcal{C}$ is consistent for $H$. If $H$ is clear from the context, we simply say that $\mathcal{C}$ is consistent.


Given a host graph $H$ and a starting configuration $\mathcal{C}_0$, a RETE net $(N, p)$ is executed via the execution of a sequence of nodes $O = n_1, n_2, ..., n_x$ with $n_i \in V^N$. This yields the trace $\mathcal{C}_0, \mathcal{C}_1, \mathcal{C}_2, ..., \mathcal{C}_x$, with $\mathcal{C}_y = execute(n_y, N, H, \mathcal{C}_{y-1})$, and the result configuration $execute(O, N, H, \mathcal{C}_0) = \mathcal{C}_x$.

$(N, p)$ can initially be executed over $H$ via a sequence $O$ that yields a consistent configuration $\mathcal{C}_x = execute(O, N, H, \mathcal{C}_0)$, where $\mathcal{C}_0$ is an empty starting configuration with $\forall n \in V^N : \mathcal{C}_0(n) = \emptyset$. Incremental execution of $(N, p)$ can be achieved by retaining a previous result configuration and using it as the starting configuration for execution over an updated host graph. This requires incremental implementations of RETE node execution procedures that update previously computed result sets for changed inputs instead of computing them from scratch.


In the most basic form, a RETE net for the query graph $Q$ consists of two kinds of nodes, \emph{edge input nodes} and \emph{join nodes}. Edge input nodes have no dependencies, are associated with a query subgraph consisting of a single edge and its adjacent vertices, and directly extract the (trivial) matches for $Q$ from a host graph. Join nodes have two dependencies with some associated query subgraphs $Q_l$ and $Q_r$ such that $V^{Q_l \cap Q_r} \neq \emptyset$ and are associated with the subgraph $Q = Q_l \cup Q_r$. A join node computes matches for this union subgraph by combining the matches from its dependencies along the overlap graph $Q_\cap = Q_l \cap Q_r$.

In the following, we assume query graphs to be weakly connected and contain at least one edge. A RETE net that computes matches for a query graph $Q$ can then be constructed as a tree of join nodes over edge input nodes. The join nodes thus gradually compose the trivial matches at the bottom into matches for more complex query subgraphs. We call such tree-like RETE nets \emph{well-formed}. An execution sequence that always produces a consistent configuration is given by a reverse topological sorting of the net. The root node of the tree then computes the set of all matches for $Q$ and is designated as the net's production node.

Connected graphs without edges consist of only a single vertex, making query execution via \emph{vertex input nodes}, which function analogously to edge input nodes, trivial. Disconnected query graphs can be handled via separate RETE nets for all query graph components and the computation of a cartesian product.

Figure \ref{fig:example_query_and_rete} shows an example graph query from the software domain and an associated RETE net. The query searches for paths of a package, class, and field. The RETE net constructs matches for the query by combining edges from a package to a class with edges from a class to a field via a natural join.

\begin{figure}
	\centering
	\begin{subfigure}{0.5\textwidth}
		\centering

\begin{tikzpicture}

\node (p) [vertex] {\textit{p:Pkg}};

\node[right = 0.5cm of p] (c) [vertex] {\textit{c:Class}};

\node[right = 0.5cm of c] (f) [vertex] {\textit{f:Field}};

\draw [-{Latex}] (p) -- (c) node [midway, above] {\textit{ce}};

\draw [-{Latex}] (c) -- (f) node [midway, above] {\textit{fe}};

\end{tikzpicture}
	\end{subfigure}\quad\vline\quad
	\begin{subfigure}{0.4\textwidth}
		\centering

\begin{tikzpicture}

\node (join) [retenode] {$\bowtie$};

\node[below left = 0.35cm and -0.5cm of join] (ce) [retenode] {$p \rightarrow c$};

\node[below right = 0.35cm and -0.5cm of join] (fe) [retenode] {$c \rightarrow f$};

\draw [-{Latex}] (join) -- (ce);

\draw [-{Latex}] (join) -- (fe);

\end{tikzpicture}
	\end{subfigure}
	\caption{Example graph query (left) and corresponding RETE net (right)} \label{fig:example_query_and_rete}
\end{figure}
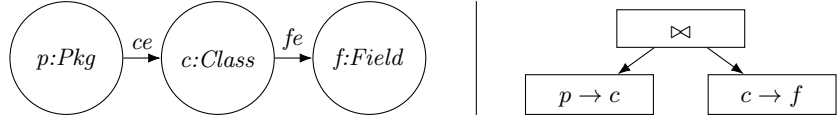


\section{Incremental Queries over Subgraphs} \label{sec:incremental_queries_over_persistent_models}


As outlined in Section \ref{sec:introduction}, users of model querying mechanisms may only be interested in query results, that is, matches, related to some part of a model that is relevant to them rather than the complete model. However, simply executing a query only over the relevant subgraph and ignoring context elements is not sufficient in many situations, for instance if the full effect of modifications to the relevant subgraph should be observed, since such modifications may affect matches that involve elements outside the relevant subgraph. Instead, completeness in such scenarios essentially requires the computation of all matches that somehow \emph{touch} the relevant subgraph.

In order to capture this need for local completeness, but avoid the requirement for global execution inherent to the characterization in Section \ref{sec:preliminaries_graphs}, we define completeness under a relevant subgraph of some host graph as follows:

\begin{definition} \label{def:completeness_subgraphs} \emph{(Completeness of Query Results under Subgraphs)}
	We say that a set of matches $M$ from query graph $Q$ into host graph $H$ is complete under a subgraph $H_p \subseteq H$
	iff $\forall m \in \allmatches{Q}{H}: (\exists v \in V^Q : m(v) \in V^{H_p}) \Rightarrow m \in M$.
\end{definition}

\subsection{Marking-sensitive RETE}


Due to its reliance on edge and vertex input nodes and their global execution semantics, incremental query execution via the standard RETE approach is unable to exploit the relaxed notion of completeness from Definition \ref{def:completeness_subgraphs} for query optimization and does not integrate well with mechanisms relying on operation locality, such as model persistence layers based on lazy loading \cite{cdo,daniel2017neoemf}.

While query execution could be localized to a relevant subgraph by executing edge and vertex inputs nodes only over the subgraph, execution would then only yield matches where \emph{all} involved elements are in the relevant subgraph. This approach would hence fail to meet the completeness criterion of Definition \ref{def:completeness_subgraphs}.

To enable incremental queries with complete results under the relevant subgraph, we instead propose to anchor RETE net execution to subgraph elements while allowing the search to retrieve elements outside the subgraph that are required to produce complete results from the full model.


Therefore, we extend the standard RETE mechanism by a marking for intermediate results, which allows such intermediate results to carry information that can later be used to control their propagation in the RETE net. In our extension, an intermediate result is therefore characterized by a tuple $(m, \phi)$ of a match $m$ and a marking $\phi \in \overline{\mathbb{N}}$, where we define $\overline{\mathbb{N}} := \mathbb{N} \cup \{\infty\}$. A marking-sensitive configuration is hence given by a function $\mathcal{C}^\Phi : V^N \rightarrow \powset{\mathcal{M}_\Omega \times \overline{\mathbb{N}}}$.

Furthermore, in our extension, result computation distinguishes between the full host graph $H$ and the relevant subgraph $H_p \subseteq H$. For marking-sensitive RETE nodes, we hence extend the function for target result sets by a parameter for $H_p$. The target result set of a marking-sensitive RETE node $n^\Phi \in V^{N^\Phi}$ with associated subgraph $Q$ for $H$, $H_p$, and a marking-sensitive configuration $\mathcal{C}^\Phi$ is then given by $\resultslocal{n^\Phi}{N^\Phi}{H}{H_p}{\mathcal{C}^\Phi}$, with $\resultslocal{n^\Phi}{N^\Phi}{H}{H_p}{\mathcal{C}^\Phi} \subseteq \allmatches{Q}{H} \times \overline{\mathbb{N}}$.


We adapt the \emph{join node}, \emph{union node}, and \emph{projection node} to marking-sensitive variants that assign result matches the maximum marking of related dependency matches and otherwise work as expected from relational algebra \cite{codd1970relational}. We also adapt the \emph{vertex input node} to only consider vertices in the relevant subgraph $H_p$ and assign matches the marking $\infty$. Finally, we introduce \emph{marking assignment nodes}, which assign matches a fixed marking value, \emph{marking filter nodes}, which filter marked matches by a minimum marking value, and \emph{forward and backward navigation nodes}, which work similarly to edge input nodes but only extract edges that are adjacent to host graph vertices included in the current result set of a designated dependency. Note that an efficient implementation of the backward navigation node requires reverse navigability of host graph edges.\footnote{Formal definitions for the adapted node semantics can be found in Appendix \ref{app:technical_details}.}

To obtain query results in the format of the standard RETE approach, we define the \emph{stripped result set} of a marking-sensitive RETE node $n^\Phi$ for a marking-sensitive configuration $\mathcal{C}^\Phi$ as $\resultsstripped{n^\Phi}{\mathcal{C}^\Phi} = \{m | (m, \phi) \in \mathcal{C}^\Phi(n^\Phi)\}$, that is, the set of matches that appear in tuples in the node's current result set in $\mathcal{C}^\Phi$.

\subsection{Localized Search with Marking-sensitive RETE} \label{sec:localized_search_with_marking-sensitive_rete}

Based on these adaptations, we introduce a recursive $localize$ procedure, which takes a regular, well-formed RETE net $(N, p)$ as input and outputs a marking-sensitive RETE net that performs a localized search that does not require searching the full model to produce complete results according to Definition \ref{def:completeness_subgraphs}.

If $p = [v \rightarrow w]$ is an edge input node, the result of localization for $(N, p)$ is given by $localize((N, p)) = (LNS(p), [\cup]^\Phi)$. The \emph{local navigation structure} $LNS(p)$ consists of seven RETE nodes as shown in Figure \ref{fig:localization_structures} (left): (1, 2) Two marking-sensitive vertex input nodes $[v]^\Phi$ and $[w]^\Phi$, (3, 4) two marking-sensitive union nodes $[\cup]^\Phi_v$ and $[\cup]^\Phi_w$ with $[v]^\Phi$ respectively $[w]^\Phi$ as a dependency, (5) a forward navigation node $[v \rightarrow_n w]^\Phi$ with $[\cup]^\Phi_v$ as a dependency, (6) a backward navigation node $[w \leftarrow_n v]^\Phi$ with $[\cup]^\Phi_w$ as a dependency, and (7) a marking-sensitive union node $[\cup]^\Phi$ with dependencies $[v \rightarrow_n w]^\Phi$ and $[w \leftarrow_n v]^\Phi$.

Importantly, the marking-sensitive vertex input nodes of the local navigation structure are executed over the relevant subgraph, whereas the forward and backward navigation nodes are executed over the full model. Intuitively, the local navigation structure thus takes the role of the edge input node, but initially only extracts edges that are adjacent to a vertex in the relevant subgraph.

If $p$ is a join node, it has two dependencies $p_l$ and $p_r$ with associated query subgraphs $Q_l$ and $Q_r$, which are the roots of two RETE subtrees $N_l$ and $N_r$. In this case, $(N, p)$ is localized as $localize((N, p)) = (N^\Phi, p^\Phi) = (N^\Phi_{\bowtie} \cup N^\Phi_l \cup N^\Phi_r \cup RPS_l \cup RPS_r, [\bowtie]^\Phi)$, where $(N^\Phi_l, p^\Phi_l) = localize((N_l, p_l))$, $(N^\Phi_r, p^\Phi_r) = localize((N_r, p_r))$, $N^\Phi_{\bowtie}$ consists of the marking-sensitive join $[\bowtie]^\Phi$ with dependencies $p^\Phi_l$ and $p^\Phi_r$, $RPS_l = RPS(p^\Phi_l, N^\Phi_r)$, and $RPS_r = RPS(p^\Phi_r, N^\Phi_l)$.

The \emph{request projection structure} $RPS_l = RPS(p^\Phi_l, N^\Phi_r)$ contains three RETE nodes as displayed in Figure \ref{fig:localization_structures} (center): (1) A marking filter node $[\phi > h]^\Phi$ with $p^\Phi_l$ as a dependency, (2) a marking-sensitive projection node $[\pi_{Q_v}]^\Phi$ with $[\phi > h]^\Phi$ as a dependency, and (3) a marking assignment node $[\phi := h]^\Phi$ with $[\pi_{Q_v}]^\Phi$ as a dependency. The value of $h$ is given by the height of the join tree of which $p$ is the root and $Q_v$ is a graph consisting of a single vertex $v \in V^{Q_l \cap Q_r}$. The request projection structure is then connected to an arbitrary local navigation structure in $N^\Phi_r$ that has a marking-sensitive node input $[v]^\Phi$ associated with $Q_v$. Therefore, it also adds a dependency from the marking-sensitive union node $[\cup]^\Phi_v$ that already depends on $[v]^\Phi$ to the marking assignment node $[\phi := h]^\Phi$. The mirrored structure $RPS_r = RPS(p^\Phi_r, N^\Phi_l)$ is constructed analogously.

Via the request projection structures, partial matches from one join dependency can be propagated to the subnet under the other dependency. Intuitively, the inserted request projection structures thereby allow the join's dependencies to request the RETE subnet under the other dependency to fetch and process the model parts required to complement the first dependency's results. The marking of a match then signals up to which height in the join tree complementarity for that match is required. Notably, matches involving elements in the relevant subgraph are marked $\infty$, as complementarity for them is required at the very top to guarantee completeness of the overall result.

The result of applying $localize$ to the example RETE net in Figure \ref{fig:example_query_and_rete} is displayed in Figure \ref{fig:localization_structures} (right). It consists of a marking sensitive join and two local navigation structures connected via request projection structures.

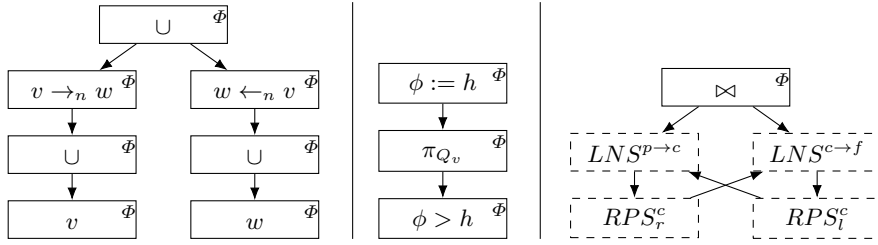
\begin{figure} [t]
	\centering
	\begin{subfigure}[b]{0.35\textwidth}
		\begin{center}

\begin{tikzpicture}

\node (union) [ms_retenode] {$\cup$};

\node[below left = 0.35cm and -0.5cm of union] (forward) [ms_retenode] {$v \rightarrow_n w$};

\node[below right = 0.35cm and -0.5cm of union] (backward) [ms_retenode] {$w \leftarrow_n v$};

\node[below = 0.35cm of forward] (unionv) [ms_retenode] {$\cup$};

\node[below = 0.35cm of backward] (unionw) [ms_retenode] {$\cup$};

\node[below = 0.35cm of unionv] (v) [ms_retenode] {$v$};

\node[below = 0.35cm of unionw] (w) [ms_retenode] {$w$};

\draw [-{Latex}] (union) -- (forward);

\draw [-{Latex}] (union) -- (backward);

\draw [-{Latex}] (forward) -- (unionv);

\draw [-{Latex}] (backward) -- (unionw);

\draw [-{Latex}] (unionv) -- (v);

\draw [-{Latex}] (unionw) -- (w);

\end{tikzpicture}
		\end{center}
	\end{subfigure}\quad\vline\quad
	\begin{subfigure}[b]{0.15\textwidth}
		\begin{center}

\begin{tikzpicture}

\node (assignment) [ms_retenode] {$\phi := h$};

\node[below = 0.35cm of assignment] (projection) [ms_retenode] {$\pi_{Q_v}$};

\node[below = 0.35cm of projection] (filter) [ms_retenode] {$\phi > h$};

\draw [-{Latex}] (assignment) -- (projection);

\draw [-{Latex}] (projection) -- (filter);

\end{tikzpicture}
		\end{center}
	\end{subfigure}\quad\vline\quad
	\begin{subfigure}[b]{0.35\textwidth}
		\begin{center}

\begin{tikzpicture}

\node (join) [ms_retenode] {$\bowtie$};

\node[below left = 0.35cm and -0.5cm of join, dashed] (ce) [retenode] {$LNS^{p \rightarrow c}$};

\node[below right = 0.35cm and -0.5cm of join, dashed] (fe) [retenode] {$LNS^{c \rightarrow f}$};

\node[below = 0.35cm of ce, dashed] (rpsp) [retenode] {$RPS^c_r$};

\node[below = 0.35cm of fe, dashed] (rpsf) [retenode] {$RPS^c_l$};

\draw [-{Latex}] (join) -- (ce);

\draw [-{Latex}] (join) -- (fe);

\draw [-{Latex}] (ce) -- (rpsp);

\draw [-{Latex}] (fe) -- (rpsf);

\draw [-{Latex}] (rpsp) -- (fe);

\draw [-{Latex}] (rpsf) -- (ce);

\end{tikzpicture}
		\end{center}
	\end{subfigure}
	\caption{LNS (left), RPS (center), and localized RETE net (right)} \label{fig:localization_structures}
\end{figure}


A consistent configuration for a localized RETE net then indeed guarantees query results that are complete according to Definition \ref{def:completeness_subgraphs}:

\begin{theorem} \label{the:completeness_consistent_configuration}
Let $H$ be a graph, $H_p \subseteq H$, $(N, p)$ a well-formed RETE net, and $Q$ the query graph associated with $p$. Furthermore, let $\mathcal{C}^\Phi$ be a consistent configuration for the localized RETE net $(N^\Phi, p^\Phi) = localize((N, p))$. The set of matches given by the stripped result set $\resultsstripped{p^\Phi}{\mathcal{C}^\Phi}$ is then complete under $H_p$.
\end{theorem}

\begin{proof} (Idea)\footnote{See Appendix \ref{app:technical_details} for detailed proofs.}
It can be shown via induction over the height of $N$ that request projection structures ensure the construction of all intermediate results required to guarantee completeness of the overall result under $H_p$.
\end{proof}


Notably, the insertion of request projection structures creates cycles in the localized RETE net, which prevents execution via a reverse topological sorting. However, the marking filter and marking assignment nodes in the request projection structures effectively prevent cyclic execution at the level of intermediate results: Because matches in the result set of a dependency of some join at height $h$ that are only computed on request from the other dependency are marked $h$, these matches are filtered out in the dependent request projection structure. An execution order for the localized RETE net $(N^\Phi, p^\Phi) = localize((N, p))$ can thus be constructed recursively via an $order$ procedure as follows:

If $p$ is an edge input node, the RETE graph given by $N^\Phi = LNS(p)$ is a tree that can be executed via a reverse topological sorting of the nodes in $LNS(p)$, that is, $order((N^\Phi, p^\Phi)) = toposort(LNS(p))^{-1}$.

If $p$ is a join, according to the construction, the localized RETE graph is given by $N^\Phi = N^\Phi_{\bowtie} \cup N^\Phi_l \cup N^\Phi_r \cup RPS_l \cup RPS_r$. In this case, an execution order for $(N^\Phi, p^\Phi)$ can be derived via the concatenation $order((N^\Phi, p^\Phi)) = order(RPS_r) \circ order(N^\Phi_l) \circ order(RPS_l) \circ order(N^\Phi_r) \circ order(RPS_r) \circ order(N^\Phi_l) \circ order(N^\Phi_{\bowtie})$, where $order(RPS_l) = toposort(RPS_l)^{-1}$, $order(RPS_r) = toposort(RPS_r)^{-1}$, and $order(N^\Phi_{\bowtie}) = [p^\Phi]$, that is, the sequence containing only $p^\Phi$.

Executing a localized RETE net $(N^\Phi, p^\Phi)$ via $order((N^\Phi, p^\Phi))$ then guarantees a consistent result configuration for any starting configuration:

\begin{theorem} \label{the:execution_order_consistency}
Let $H$ be a graph, $H_p \subseteq H$, $(N, p)$ a well-formed RETE net, and $\mathcal{C}^\Phi_0$ an arbitrary starting configuration. Executing the marking-sensitive RETE net $(N^\Phi, p^\Phi) = localize((N, p))$ via $O = order((N^\Phi, p^\Phi))$ then yields a consistent configuration $\mathcal{C}^\Phi = execute(O, N^\Phi, H, H_p, \mathcal{C}^\Phi_0)$.
\end{theorem}

\begin{proof} (Idea)
Follows because the inserted marking filter nodes prevent cyclic execution behavior at the level of intermediate results.
\end{proof}

Combined with the result from Theorem \ref{the:completeness_consistent_configuration}, this means that a localized RETE net can be used to compute complete query results for a relevant subgraph in the sense of Definition \ref{def:completeness_subgraphs}, as outlined in the following corollary:

\begin{corollary} \label{cor:completeness_execution}
Let $H$ be a graph, $H_p \subseteq H$, $(N, p)$ a well-formed RETE net, and $Q$ the query graph associated with $p$. Furthermore, let $\mathcal{C}^\Phi_0$ be an arbitrary starting configuration for the localized RETE net $(N^\Phi, p^\Phi) = localize((N, p))$ and $\mathcal{C}^\Phi = execute(order((N^\Phi, p^\Phi)), N^\Phi, H, H_p, \mathcal{C}^\Phi_0)$. The set of matches from $Q$ into $H$ given by $\resultsstripped{p^\Phi}{\mathcal{C}^\Phi}$ is then complete under $H_p$.
\end{corollary}

\begin{proof}
Follows directly from Theorem \ref{the:completeness_consistent_configuration} and Theorem \ref{the:execution_order_consistency}.
\end{proof}

\subsection{Performance of Localized RETE Nets} \label{sec:localized_rete_performance}

Performance of a RETE net $(N, p)$ with respect to both execution time and memory consumption is largely determined by the \emph{effective size} of a consistent configuration $\mathcal{C}$ for $(N, p)$, which we define as $|\mathcal{C}|_e := \sum_{n \in V^N} \sum_{m \in \mathcal{C}(n)} |m|$. In this context, we define the size of a match $m : Q \rightarrow H$ as $|m| := |m^V| + |m^E| = |V^Q| + |E^Q|$. It can then be shown that localization incurs only a constant factor overhead on the effective size of $\mathcal{C}$ for any edge-dominated host graph:

\begin{theorem} \label{the:upper_bound_configuration_size}
Let $H$ be an edge-dominated graph, $H_p \subseteq H$, $(N, p)$ a well-formed RETE net with $Q$ the associated query graph of $p$, $\mathcal{C}$ a consistent configuration for $(N, p)$ for host graph $H$, and $\mathcal{C}^\Phi$ a consistent configuration for the localized RETE net $(N^\Phi, p^\Phi) = localize((N, p))$ for host graph $H$ and relevant subgraph $H_p$. It then holds that $\sum_{n^\Phi \in V^{N^\Phi}} \sum_{(m, \phi) \in \mathcal{C}(n^\Phi)} |m| \leq 7 \cdot |\mathcal{C}|_e$.
\end{theorem}

\begin{proof} (Idea)
Follows because localization only increases the number of RETE nodes by factor 7 and marking-sensitive result sets contain no duplicate matches.
\end{proof}

By Theorem \ref{the:upper_bound_configuration_size}, it then follows that localization of a RETE net incurs only a constant factor overhead on memory consumption even in the worst case where the relevant subgraph is equal to the full model:

\begin{corollary}
Let $H$ be an edge-dominated graph, $H_p \subseteq H$, $(N, p)$ a well-formed RETE net, $\mathcal{C}$ a consistent configuration for $(N, p)$ for host graph $H$, and $\mathcal{C}^\Phi$ a consistent configuration for the localized RETE net $(N^\Phi, p^\Phi) = localize((N, p))$ for host graph $H$ and relevant subgraph $H_p$. Assuming that storing a match $m$ requires an amount of memory in $O(|m|)$ and storing an element from $\overline{\mathbb{N}}$ requires an amount of memory in $O(1)$, storing $\mathcal{C}^\Phi$ requires memory in $O(|\mathcal{C}|_e)$.
\end{corollary}

\begin{proof} (Idea)
Follows from Theorem \ref{the:upper_bound_configuration_size}.
\end{proof}

In the worst case, a localized RETE net would still be required to compute a complete result. In this scenario, the execution of the localized RETE net may essentially require superfluous recomputation of match markings, causing computational overhead. When starting with an empty configuration, the number of such recomputations per match is however limited by the size of the query graph, only resulting in a small increase in computational complexity:

\begin{theorem} \label{the:complexity_time_localized}
Let $H$ be an edge-dominated graph, $H_p \subseteq H$, $(N, p)$ a well-formed RETE net for query graph $Q$, $\mathcal{C}$ a consistent configuration for $(N, p)$, and $\mathcal{C}^\Phi_0$ the empty configuration for $(N^\Phi, p^\Phi) = localize((N, p))$. Executing $(N^\Phi, p^\Phi)$ via $execute(order((N^\Phi, p^\Phi)), N^\Phi, H, H_p, \mathcal{C}^\Phi_0)$ then takes $O(T \cdot (|Q_a| + |Q|))$ steps, with $|Q_a|$ the average size of matches in $\mathcal{C}$ and $T = \sum_{n \in V^N} |\mathcal{C}(n)|$.
\end{theorem}

\begin{proof} (Idea)
Follows since the effort for initial construction of matches by the marking-sensitive RETE net is linear in the total size of the constructed matches and the marking of a match changes at most $|Q|$ times.
\end{proof}

Assuming an empty starting configuration, a regular well-formed RETE net $(N, p)$ can be executed in $O(|\mathcal{C}|_e)$ steps, which can also be expressed as $O(T \cdot |Q_a|)$. The overhead of a localized RETE net compared to the original net can thus be characterized by the factor $\frac{|Q|}{|Q_a|}$. Assuming that matches for the larger query subgraphs constitute the bulk of intermediate results, which seems reasonable in many scenarios, $\frac{|Q|}{|Q_a|}$ may be approximated by a constant factor.

For non-empty starting configurations and incremental changes, no sensible guarantees can be made. On the one hand, in a localized RETE net, a host graph modification may trigger the computation of a large number of intermediate results that were previously omitted due to localization. On the other hand, in a standard RETE net, a modification may result in substantial effort for constructing superfluous intermediate results that can be avoided by localization. Depending on the exact host graph structure and starting configuration, execution may thus essentially require a full recomputation for either the localized or standard RETE net but cause almost no effort for the other variant.


\section{Evaluation} \label{sec:evaluation}

We aim to investigate whether RETE net localization can improve performance of query execution in scenarios where the relevant subgraph constitutes only a fraction of the full model, considering initial query execution time, execution time for incrementally processing model updates, and memory consumption as performance indicators.\footnote{Experiments were executed on a Linux SMP Debian 4.19.67-2 computer with Intel Xeon E5-2630 CPU (2.3\,GHz clock rate) and 386\,GB system memory running OpenJDK 11.0.6. Reported measurements correspond to the arithmetic mean of measurements for 10 runs. Memory measurements were obtained using the Java Runtime class. To represent graph data, all experiments use EMF \cite{emf} object graphs that enable reverse navigation of edges. Our implementation is available under \cite{implementation}. More details and query visualizations can be found in Appendix \ref{app:queries} and \ref{app:additional_measurements}.} We experiment with the following querying techniques:

\begin{itemize}
\item \textbf{STANDARD}: Our own implementation of a regular RETE net with global execution semantics \cite{barkowsky2023host}.
\item \textbf{LOCALIZED}: Our own implementation of the RETE net used for STANDARD, localized according to the description in Section \ref{sec:localized_search_with_marking-sensitive_rete}.
\item \textbf{VIATRA}: The external RETE-based VIATRA tool \cite{varro2016}.
\item \textbf{SDM*}: Our own local-search-based Story Diagram Interpreter tool \cite{giese2009improved}. Note that we only consider searching for new query results. We thus underapproximate the time and memory required for a full solution with this strategy, which would also require maintaining previously found results.
\end{itemize}

\subsection{Synthetic Abstract Syntax Graphs} \label{sec:evaluation_java_scalability}

We first attempt a systematic evaluation via a synthetic experiment, which emulates a developer loading part of a large model into their workspace and monitoring some well-formedness constraints as they modify the loaded part, that is, relevant subgraph, without simultaneous changes to other model parts.

We therefore generate Java abstract syntax graphs with 1, 10, 100, 1000, and 10000 packages, with each package containing 10 classes with 10 fields referencing other classes. As relevant subgraph, we consider a single package and its contents. We then execute a query searching for paths consisting of a package and four classes connected via fields. After initial query execution, we modify the relevant subgraph by creating a class with 10 fields in the considered package and perform an incremental update of query results. This step is repeated 10 times.

Figure \ref{fig:java_scalability_bars} (right) displays the execution times for the initial execution of the query. The execution time of LOCALIZED remains around 120\,ms for all model sizes. The execution time for SDM* slowly grows from around 350\,ms to 1025\,ms due to indexing effort that is necessary for observing model changes. In contrast, the execution time for the other RETE-based strategies clearly scales with model size, with the execution time for STANDARD growing from around 13\,ms for the smallest model to more than 184\,000\,ms for the largest model. On the one hand, localization thus incurs a noticeable overhead in initial execution time for the smallest model, where even localized query execution is essentially global. On the other hand, it significantly improves execution time for the larger models and even achieves better scalability than the local-search-based tool in this scenario.

The average times for processing a model update are displayed in Figure \ref{fig:java_scalability_bars} (center). Here, all strategies achieve execution times mostly independent of model size. While the measurements for STANDARD fluctuate, likely due to the slightly unpredictable behavior of hash-based indexing structures, average execution times remain low overall and below 10\,ms for LOCALIZED. Still, localization incurs a noticeable overhead up to factor 6 compared to STANDARD and VIATRA. This overhead is expected, since in this scenario, all considered updates affect the relevant subgraph and thus impact the results of the localized RETE net similarly to the results of the standard RETE nets. Consequently, localization does not reduce computational effort, but causes overhead instead.

Finally, Figure \ref{fig:java_scalability_bars} (left) shows the memory measurements for all strategies and models after the final update. Here, LOCALIZED again achieves a substantial improvement in scalability compared to the other RETE-based strategies, with a slightly higher memory consumption for the smallest model and an improvement by factor 120 over STANDARD for the largest model. This is a result of the localized RETE net producing the same number of intermediate results for all model sizes, with the slight growth in memory consumption likely a product of the growing size of the model itself. SDM*, not storing any matches, performs better for all but the largest model, where memory consumption surpasses the measurement for LOCALIZED. This surprising result can probably be explained by the fact that SDM* has to index the full model to observe modifications, causing an overhead in memory consumption.

\begin{figure}
\centering
\includegraphics[width=\textwidth]{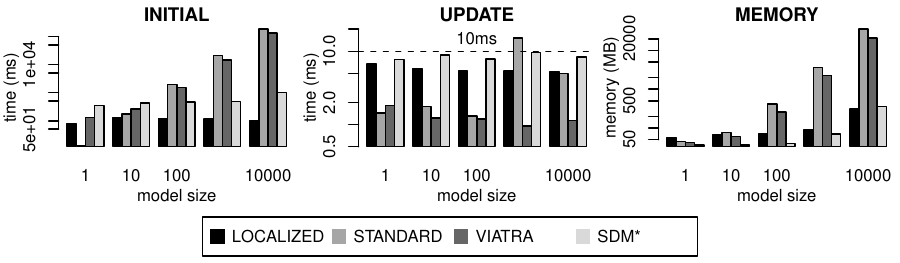}
\caption{Measurements for the synthetic abstract syntax graph scenario (log scale)}\label{fig:java_scalability_bars}
\end{figure}

In addition to these experiments, where the full model is always stored in main memory, we also experimented with a model initially stored on disk via the persistence layer CDO \cite{cdo}. Measurements mostly mirror those for the main-memory-based experiment. Notably though, in conjunction with CDO the LOCALIZED strategy achieves almost ideal scalability regarding memory consumption for this scenario, with measurements around 70\,MB for all model sizes.\footnote{See Appendix \ref{app:additional_measurements} for measurement results.}

\subsection{Real Abstract Syntax Graphs} \label{sec:evaluation_java_incremental}

To evaluate our approach in the context of a more realistic application scenario, we perform a similar experiment using real data and queries inspired by object-oriented design patterns \cite{gamma1993design}. In contrast to the synthetic scenario, this experiment emulates a situation where modifications may concern not only the relevant subgraph but the entire model, for instance when multiple developers are simulatenously working on different model parts.

We therefore extract a history of real Java abstract syntax graphs with about 16\,000 vertices and 45\,000 edges from a software repository using the MoDisco tool \cite{bruneliere2010modisco}. After executing the queries over the initial commit, we replay the history and perform incremental updates of query results after each commit. As relevant subgraph, we again consider a single package and its contents.

Figure \ref{fig:java_incremental_time} displays the aggregate execution time for processing the commits one after another for the queries where LOCALIZED performed best and worst compared to STANDARD, with the measurement at $x = 0$ indicating the initial execution time for the starting model. Initial execution times are similarly low due to a small starting model and in fact slightly higher for LOCALIZE. However, on aggregate LOCALIZED outperforms STANDARD with an improvement between factor 5 and 18 due to significantly lower update times, which are summarized in Figure \ref{fig:update_clouds} (left). In this case, the improvement mostly stems from the more precise monitoring of the model for modifications: The RETE nets of both STANDARD and LOCALIZED remain small due to strong filtering effects in the query graphs. However, while STANDARD spends significant effort on processing model change notifications due to observing all appropriately typed model elements, this effort is substantially reduced for LOCALIZED, which only monitors elements relevant to query results required for completeness under the relevant subgraph. The execution times of SDM* can be explained by the same effect. Interestingly, VIATRA seems to implement an improved handling of such notifications, achieving improved execution times for particularly small updates even compared to localize, but requiring more time if an update triggers changes to the RETE net. Combined with a higher RETE net initialization time, this results in LOCALIZED also outperforming VIATRA for all considered queries.

Regarding memory consumption, all strategies perform very similarly, which is mostly a result of the size of the model itself dominating the measurement and hiding the memory impact of the rather small RETE nets.

\begin{figure}
\centering
\includegraphics[width=\textwidth]{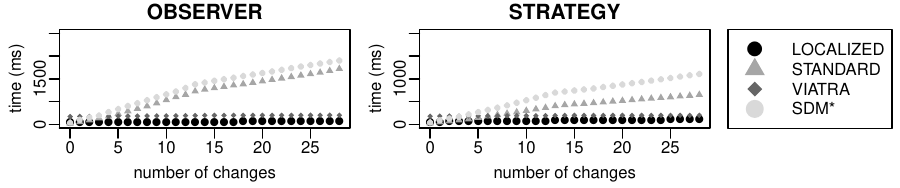}
\caption{Execution times for the real abstract syntax graph scenario}\label{fig:java_incremental_time}
\end{figure}

\begin{figure}
\centering
\includegraphics[width=\textwidth]{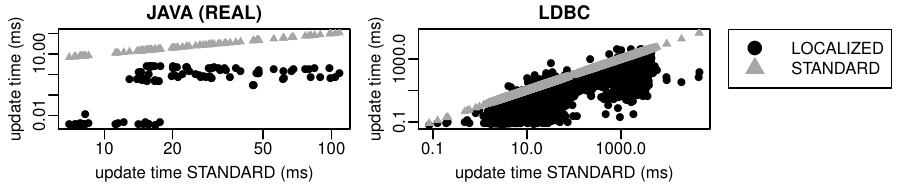}
\caption{Summary of update processing times for different scenarios (log scale)}\label{fig:update_clouds}
\end{figure}

\subsection{LDBC Social Network Benchmark} \label{sec:evaluation_ldbc}

Finally, we also experiment with the independent LDBC Social Network Benchmark \cite{erling2015ldbc,angles2024ldbc}, simulating a case where a user of a social network wants to incrementally track query results relating to them personally.

We therefore generate a synthetic social network consisting of around 850\,000 vertices, including about 1700 persons, and 5\,500\,000 edges using the benchmark's data generator and the predefined scale factor 0.1. We subsequently transform this dataset into a sequence of element creations and deletions based on the timestamps included in the data. We then create a starting graph by replaying the first half of the sequence and perform an initial execution of adapted versions of the benchmark queries consisting of plain graph patterns, with a person with a close-to-average number of contacts in the final social network designated as relevant subgraph. After the initial query execution, we replay the remaining changes, incrementally updating the query results after each change.

The resulting execution times for the benchmark queries where LOCALIZED performed best and worst compared to STANDARD are displayed in Figure \ref{fig:ldbc_time}. A summary of all update time measurements for LOCALIZED in comparison with STANDARD is also displayed in Figure \ref{fig:update_clouds} (right). For all queries, LOCALIZED ultimately outperforms the other approaches by a substantial margin, as the localized RETE version forgoes the computation of a large number of irrelevant intermediate results due to the small relevant subgraph on the one hand and avoids redundant computations on the other hand.

\begin{figure}
\centering
\includegraphics[width=\textwidth]{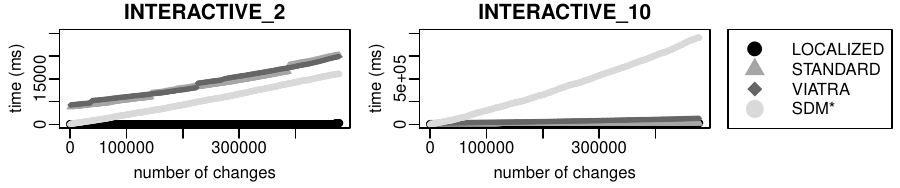}
\caption{Execution times for the LDBC scenario}\label{fig:ldbc_time}
\end{figure}

\begin{figure}
\centering
\includegraphics[width=\textwidth]{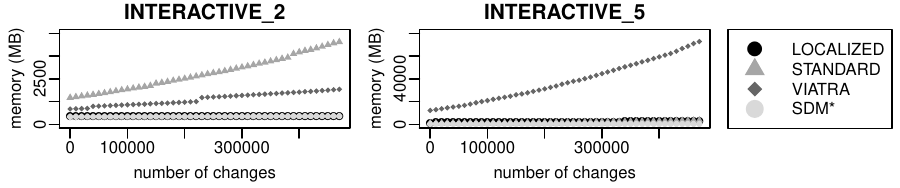}
\caption{Memory measurements for the LDBC scenario}\label{fig:ldbc_memory}
\end{figure}

The memory measurements in Figure \ref{fig:ldbc_memory} mostly mirror execution times for RETE-based approaches, with the memory consumption for LOCALIZED always lower than for STANDARD and VIATRA except for a period at the beginning of the execution of the query INTERACTIVE\_5, where STANDARD outperforms LOCALIZED. The weaker performance for INTERACTIVE\_5 and INTERACTIVE\_10 likely stems from the fact that these queries contain cycles, which act as strong filters for subsequent (intermediate) results and achieve a somewhat similar effect as localization. The weaker performance of VIATRA for INTERACTIVE\_5 is a product of the usage of a suboptimal RETE net. As expected, memory consumption is lowest for SDM* for all queries.

\subsection{Discussion}

On the one hand, our experimental results demonstrate that in situations where the relevant subgraph constitutes only a fraction of the full model, RETE net localization can improve the performance of incremental query execution compared to both the standard RETE approach and a solution based on local search. In such scenarios, localization can improve scalability with respect to initial query execution time and memory consumption, as demonstrated in Sections \ref{sec:evaluation_java_scalability} and \ref{sec:evaluation_ldbc} and, if changes are not restricted to the relevant subgraph, also update processing time, as shown in Sections \ref{sec:evaluation_java_incremental} and \ref{sec:evaluation_ldbc}.

On the other hand, as demonstrated in Section \ref{sec:evaluation_java_scalability}, localization incurs an overhead on update processing time if changes are only made to the relevant subgraph and on initial execution time and memory consumption if the relevant subgraph contains most of the elements in the full model. While this overhead will essentially be limited to a constant factor in many scenarios, as analyzed in Section \ref{sec:localized_rete_performance}, the standard RETE approach remains preferable for query execution with global semantics or if modifications are restricted to the relevant subgraph and initial query execution time and memory consumption are irrelevant.

To mitigate internal threats to the validity of our results resulting from unexpected JVM behavior, we have performed 10 runs of all experiments. However, with reliable memory measurements a known pain point of Java-based experiments, the reported memory consumption values are still not necessarily accurate and can only serve as an indicator. To minimize the impact of the implementation on conceptual observations, we compare the prototypical implementation of our approach to a regular RETE implementation \cite{barkowsky2023host}, which shares a large portion of the involved code, and to two existing tools \cite{varro2016,giese2009improved}.

We have attempted to address external threats to validity via experiments accounting for different application domains and a combination of synthetic and real-world queries and data, including a setting from an established, independent benchmark. Still, our results cannot be generalized and do not support quantitative claims, but serve to demonstrate conceptual advantages and disadvantages of the presented approach.


\section{Related Work} \label{sec:related_work}

With graph query execution forming the foundation of many applications, there already exists an extensive body of research regarding its optimization.

Techniques based on local search \cite{cordella2004sub,geiss2006grgen,giese2009improved,arendt2010henshin,han2013turboiso,bi2016efficient,juttner2018vf2++} constitute one family of graph querying approaches. While they are designed to exploit locality in the host graph to improve execution time, repeated query execution leads to redundant computations that are only avoided by fully incremental techniques.

In \cite{egyed2006instant}, Egyed proposes a scoping mechanism for local search to support incremental query execution, only recomputing query results when a graph element touched during query execution changes. While this approach offers some degree of incrementality, it is limited to queries with designated root elements that serve as an anchor in the host graph and may still result in redundant computations, since query reevaluation is only controlled at the granularity of root elements.

A second family of solutions is based on discrimination networks \cite{hanson2002trigger,varro2013rete,varro2016,beyhl2018framework}, the most prominent example of which are RETE nets.

VIATRA \cite{varro2016} is a mature tool for incremental graph query execution based on the RETE algorithm \cite{forgy1989rete}, which supports advanced concepts for query specification and optimization not considered in this paper. Notably, VIATRA allows reuse of matches for isomorphic query subgraphs within a single RETE net. This is achieved via RETE structures not covered by the rather restrictive definition of well-formedness used in this paper, which points to an interesting direction for future work. However, while VIATRA also has a local search option for query execution, it does not integrate local search with the incremental query engine but rather offers it as an alternative.

Beyhl \cite{beyhl2018framework} presents an incremental querying technique based on a generalized version of RETE nets, called Generalized Discrimination Networks (GDNs) \cite{hanson2002trigger}. The main difference compared to the RETE algorithm is the lack of join nodes. Instead, more complex nodes that directly compute complex matches using local search are employed. The approach however represents more of a means of controlling the trade-off between local search and RETE rather than an integration and still requires a global computation of matches for the entire host graph.

In previous work \cite{barkowsky2021improving}, we have made a first step in the direction of localizing RETE-based query execution. While this earlier technique already allowed anchoring the execution of a RETE net to certain host graph vertices, this anchoring was based on typing information and its results did not meet the definition of completeness introduced in this paper.

Model repositories such as CDO \cite{cdo} and NeoEMF \cite{daniel2017neoemf} provide support for query execution over partial models via lazy loading. As persistence layers, these solutions however focus on implementing an interface of atomic model access operations in order to be agnostic regarding the employed query mechanism.

The Mogwa\"{\i} tool \cite{daniel2018scalable} aims to improve query execution over persistence layers like CDO and NeoEMF by mapping model queries to native queries for the underlying database system instead of using the atomic model access operations provided by the layer's API, avoiding loading the entire model into main memory. The tool however does not consider incremental query execution.

Jahanbin et al. propose an approach for querying partially loaded models stored via persistence layers \cite{jahanbin2022partial} or as XMI files \cite{jahanbin2023towards}. In contrast to the solution presented in this paper, their approach still aims to always provide complete query results for the full model and is thus based on static analysis and typing information rather than dynamic exploitation of locality.

Query optimization for relational databases is a research topic that has been under intense study for decades \cite{krishnamurthy1986optimization,lee2001optimizing,leis2015good}. Generally, many of the techniques from this field are applicable to RETE nets, which are ultimately based on relational algebra and related to materialized views in relational databases \cite{gupta1995maintenance}. However, relational databases lack the notion of locality inherent to graph-based encodings and are hence not tailored to exploit local navigation.

This shortcoming has given rise to a number of graph databases \cite{angles2012comparison}, which employ a graph-based data representation instead of a relational encoding and form the basis of some model persistance layers like NeoEMF \cite{daniel2017neoemf}. While these database systems are designed to accommodate local navigation, to the best of our knowledge, support for incremental query execution is still lacking.


\section{Conclusion} \label{sec:conclusion}

In this paper, we have presented a relaxed notion of completeness for query results that lifts the requirement of strict completeness of results for graph queries and thereby the need for necessarily global query execution. Based on this relaxed notion of completeness, we have developed an extension of the RETE approach that allows local, yet fully incremental execution of graph queries. An initial evaluation demonstrates that the approach can improve scalability in scenarios with small relevant subgraphs, but causes an overhead in unfavorable cases.

In future work, we plan to extend our technique to accommodate advanced concepts for query specification, most importantly nested graph conditions. Furthermore, we want to investigate whether the proposed solution can be adapted to support bulk loading of partial models in order to reduce overhead caused by lazy loading strategies employed by model persistence layers.

\bibliographystyle{llncs_template/splncs04}
\bibliography{localized_rete_for_incremental_graph_queries}

\clearpage

\appendix


\section{Localized RETE Example Execution} \label{app:example_execution}

Abstracting from the internal structure of local navigation structures and request projection structures, the $order$ procedure would produce the execution order $RPS^c_r$, $LNS^{p \rightarrow c}$, $RPS^c_l$, $LNS^{c \rightarrow f}$, $RPS^c_r$, $LNS^{p \rightarrow c}$, $[\bowtie]^\Phi$ for the example localized RETE net in Figure \ref{fig:localization_structures} (right). An execution of this sequence, which is visualized in Figure \ref{fig:example_execution}, would then proceed as follows:

\begin{enumerate}
\renewcommand{\theenumi}{\alph{enumi}}
\item The starting configuration for the RETE net may be inconsistent for all of the net's nodes, for instance due to being empty (if this is the net's initial execution), because of unprocessed host graph changes (if this is a repeated execution after incremental host graph changes), or because the provided starting configuration is simply erroneous.
\item The execution of $RPS^c_r$ in the first step (indicated by a bold border and label in Figure \ref{fig:example_execution_1}), which is only necessary to handle an erroneous starting configuration, then guarantees a consistent configuration for all nodes in $RPS^c_r$ by effectively correcting any errors. This is indicated by the label being underlined in Figure \ref{fig:example_execution_2}.
\item After the execution of $LNS^{p \rightarrow c}$ in the second step, the resulting configuration is also consistent for all nodes in $LNS^{p \rightarrow c}$, since $LNS^{p \rightarrow c}$ is a tree and can thus simply be executed via a reverse topological sorting. Essentially, this step either initially populates the empty current result sets of the nodes in $LNS^{p \rightarrow c}$ in case of an empty starting configuration or handles any relevant changes to host graph or relevant subgraph. In this context, a relevant change can either be the creation or deletion of an edge of type $ce$ or the inclusion or exclusion of a vertex to or from the relevant subgraph. Ultimately, the current result set of the union node at the top of $LNS^{p \rightarrow c}$ then contains at least all edges of type $ce$ that are adjacent to a vertex in the relevant subgraph, with these matches marked $\infty$.
\item The execution of $RPS^c_l$ in the third step ensures consistency for all nodes in $RPS^c_l$. Essentially, each host graph vertex to which the query graph vertex $c$ is mapped in at least one match with marking $\phi > 1$ in the root of $LNS^{p \rightarrow c}$ is added to the result set of the root of $RPS^c_l$ and assigned a marking of $1$ (the height of the original, well-formed RETE net in Figure \ref{fig:example_query_and_rete} (right)).
\item Executing $LNS^{c \rightarrow f}$ in the fourth step now not only guarantees that all edges of type $fe$ that touch the relevant subgraph are included in the result set of the root of $LNS^{c \rightarrow f}$ with a marking of $\infty$. In addition, because of the dependency of $LNS^{c \rightarrow f}$ on $RPS^c_l$, the execution also ensures the same for $fe$ edges adjacent to a host graph vertex in the image of a match for the root of $RPS^c_l$ and thus $LNS^{p \rightarrow c}$, if the marking of that match is greater than $1$. These matches are however only marked $1$, signalling that they were only computed on request from $LNS^{p \rightarrow c}$. Intuitively, this guarantees that the current result set for the root of $LNS^{c \rightarrow f}$ also contains all matches that are required to complement matches in $LNS^{p \rightarrow c}$ that involve a host graph vertex in the relevant subgraph. Because $RPS^c_r$ depends on $LNS^{c \rightarrow f}$, $RPS^c_r$ may now be inconsistent again.
\item The execution of $RPS^c_r$ restores consistency for $RPS^c_r$ by adding all matches for $LNS^{c \rightarrow f}$ that are marked $\infty$, that is, involve a vertex from the relevant subgraph, to the root node of $RPS^c_r$. These matches are marked $1$. Notably, due to the marking filter node in $RPS^c_r$, matches that are only marked $1$ in the current result set of $LNS^{c \rightarrow f}$ are not considered.
\item Executing $LNS^{p \rightarrow c}$ now also adds all edges of type $ce$ required to complement matches for $LNS^{c \rightarrow f}$ that touch the relevant subgraph to the result set of $LNS^{p \rightarrow c}$. Again, these matches are only marked $1$. The dependency of $RPS^c_l$ on $LNS^{p \rightarrow c}$ could now mean that $RPS^c_l$ may be inconsistent again. However, the last execution of $LNS^{p \rightarrow c}$ can only have created (or potentially deleted) matches with marking $1$. The reason for this is that the only changes that were processed during the last execution of $LNS^{p \rightarrow c}$ were changes to the result set of $RPS^c_r$, which only contains matches marked $1$ due to the marking assignment node in the request projection structure. Due to the marking filter node at the bottom of $RPS^c_l$, these changes do not affect the consistency of any node in $RPS^c_l$, meaning that $RPS^c_l$ is still consistent despite the dependency.
\item Finally, the execution of $[\bowtie]^\Phi$ ensures the consistency of the last remaining RETE node. According to Corollary \ref{cor:completeness_execution}, the result set of $[\bowtie]^\Phi$ is now complete under the relevant subgraph. Intuitively, this can be explained by the fact that the result set of $LNS^{p \rightarrow c}$ contains all edges of type $ce$ that touch the relevant subgraph, $LNS^{c \rightarrow f}$ contains all matches required to complement these edges, and vice-versa.
\end{enumerate}

\clearpage

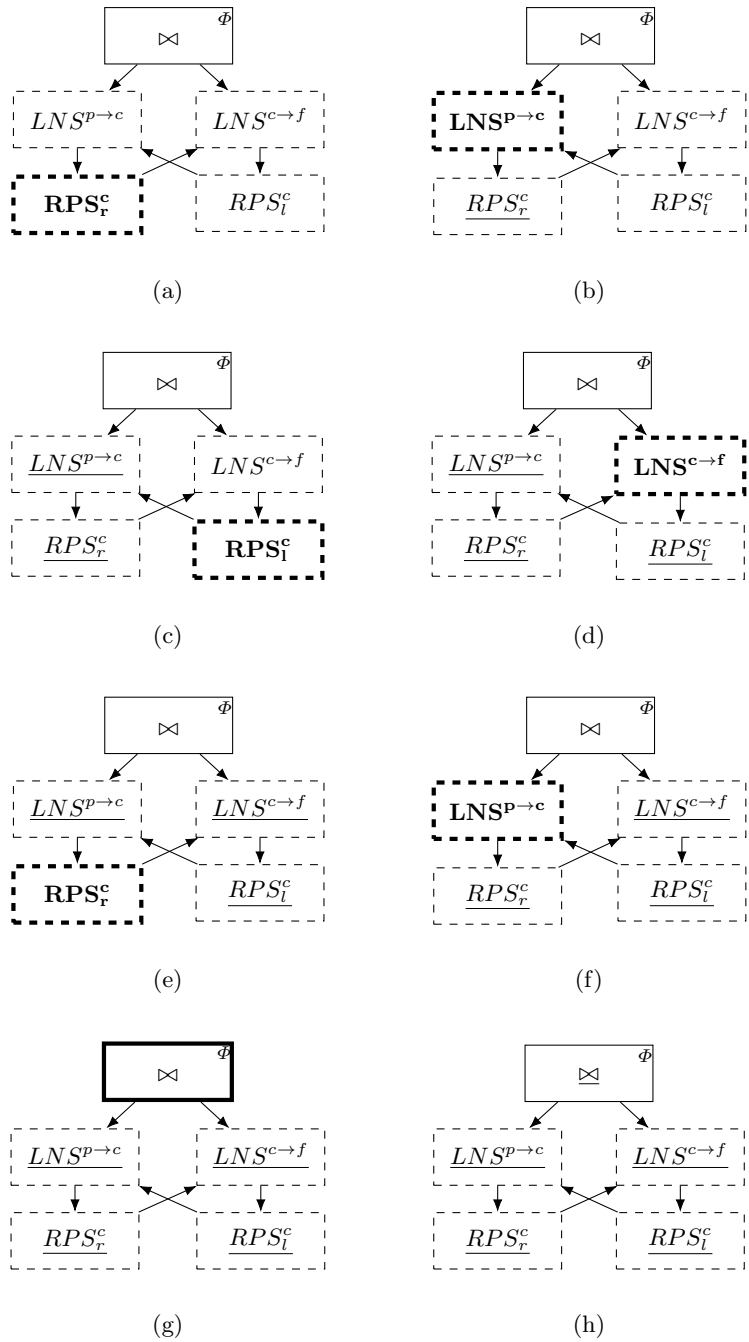
\begin{figure} [!ht]
	\centering
	\begin{subfigure}[b]{0.45\textwidth}
		\begin{center}

\begin{tikzpicture}

\node (join) [ms_retenode_app] {$\bowtie$};

\node[below left = 0.35cm and -0.5cm of join, dashed] (ce) [retenode_app] {$LNS^{p \rightarrow c}$};

\node[below right = 0.35cm and -0.5cm of join, dashed] (fe) [retenode_app] {$LNS^{c \rightarrow f}$};

\node[below = 0.35cm of ce, dashed, ultra thick] (rpsp) [retenode_app] {$\mathbf{RPS^c_r}$};

\node[below = 0.35cm of fe, dashed] (rpsf) [retenode_app] {$RPS^c_l$};

\draw [-{Latex}] (join) -- (ce);

\draw [-{Latex}] (join) -- (fe);

\draw [-{Latex}] (ce) -- (rpsp);

\draw [-{Latex}] (fe) -- (rpsf);

\draw [-{Latex}] (rpsp) -- (fe);

\draw [-{Latex}] (rpsf) -- (ce);

\end{tikzpicture}
		\end{center}
		\caption{}\label{fig:example_execution_1}
	\end{subfigure}
	\begin{subfigure}[b]{0.45\textwidth}
		\begin{center}

\begin{tikzpicture}

\node (join) [ms_retenode_app] {$\bowtie$};

\node[below left = 0.35cm and -0.5cm of join, dashed, ultra thick] (ce) [retenode_app] {$\mathbf{LNS^{p \rightarrow c}}$};

\node[below right = 0.35cm and -0.5cm of join, dashed] (fe) [retenode_app] {$LNS^{c \rightarrow f}$};

\node[below = 0.35cm of ce, dashed] (rpsp) [retenode_app] {\underline{$RPS^c_r$}};

\node[below = 0.35cm of fe, dashed] (rpsf) [retenode_app] {$RPS^c_l$};

\draw [-{Latex}] (join) -- (ce);

\draw [-{Latex}] (join) -- (fe);

\draw [-{Latex}] (ce) -- (rpsp);

\draw [-{Latex}] (fe) -- (rpsf);

\draw [-{Latex}] (rpsp) -- (fe);

\draw [-{Latex}] (rpsf) -- (ce);

\end{tikzpicture}
		\end{center}
		\caption{}\label{fig:example_execution_2}
	\end{subfigure}
	
	\vspace{15pt}
	
	\begin{subfigure}[b]{0.45\textwidth}
		\begin{center}

\begin{tikzpicture}

\node (join) [ms_retenode_app] {$\bowtie$};

\node[below left = 0.35cm and -0.5cm of join, dashed] (ce) [retenode_app] {\underline{$LNS^{p \rightarrow c}$}};

\node[below right = 0.35cm and -0.5cm of join, dashed] (fe) [retenode_app] {$LNS^{c \rightarrow f}$};

\node[below = 0.35cm of ce, dashed] (rpsp) [retenode_app] {\underline{$RPS^c_r$}};

\node[below = 0.35cm of fe, dashed, ultra thick] (rpsf) [retenode_app] {$\mathbf{RPS^c_l}$};

\draw [-{Latex}] (join) -- (ce);

\draw [-{Latex}] (join) -- (fe);

\draw [-{Latex}] (ce) -- (rpsp);

\draw [-{Latex}] (fe) -- (rpsf);

\draw [-{Latex}] (rpsp) -- (fe);

\draw [-{Latex}] (rpsf) -- (ce);

\end{tikzpicture}
		\end{center}
		\caption{}\label{fig:example_execution_3}
	\end{subfigure}
	\begin{subfigure}[b]{0.45\textwidth}
		\begin{center}

\begin{tikzpicture}

\node (join) [ms_retenode_app] {$\bowtie$};

\node[below left = 0.35cm and -0.5cm of join, dashed] (ce) [retenode_app] {\underline{$LNS^{p \rightarrow c}$}};

\node[below right = 0.35cm and -0.5cm of join, dashed, ultra thick] (fe) [retenode_app] {$\mathbf{LNS^{c \rightarrow f}}$};

\node[below = 0.35cm of ce, dashed] (rpsp) [retenode_app] {\underline{$RPS^c_r$}};

\node[below = 0.35cm of fe, dashed] (rpsf) [retenode_app] {\underline{$RPS^c_l$}};

\draw [-{Latex}] (join) -- (ce);

\draw [-{Latex}] (join) -- (fe);

\draw [-{Latex}] (ce) -- (rpsp);

\draw [-{Latex}] (fe) -- (rpsf);

\draw [-{Latex}] (rpsp) -- (fe);

\draw [-{Latex}] (rpsf) -- (ce);

\end{tikzpicture}
		\end{center}
		\caption{}\label{fig:example_execution_4}
	\end{subfigure}
	
	\vspace{15pt}
	
	\begin{subfigure}[b]{0.45\textwidth}
		\begin{center}

\begin{tikzpicture}

\node (join) [ms_retenode_app] {$\bowtie$};

\node[below left = 0.35cm and -0.5cm of join, dashed] (ce) [retenode_app] {\underline{$LNS^{p \rightarrow c}$}};

\node[below right = 0.35cm and -0.5cm of join, dashed] (fe) [retenode_app] {\underline{$LNS^{c \rightarrow f}$}};

\node[below = 0.35cm of ce, dashed, ultra thick] (rpsp) [retenode_app] {$\mathbf{RPS^c_r}$};

\node[below = 0.35cm of fe, dashed] (rpsf) [retenode_app] {\underline{$RPS^c_l$}};

\draw [-{Latex}] (join) -- (ce);

\draw [-{Latex}] (join) -- (fe);

\draw [-{Latex}] (ce) -- (rpsp);

\draw [-{Latex}] (fe) -- (rpsf);

\draw [-{Latex}] (rpsp) -- (fe);

\draw [-{Latex}] (rpsf) -- (ce);

\end{tikzpicture}
		\end{center}
		\caption{}\label{fig:example_execution_5}
	\end{subfigure}
	\begin{subfigure}[b]{0.45\textwidth}
		\begin{center}

\begin{tikzpicture}

\node (join) [ms_retenode_app] {$\bowtie$};

\node[below left = 0.35cm and -0.5cm of join, dashed, ultra thick] (ce) [retenode_app] {$\mathbf{LNS^{p \rightarrow c}}$};

\node[below right = 0.35cm and -0.5cm of join, dashed] (fe) [retenode_app] {\underline{$LNS^{c \rightarrow f}$}};

\node[below = 0.35cm of ce, dashed] (rpsp) [retenode_app] {\underline{$RPS^c_r$}};

\node[below = 0.35cm of fe, dashed] (rpsf) [retenode_app] {\underline{$RPS^c_l$}};

\draw [-{Latex}] (join) -- (ce);

\draw [-{Latex}] (join) -- (fe);

\draw [-{Latex}] (ce) -- (rpsp);

\draw [-{Latex}] (fe) -- (rpsf);

\draw [-{Latex}] (rpsp) -- (fe);

\draw [-{Latex}] (rpsf) -- (ce);

\end{tikzpicture}
		\end{center}
		\caption{}\label{fig:example_execution_6}
	\end{subfigure}
	
	\vspace{15pt}
	
	\begin{subfigure}[b]{0.45\textwidth}
		\begin{center}

\begin{tikzpicture}

\node (join) [ms_retenode_app, ultra thick] {$\mathbf{\bowtie}$};

\node[below left = 0.35cm and -0.5cm of join, dashed] (ce) [retenode_app] {\underline{$LNS^{p \rightarrow c}$}};

\node[below right = 0.35cm and -0.5cm of join, dashed] (fe) [retenode_app] {\underline{$LNS^{c \rightarrow f}$}};

\node[below = 0.35cm of ce, dashed] (rpsp) [retenode_app] {\underline{$RPS^c_r$}};

\node[below = 0.35cm of fe, dashed] (rpsf) [retenode_app] {\underline{$RPS^c_l$}};

\draw [-{Latex}] (join) -- (ce);

\draw [-{Latex}] (join) -- (fe);

\draw [-{Latex}] (ce) -- (rpsp);

\draw [-{Latex}] (fe) -- (rpsf);

\draw [-{Latex}] (rpsp) -- (fe);

\draw [-{Latex}] (rpsf) -- (ce);

\end{tikzpicture}
		\end{center}
		\caption{}\label{fig:example_execution_7}
	\end{subfigure}
	\begin{subfigure}[b]{0.45\textwidth}
		\begin{center}

\begin{tikzpicture}

\node (join) [ms_retenode_app] {\underline{${\bowtie}$}};

\node[below left = 0.35cm and -0.5cm of join, dashed] (ce) [retenode_app] {\underline{$LNS^{p \rightarrow c}$}};

\node[below right = 0.35cm and -0.5cm of join, dashed] (fe) [retenode_app] {\underline{$LNS^{c \rightarrow f}$}};

\node[below = 0.35cm of ce, dashed] (rpsp) [retenode_app] {\underline{$RPS^c_r$}};

\node[below = 0.35cm of fe, dashed] (rpsf) [retenode_app] {\underline{$RPS^c_l$}};

\draw [-{Latex}] (join) -- (ce);

\draw [-{Latex}] (join) -- (fe);

\draw [-{Latex}] (ce) -- (rpsp);

\draw [-{Latex}] (fe) -- (rpsf);

\draw [-{Latex}] (rpsp) -- (fe);

\draw [-{Latex}] (rpsf) -- (ce);

\end{tikzpicture}
		\end{center}
		\caption{}\label{fig:example_execution_8}
	\end{subfigure}
	\caption{Example execution of the RETE net from Figure \ref{fig:localization_structures} (right)} \label{fig:example_execution}
\end{figure}

\clearpage

\clearpage


\section{Additional Measurements} \label{app:additional_measurements}

\subsection{Synthetic Abstract Syntax Graphs (CDO)}

\begin{figure}[!h]
\centering
\includegraphics[width=\textwidth]{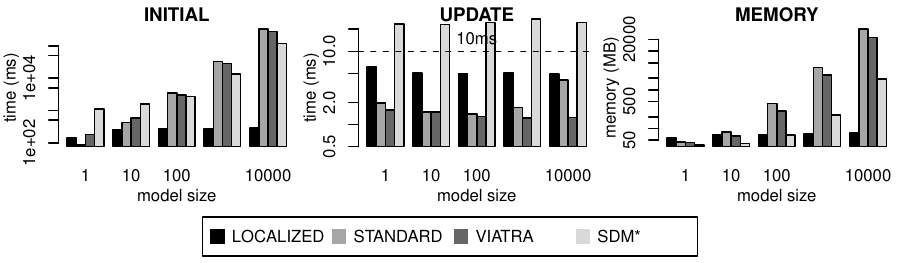}
\caption{Measurements for the synthetic abstract syntax graph scenario, with models stored via CDO (log scale)}
\end{figure}

\clearpage

\subsection{Real Abstract Syntax Graphs}

\begin{figure}[!ht]
\centering
\includegraphics[width=0.9\textwidth]{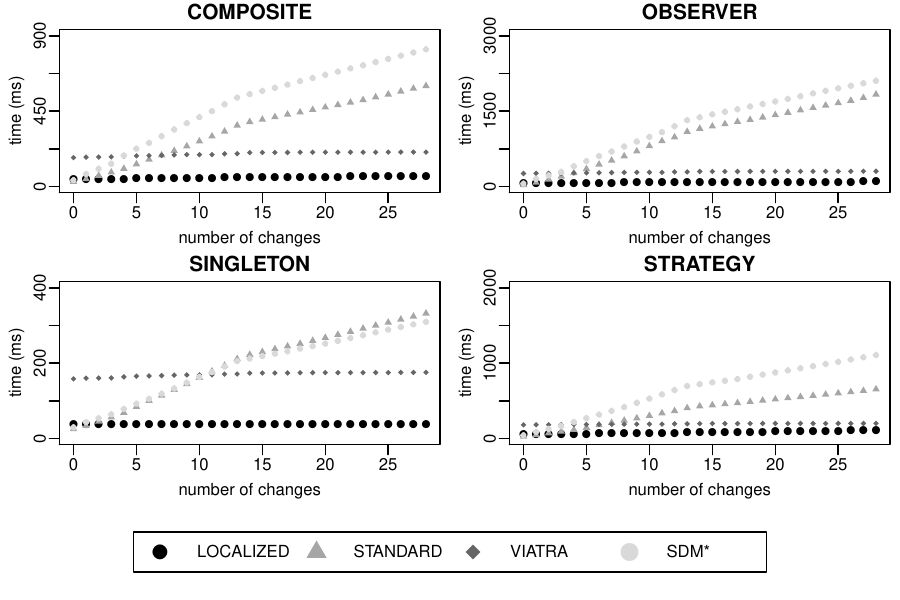}
\caption{Execution times for the real abstract syntax graph scenario}
\end{figure}

\begin{figure}[!ht]
\centering
\includegraphics[width=0.9\textwidth]{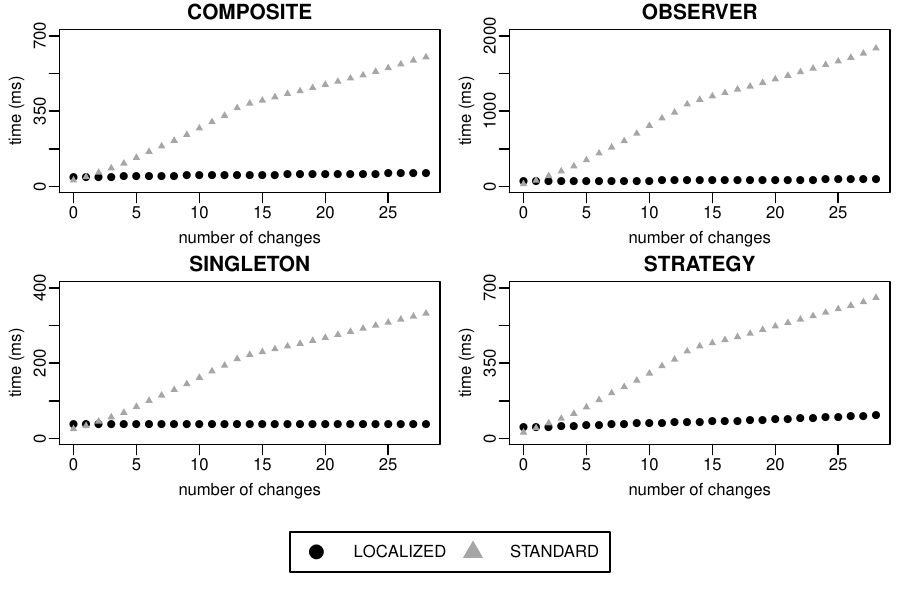}
\caption{Execution times for the real abstract syntax graph scenario}
\end{figure}

\begin{figure}[!ht]
\centering
\includegraphics[width=0.9\textwidth]{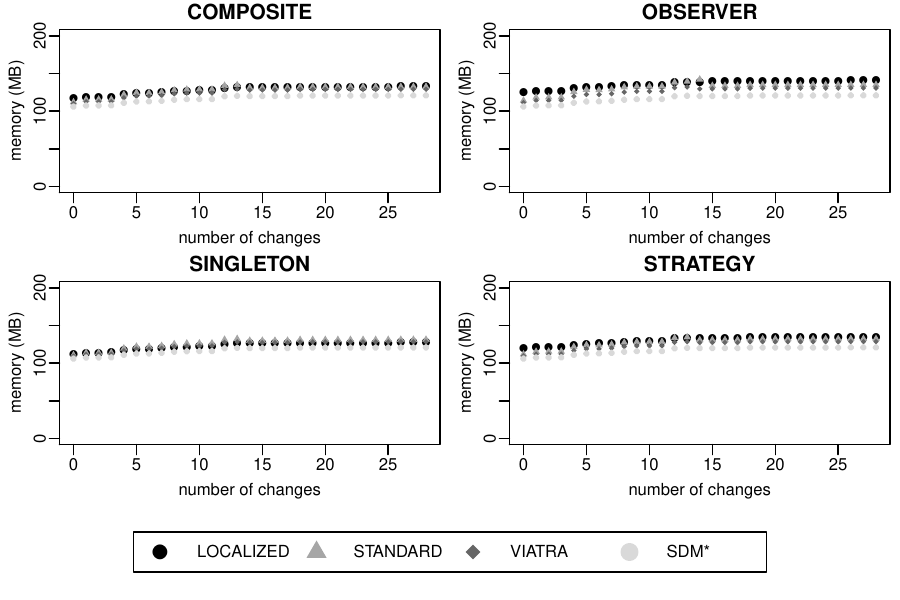}
\caption{Memory measurements for the real abstract syntax graph scenario}
\end{figure}

\begin{figure}[!ht]
\centering
\includegraphics[width=0.9\textwidth]{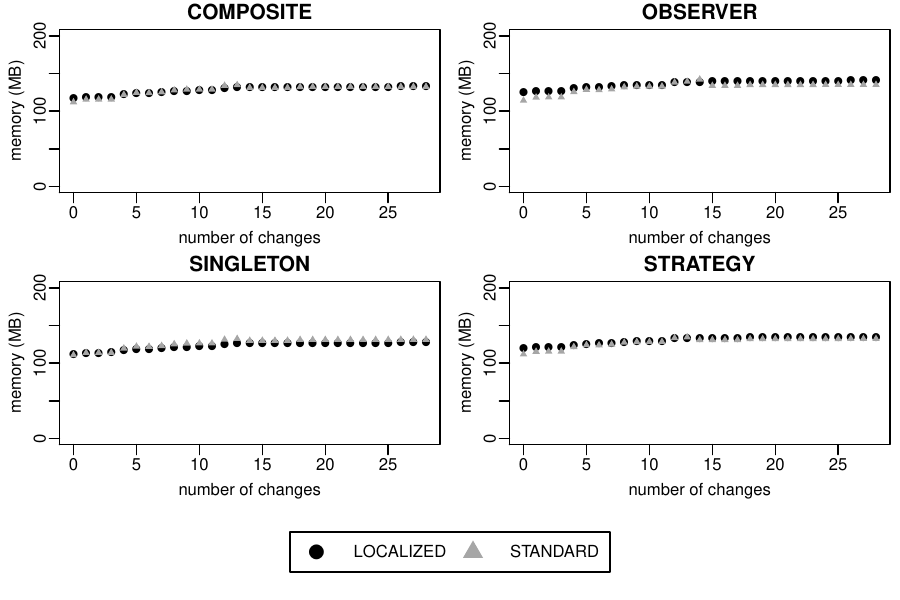}
\caption{Memory measurements for the real abstract syntax graph scenario}
\end{figure}

\clearpage

\subsection{LDBC Social Network Benchmark}

\begin{figure}[!ht]
\centering
\includegraphics[width=0.85\textwidth]{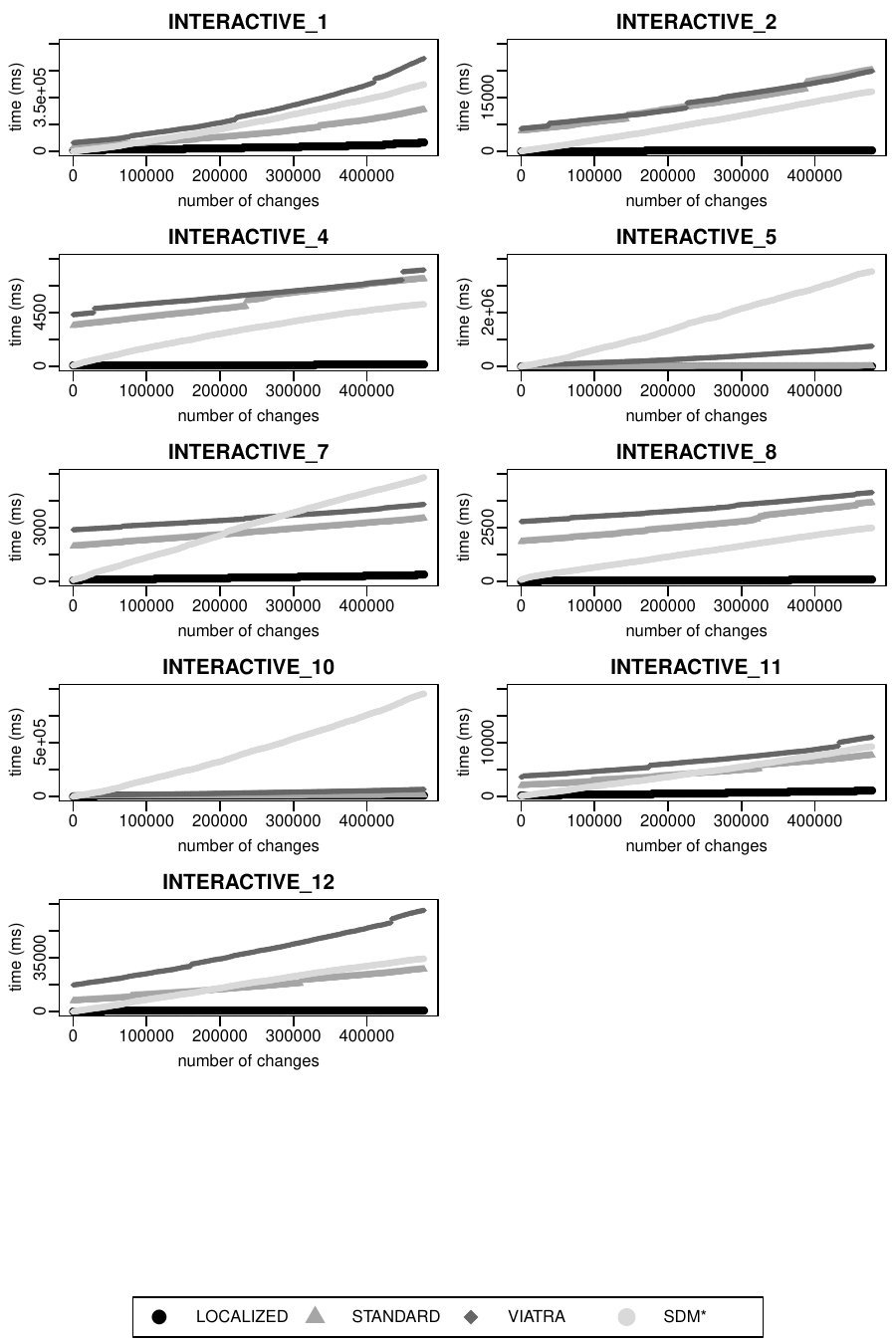}
\caption{Execution times for the LDBC scenario. For INTERACTIVE\_6 and INTERACTIVE\_9, VIATRA aborted with an exception presumably indicating too high memory consumption. For INTERACTIVE\_3, all strategies ran out of memory.}
\end{figure}

\begin{figure}[!ht]
\centering
\includegraphics[width=0.85\textwidth]{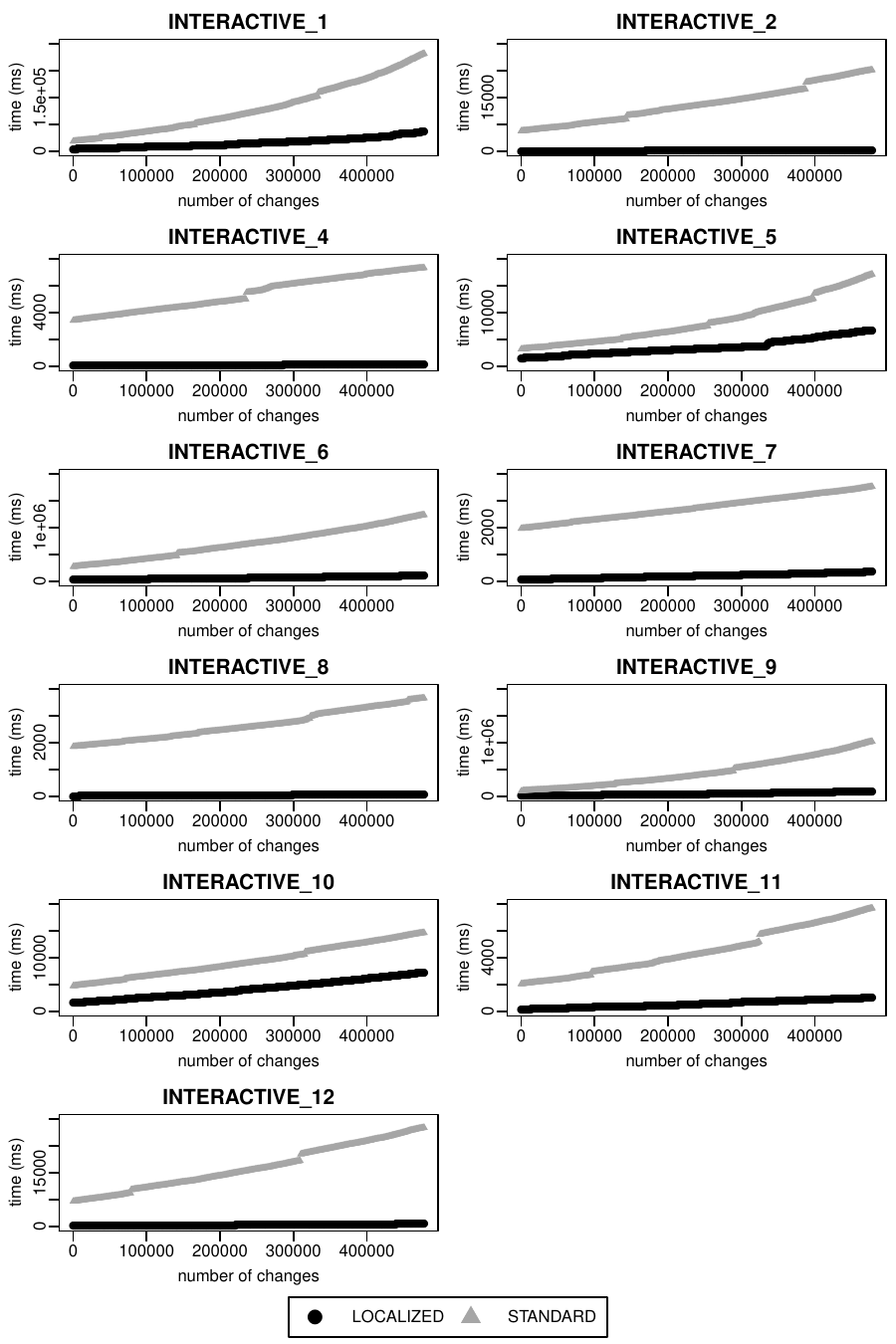}
\caption{Execution times for the LDBC scenario. For INTERACTIVE\_3, all strategies ran out of memory.}
\end{figure}

\begin{figure}[!ht]
\centering
\includegraphics[width=0.85\textwidth]{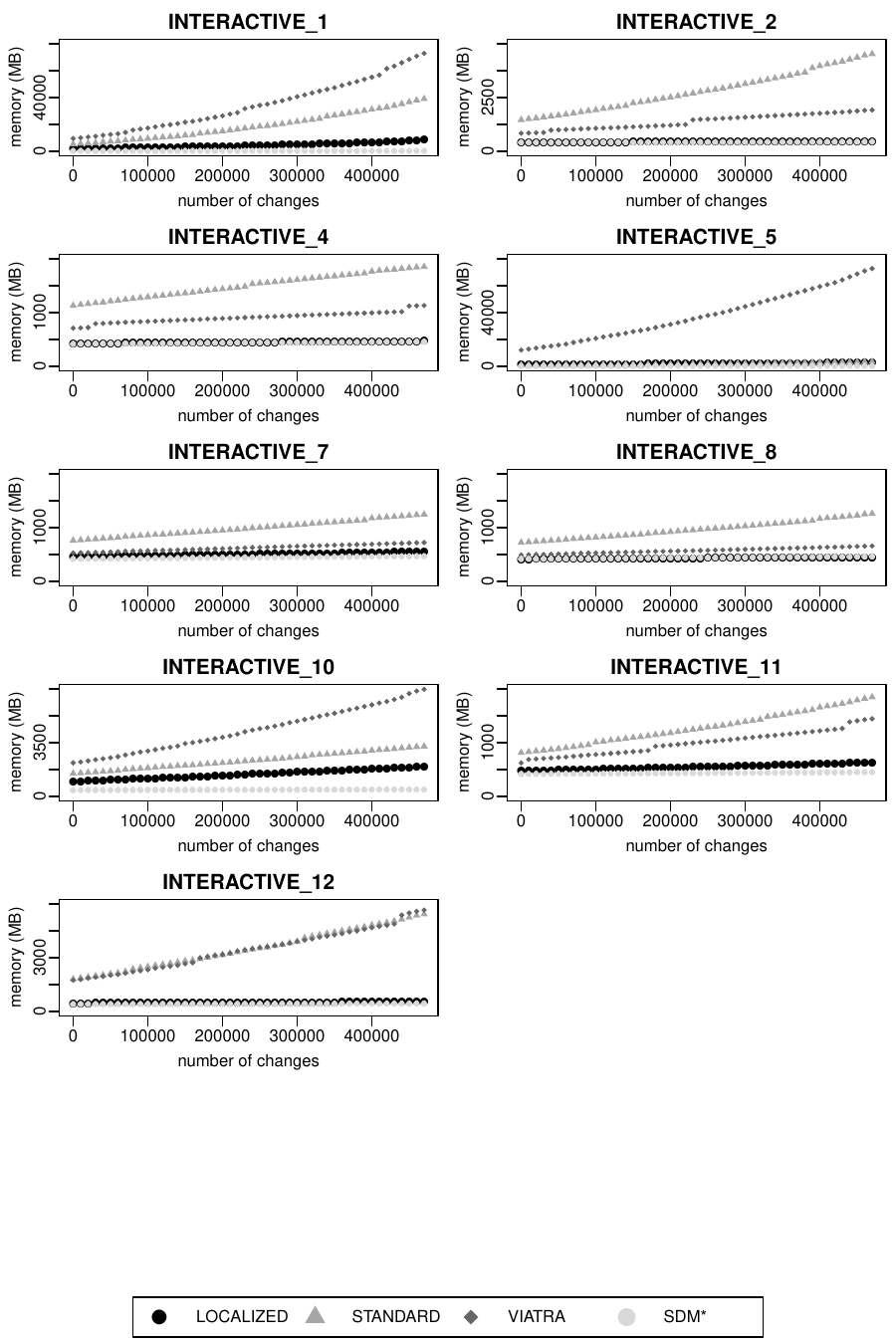}
\caption{Memory measurements for the LDBC scenario. For INTERACTIVE\_6 and INTERACTIVE\_9, VIATRA aborted with an exception presumably indicating too high memory consumption. For INTERACTIVE\_3, all strategies ran out of memory.}
\end{figure}

\begin{figure}[!ht]
\centering
\includegraphics[width=0.85\textwidth]{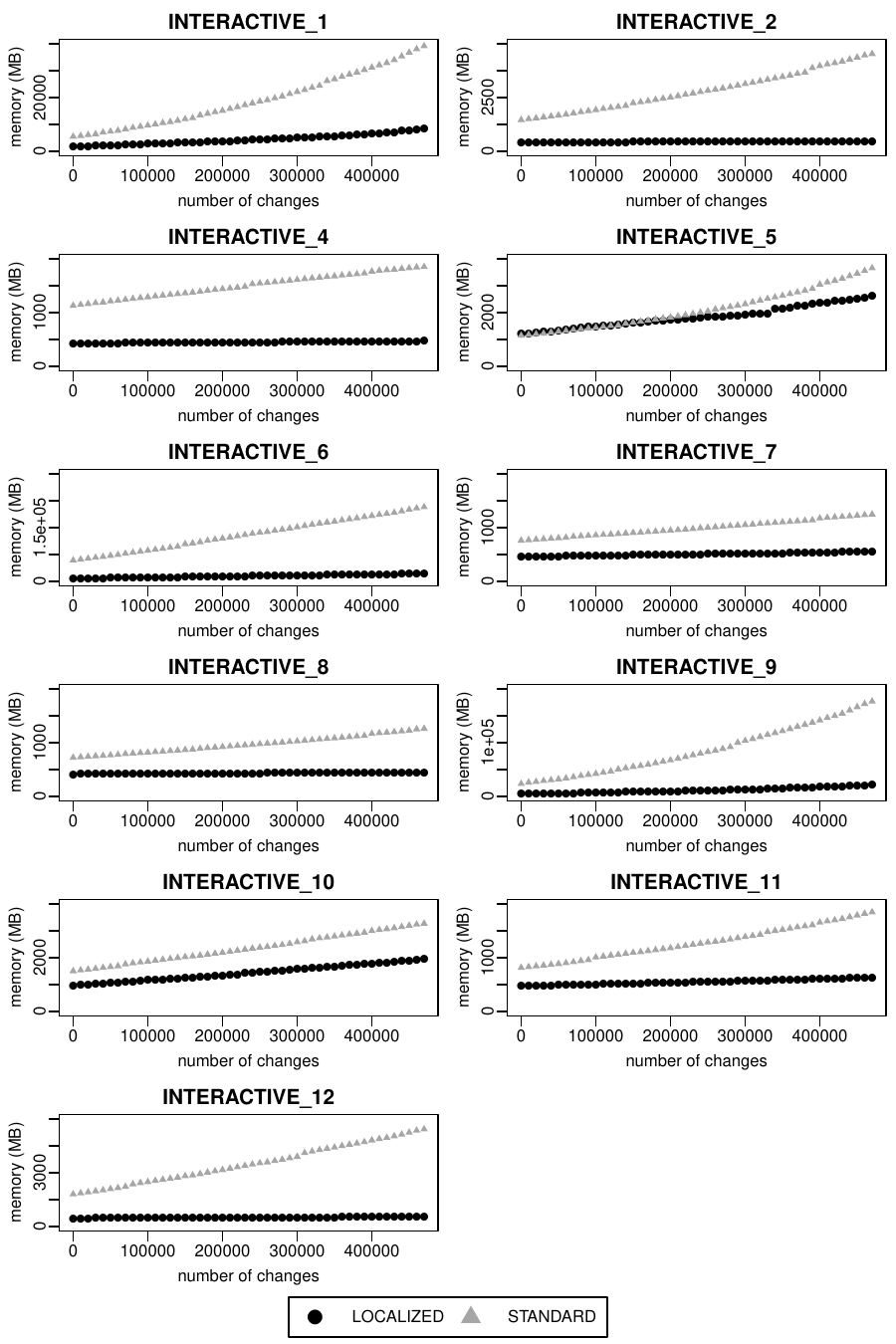}
\caption{Memory measurements for the LDBC scenario. For INTERACTIVE\_3, all strategies ran out of memory.}
\end{figure}

\clearpage


\section{Queries} \label{app:queries}

\subsection{Synthetic Abstract Syntax Graphs Queries}

\begin{figure}[!ht]
\centering
\includegraphics[width=\textwidth]{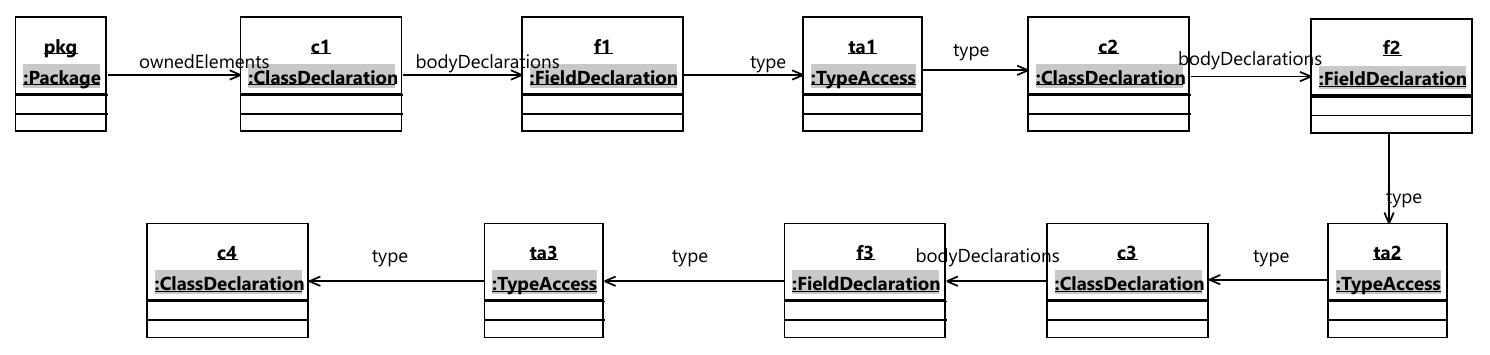}
\caption{Simple path query for the synthetic abstract syntax graph scenario}
\end{figure}

\clearpage

\subsection{Real Abstract Syntax Graphs Queries}

\begin{figure}[!ht]
\centering
\includegraphics[width=\textwidth]{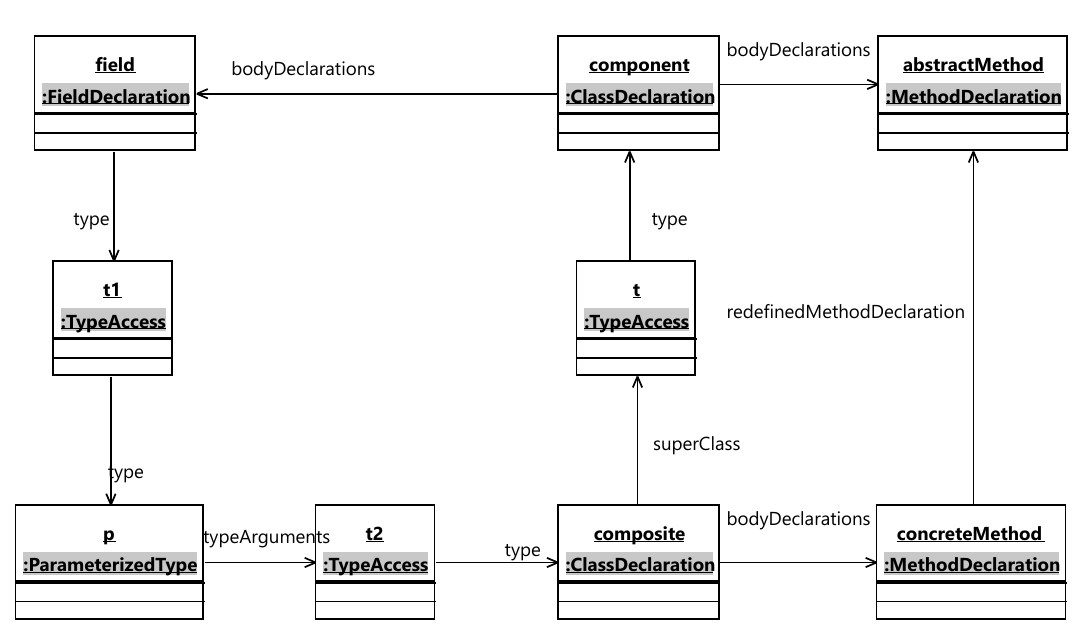}
\caption{Composite query for the real abstract syntax graph scenario}
\end{figure}

\begin{figure}[!ht]
\centering
\includegraphics[width=\textwidth]{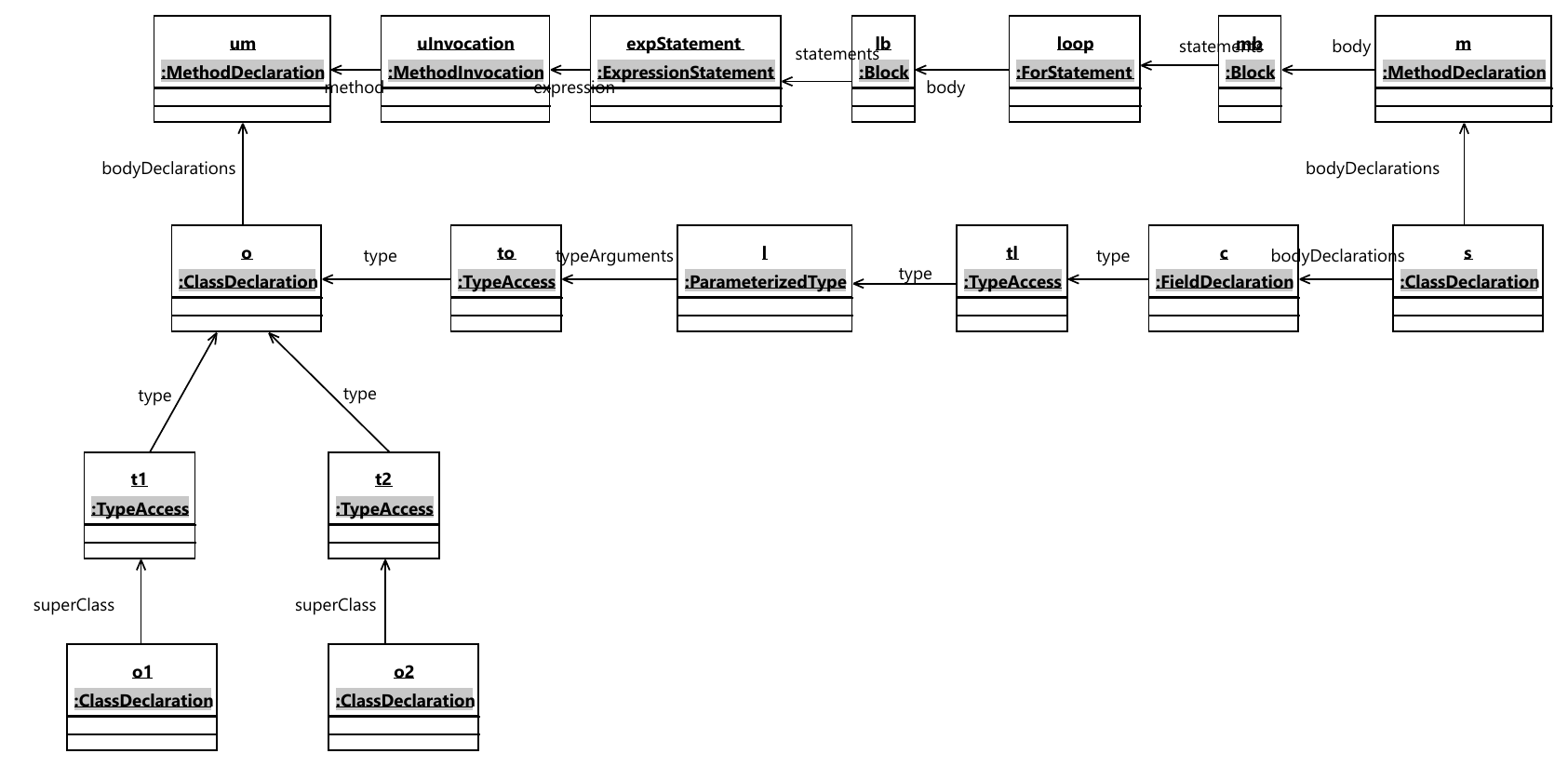}
\caption{Observer query for the real abstract syntax graph scenario}
\end{figure}

\begin{figure}[!ht]
\centering
\includegraphics[width=\textwidth]{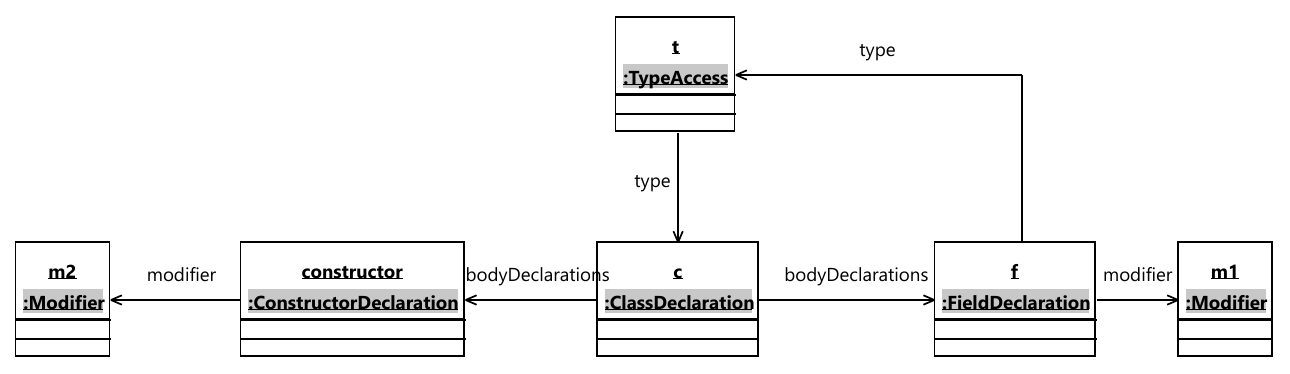}
\caption{Singleton query for the real abstract syntax graph scenario}
\end{figure}

\begin{figure}[!ht]
\centering
\includegraphics[width=\textwidth]{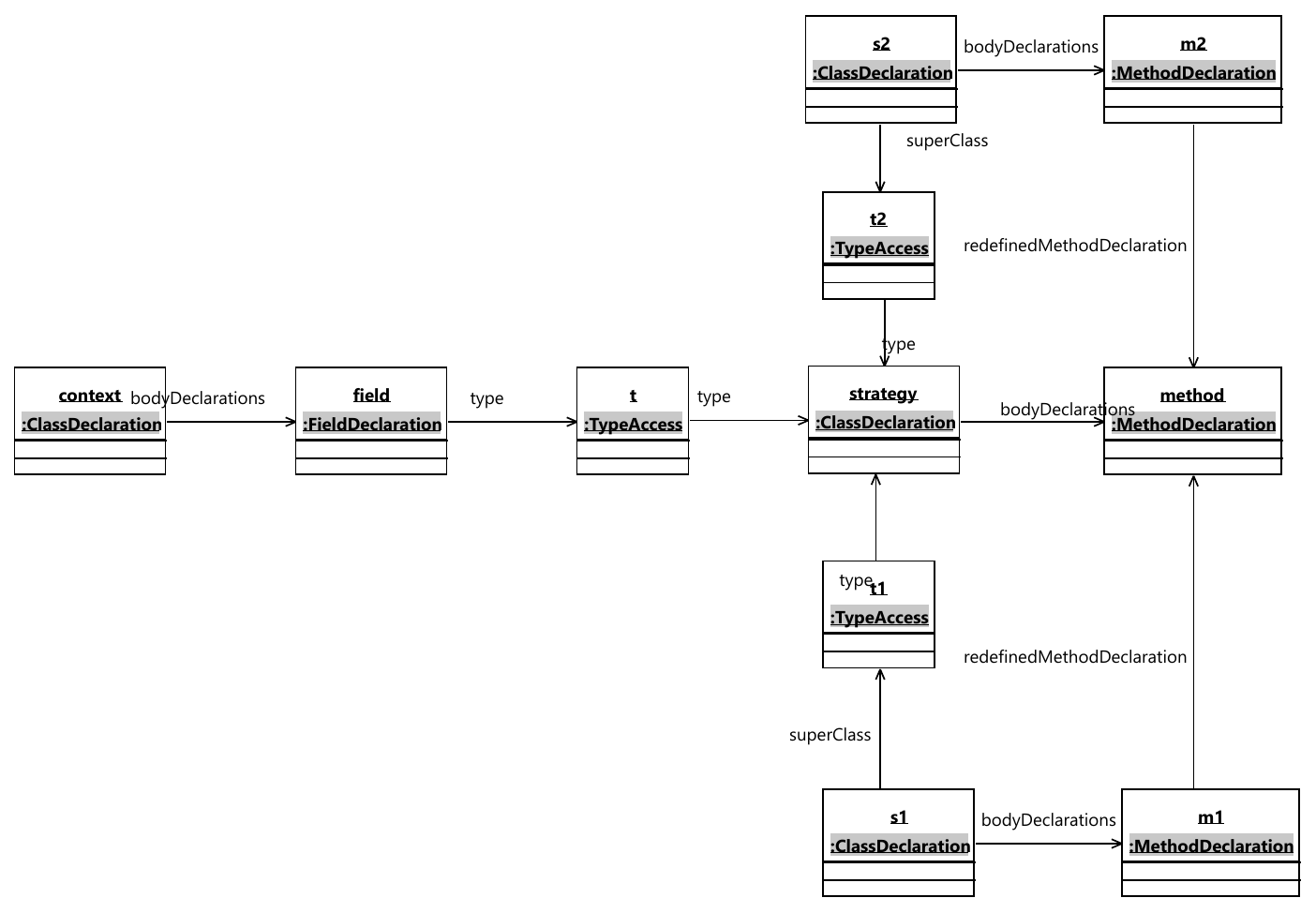}
\caption{Strategy query for the real abstract syntax graph scenario}
\end{figure}

\clearpage

\subsection{LDBC Social Network Benchmark Queries}

\begin{figure}[!ht]
\centering
\includegraphics[width=\textwidth]{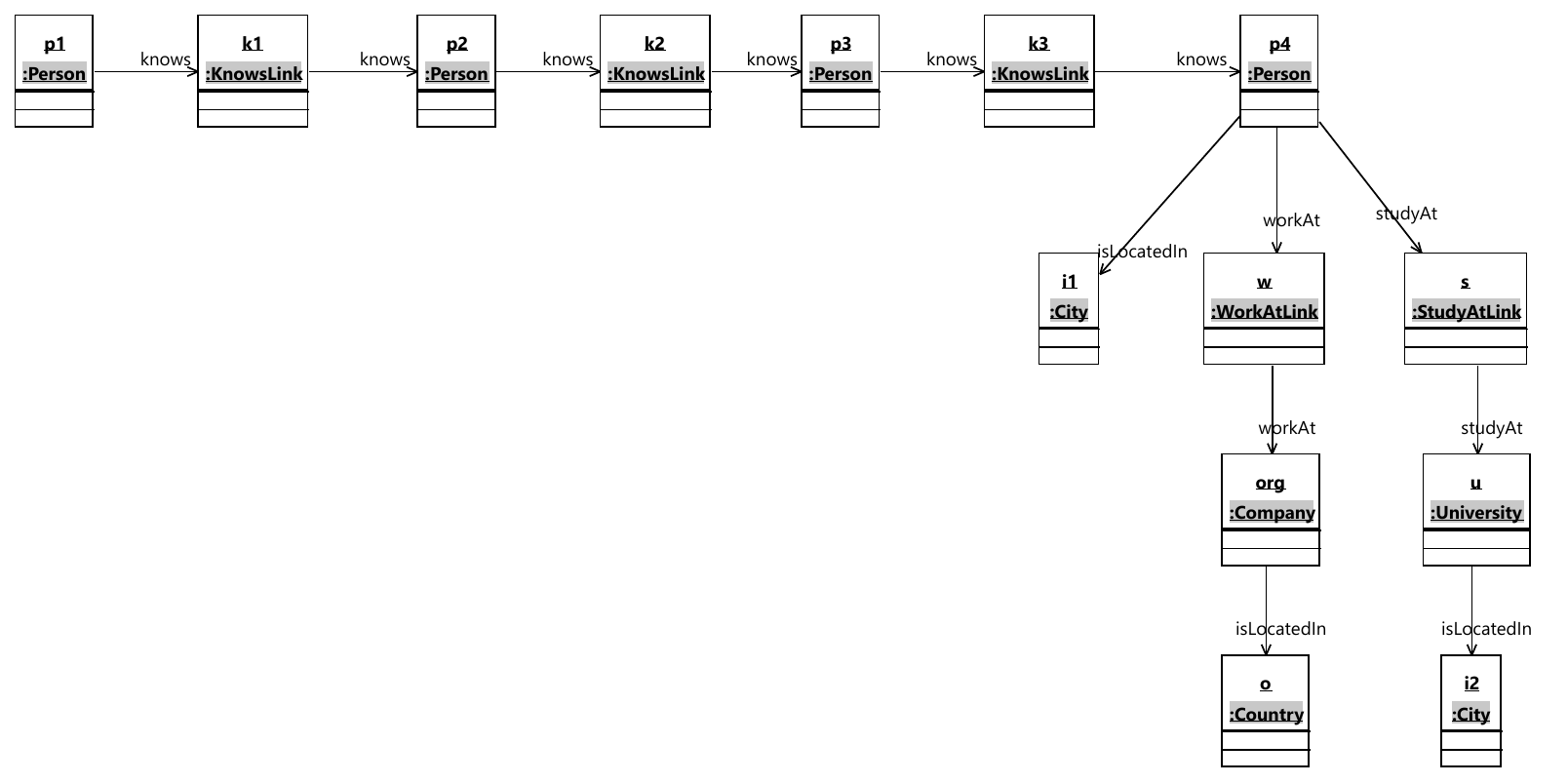}
\caption{Query ``Interactive 1'' for the LDBC scenario}
\end{figure}

\begin{figure}[!ht]
\centering
\includegraphics[width=\textwidth]{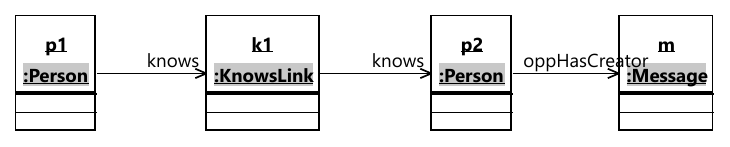}
\caption{Query ``Interactive 2'' for the LDBC scenario}
\end{figure}

\begin{figure}[!ht]
\centering
\includegraphics[width=\textwidth]{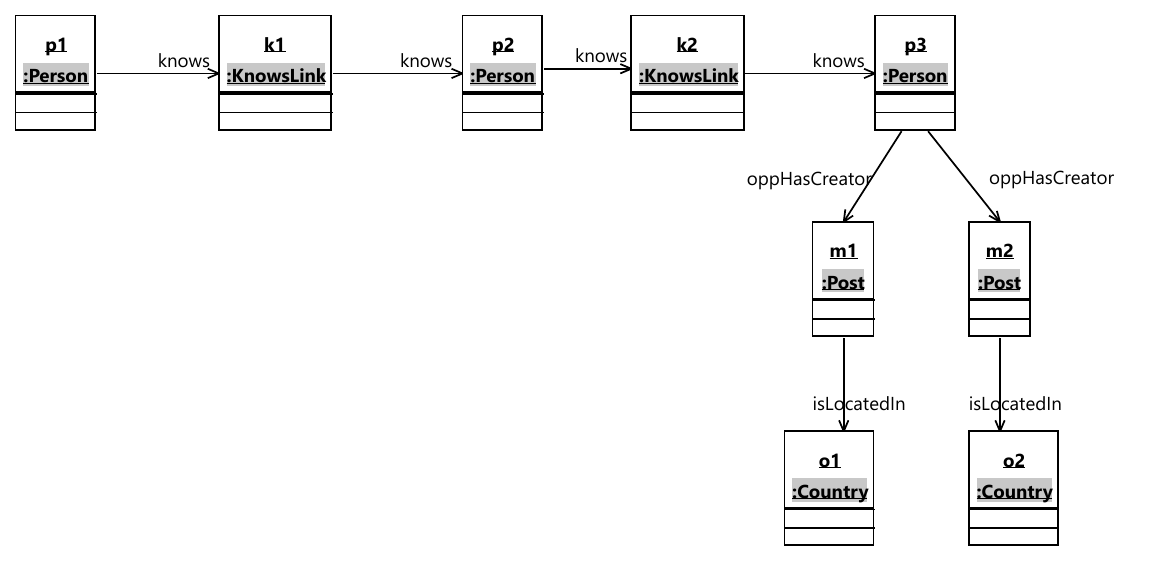}
\caption{Query ``Interactive 3'' for the LDBC scenario}
\end{figure}

\begin{figure}[!ht]
\centering
\includegraphics[width=\textwidth]{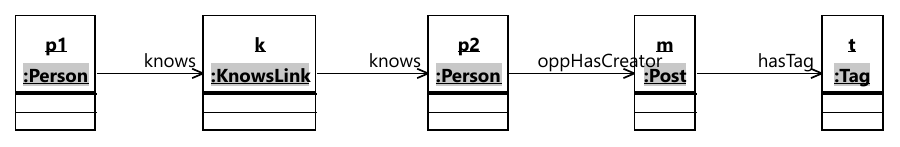}
\caption{Query ``Interactive 4'' for the LDBC scenario}
\end{figure}

\begin{figure}[!ht]
\centering
\includegraphics[width=\textwidth]{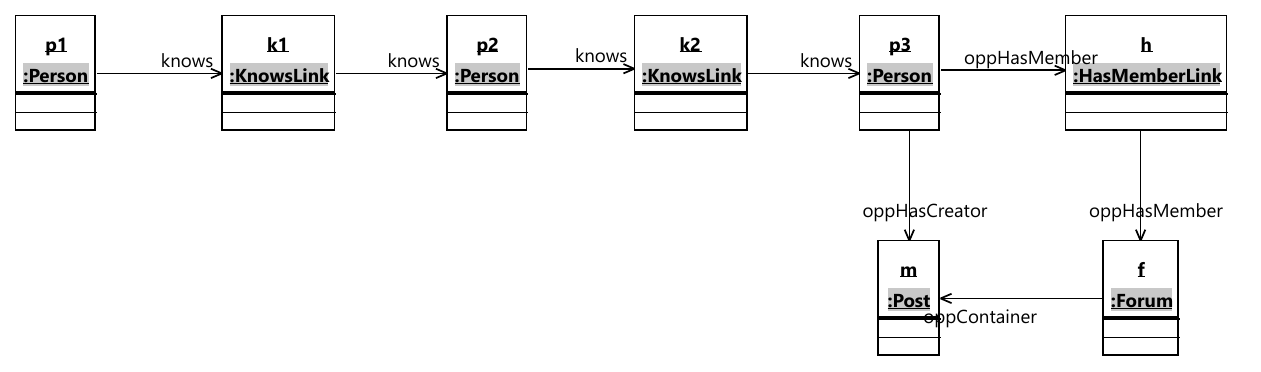}
\caption{Query ``Interactive 5'' for the LDBC scenario}
\end{figure}

\begin{figure}[!ht]
\centering
\includegraphics[width=\textwidth]{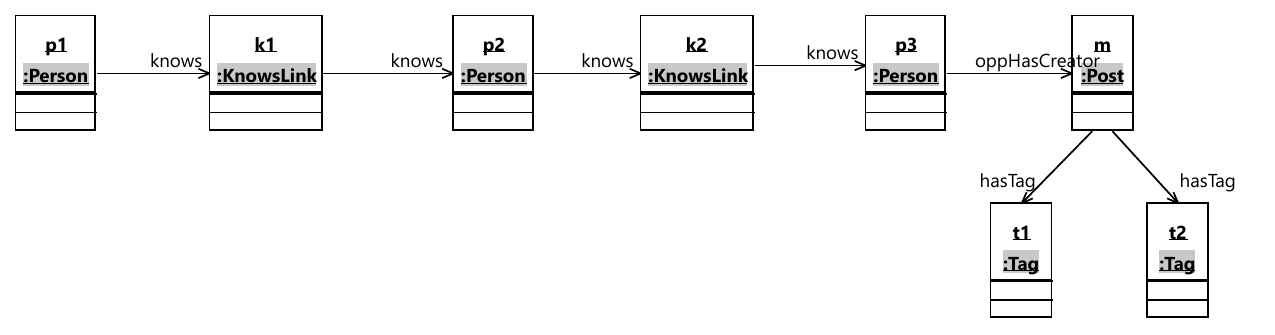}
\caption{Query ``Interactive 6'' for the LDBC scenario}
\end{figure}

\begin{figure}[!ht]
\centering
\includegraphics[width=\textwidth]{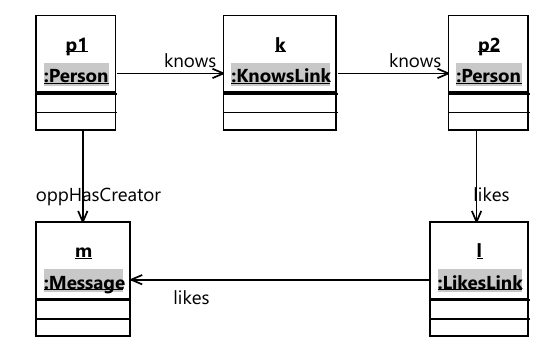}
\caption{Query ``Interactive 7'' for the LDBC scenario}
\end{figure}

\begin{figure}[!ht]
\centering
\includegraphics[width=\textwidth]{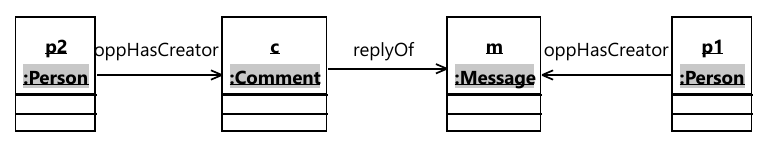}
\caption{Query ``Interactive 8'' for the LDBC scenario}
\end{figure}

\begin{figure}[!ht]
\centering
\includegraphics[width=\textwidth]{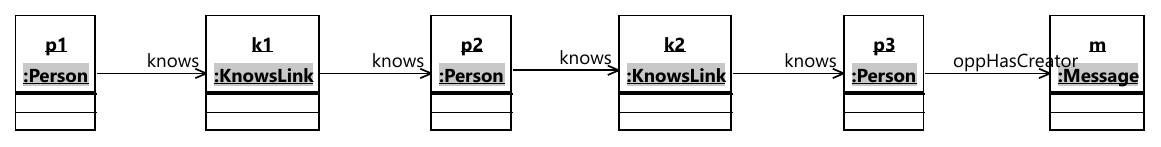}
\caption{Query ``Interactive 9'' for the LDBC scenario}
\end{figure}

\begin{figure}[!ht]
\centering
\includegraphics[width=\textwidth]{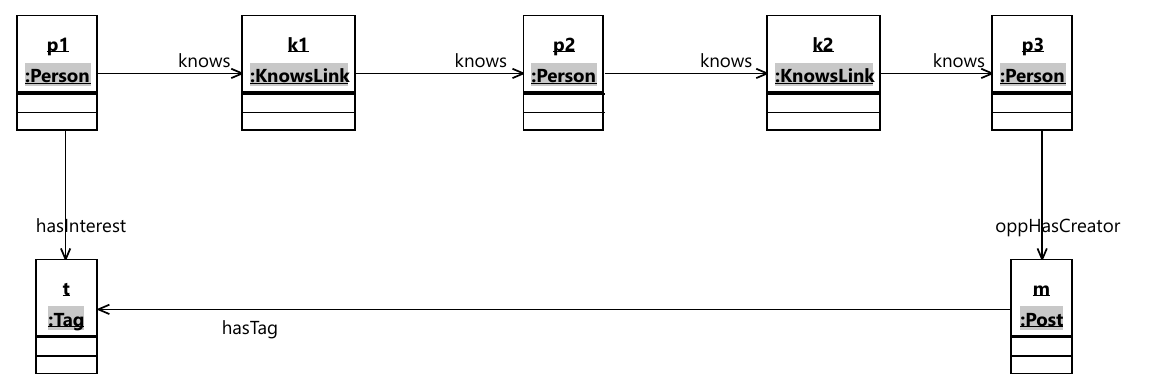}
\caption{Query ``Interactive 10'' for the LDBC scenario}
\end{figure}

\begin{figure}[!ht]
\centering
\includegraphics[width=\textwidth]{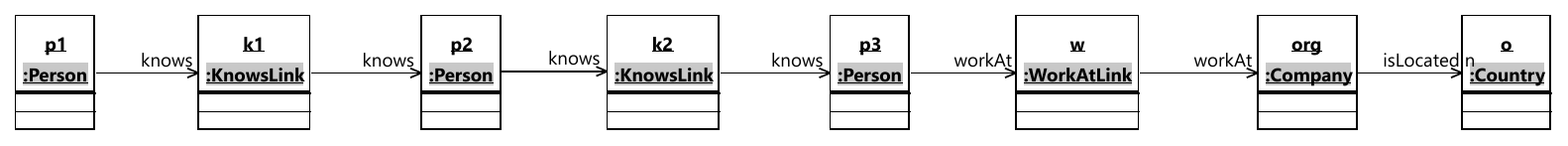}
\caption{Query ``Interactive 11'' for the LDBC scenario}
\end{figure}

\begin{figure}[!ht]
\centering
\includegraphics[width=\textwidth]{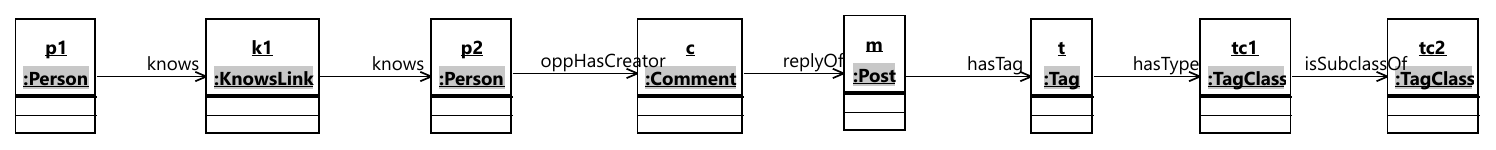}
\caption{Query ``Interactive 12'' for the LDBC scenario}
\end{figure}

\clearpage


\section{Technical Details} \label{app:technical_details}

\subsection{Additional Definitions}

\subsubsection{Fully Complete Query Results}

\begin{definition} \label{def:completeness} \emph{Completeness of Query Results}
	We say that a set of matches $M$ from query graph $Q$ into host graph $H$ is complete
	iff $M = \allmatches{Q}{H}$.
\end{definition}


\subsubsection{Target Result Sets of Marking-Sensitive RETE Nodes}

Let $H$ be a host graph, $H_p \subseteq H$ the relevant subgraph, $(N^\Phi, p^\Phi)$ the containing marking-sensitive RETE net, and $\mathcal{C}^\Phi$ a configuration for $(N^\Phi, p^\Phi)$. In order to allow targeted processing of markings and local navigation in the model, we introduce several types of marking-sensitive RETE nodes with the following target result set functions:

\begin{itemize}
	\item The target result set of a \emph{marking-sensitive join node} $[\bowtie]^\Phi$ with marking-sensitive dependencies $n^\Phi_l$ and $n^\Phi_r$ with associated subgraphs $Q_l$ and $Q_r$ such that $V^{Q_\cap} \neq \emptyset$, where $Q_\cap = Q_l \cap Q_r$, is given by
	$\resultslocal{[\bowtie]^\Phi}{N^\Phi}{H}{H_p}{\mathcal{C}^\Phi} = \{(m_l \cup m_r, max(\phi_l, \phi_r)) | (m_l, \phi_l) \in \mathcal{C}^\Phi(n^\Phi_l) \wedge (m_r, \phi_r) \in \mathcal{C}^\Phi(n^\Phi_r) \wedge m_l|_{Q_\cap} = m_r|_{Q_\cap}\}$.
	\item The target result set of a \emph{marking-sensitive union node} $[\cup]^\Phi$ with a set of marking-sensitive dependencies $N^\Phi_\alpha$ with common associated query subgraph $Q_\alpha$ is given by
	$\resultslocal{[\cup]^\Phi}{N^\Phi}{H}{H_p}{\mathcal{C}^\Phi} = \{(m, \phi_{max}))| (m, \phi_{max}) \in \bigcup_{n_\alpha \in N^\Phi_\alpha} \mathcal{C}^\Phi(n^\Phi_\alpha) \wedge \phi_{max} = max(\{ \phi' | (m, \phi') \in \bigcup_{n_\alpha \in N^\Phi_\alpha} \mathcal{C}^\Phi(n^\Phi_\alpha)\})\}$.
	\item The target result set of a \emph{marking-sensitive projection node} $[\pi_Q]^\Phi$ with marking-sensitive dependency $n^\Phi_\alpha$ and associated query subgraph $Q$ is given by
	$\resultslocal{[\pi_Q]^\Phi}{N^\Phi}{H}{H_p}{\mathcal{C}^\Phi} = \{(m|_Q, \phi_{max}))| (m, \phi_{max}) \in \mathcal{C}^\Phi(n^\Phi_\alpha) \wedge \phi_{max} = max(\{ \phi' | (m', \phi') \in \mathcal{C}^\Phi(n^\Phi_\alpha) \wedge m'|_Q = m|_Q\})\}$.
	\item The target result set of a \emph{marking assignment node} $[\phi := x]^\Phi$ with marking-sensitive dependency node $n^\Phi_\alpha$ is given by
	$\resultslocal{[\phi := x]^\Phi}{N^\Phi}{H}{H_p}{\mathcal{C}^\Phi} = \{(m, x) | (m, \phi) \in \mathcal{C}^\Phi(n^\Phi_\alpha)\}$.
	\item The target result set of a \emph{marking filter node} $[\phi > x]^\Phi$ with marking-sensitive dependency $n^\Phi_\alpha$ is given by
	$\resultslocal{[\phi > x]^\Phi}{N^\Phi}{H}{H_p}{\mathcal{C}^\Phi} = \{(m, \phi) | (m, \phi) \in \mathcal{C}^\Phi(n^\Phi_\alpha) \wedge \phi > x\}$.
	\item The target result set of a \emph{forward navigation node} $[v \rightarrow_n w]^\Phi$ with associated subgraph
	$Q = (\{v, w\}, \{e\}, \{(e, v)\}, \{(e, w)\})$ and marking-sensitive dependency $n^\Phi_v$ with associated subgraph
	$Q_v = (\{v\}, \emptyset, \emptyset, \emptyset)$ is given by
	$\resultslocal{[v \rightarrow_n w]^\Phi}{N^\Phi}{H}{H_p}{\mathcal{C}^\Phi} = \{(m, \phi_{max}) | m \in \allmatches{Q}{H} \wedge \exists (m_v, \phi_{max}) \in \mathcal{C}^\Phi(n^\Phi_v): m_v = m|_{Q_v} \wedge \phi_{max} = max(\{ \phi' | (m_v', \phi') \in \mathcal{C}^\Phi(n^\Phi_v) \wedge m_v' = m|_{Q_v}\})\}\}$.
	\item The target result set of a \emph{backward navigation node} $[w \leftarrow_n v]^\Phi$ with associated subgraph
	$Q = (\{v, w\}, \{e\}, \{(e, v)\}, \{(e, w)\})$ and marking-sensitive dependency $n^\Phi_w$ with associated subgraph
	$Q_w = (\{w\}, \emptyset, \emptyset, \emptyset)$ is given by
	$\resultslocal{[w \leftarrow_n v]^\Phi}{N^\Phi}{H}{H_p}{\mathcal{C}^\Phi} = \{(m, \phi_{max}) | m \in \allmatches{Q}{H} \wedge \exists (m_w, \phi_{max}) \in \mathcal{C}^\Phi(n^\Phi_w): m_w = m|_{Q_w} \wedge \phi_{max} = max(\{ \phi' | (m_w', \phi') \in \mathcal{C}^\Phi(n^\Phi_w) \wedge m_w' = m|_{Q_w}\})\}\}$.
	\item The target result set of a \emph{marking-sensitive vertex input node} $[v]^\Phi$ with associated subgraph $Q = (\{v\}, \emptyset, \emptyset, \emptyset)$ is given by
	$\resultslocal{[v]^\Phi}{N^\Phi}{H}{H_p}{\mathcal{C}^\Phi} = \{(m, \infty) | m \in \allmatches{Q}{H_p}\}$.
\end{itemize}


\subsubsection{Extension Points of Localized RETE Nets}

Given a local navigation structure $LNS(p)$ containing the marking-sensitive vertex input nodes $[v]^\Phi$ and $[w]^\Phi$, we call $\chi(LNS(p)) = \{[\cup]^\Phi_v, [\cup]^\Phi_w\}$ the \emph{extension points} of $LNS(p)$.

Let $(N, p)$ be a well-formed RETE net with $p$ a join node. For the localized RETE net $(N^\Phi, p^\Phi) = localize((N, p))$ with $N^\Phi = N^\Phi_{\bowtie} \cup N^\Phi_l \cup N^\Phi_r \cup RPS_l \cup RPS_r$, the extension points of $N^\Phi$ are given by $\chi(N^\Phi) = \chi(N^\Phi_l) \cup \chi(N^\Phi_r)$.

\subsection{Theorems, Lemmata, and Proofs}

\subsubsection{Localized RETE Result Completeness}


\begin{theorem} \label{the:completeness_consistent_configuration_appendix}
Let $H$ be a graph, $H_p \subseteq H$, $(N, p)$ a well-formed RETE net, and $Q$ the query graph associated with $p$. Furthermore, let $\mathcal{C}^\Phi$ be a consistent configuration for the localized RETE net $(N^\Phi, p^\Phi) = localize((N, p))$. The set of matches from $Q$ into $H$ given by the stripped result set $\resultsstripped{p^\Phi}{\mathcal{C}^\Phi}$ is then complete under $H_p$.
\end{theorem}

\begin{proof}
As a consequence of the definition of well-formed RETE nets, for each edge $e \in E^Q$, there must be an edge input node in the RETE net $(N, p)$. By construction of $(N^\Phi, p^\Phi)$, there must then be a local navigation structure in $N^\Phi$ for each such edge input node. By the definition of well-formed RETE nets, we furthermore know that each query graph vertex $v \in V^Q$ must have at least one adjacent edge.

Consequently, for each query graph vertex $v \in V^Q$, there must be a local navigation structure in $N^\Phi$ that contains a marking-sensitive vertex input node $[v]^\Phi$ that is associated with the query subgraph $Q_v = (\{v\}, \emptyset, \emptyset, \emptyset)$. By construction, $[v]^\Phi$ then is a dependendency of a marking-sensitive union node $[\cup]^\Phi_v$ that is also associated with the query subgraph $Q_v$ and is an extension point of the local navigation structure and thus of $N^\Phi$.

Since $\mathcal{C}^\Phi$ is consistent, by the semantics of the marking-sensitive vertex input node, it follows that $\forall m_v \in \allmatches{Q_v}{H}: m_v(v) \in V^{H_p} \Rightarrow (m_v, \infty) \in C([v]^\Phi)$. By the semantics of the marking-sensitive union node, it must then also hold that $\forall m_v \in \allmatches{Q_v}{H}: m_v(v) \in V^{H_p} \Rightarrow (m_v, \infty) \in C([\cup]^\Phi_v)$.

For any such match $m_v \in \allmatches{Q_v}{H}$ with $m_v(v) \in H_p$, from Lemma \ref{lem:completeness_extension_points}, it then follows that $\forall m \in \allmatches{Q}{H}: m(v) = m_v(v) \Rightarrow (m, \infty) \in \mathcal{C}^\Phi(p^\Phi)$. Because $Q_v \subseteq Q$ and thus $\forall m \in \allmatches{Q}{H}: m(v) \in V^{H_p} \Rightarrow \exists m_v \in \allmatches{Q_v}{H} : m_v(v) \in V^{H_p}$, it must thereby hold that $\forall m \in \allmatches{Q}{H}: m(v) \in V^{H_p} \Rightarrow (m, \infty) \in \mathcal{C}^\Phi(p^\Phi)$.

Since this is the case for all $v \in V^Q$, it follows that $\forall m \in \allmatches{Q}{H}: (\exists v \in V^Q : m(v) \in V^{H_p}) \Rightarrow (m, \infty) \in \mathcal{C}^\Phi(p^\Phi)$.

Thus, it finally follows that $\forall m \in \allmatches{Q}{H}: (\exists v \in V^Q : m(v) \in V^{H_p}) \Rightarrow m \in \resultsstripped{p^\Phi}{\mathcal{C}^\Phi}$. Hence, $\resultsstripped{p^\Phi}{\mathcal{C}^\Phi}$ is complete under $H_p$.
\end{proof}

\clearpage

\subsubsection{Localized RETE Execution Order Consistency}


\begin{theorem} \label{the:execution_order_consistency_appendix}
Let $H$ be a graph, $H_p \subseteq H$, $(N, p)$ a well-formed RETE net, and $\mathcal{C}^\Phi_0$ an arbitrary starting configuration. Executing the marking-sensitive RETE net $(N^\Phi, p^\Phi) = localize((N, p))$ via $O = order((N^\Phi, p^\Phi))$ then yields a consistent configuration $\mathcal{C}^\Phi = execute(O, N^\Phi, H, H_p, \mathcal{C}^\Phi_0)$.
\end{theorem}

\begin{proof}
Follows directly from Lemma \ref{lem:order_robustness}.
\end{proof}

\subsubsection{Localized RETE Execution Correctness}


\begin{corollary}
Let $H$ be a graph, $H_p \subseteq H$, $(N, p)$ a well-formed RETE net, and $Q$ the query graph associated with $p$. Furthermore, let $\mathcal{C}^\Phi_0$ be an arbitrary starting configuration for the localized RETE net $(N^\Phi, p^\Phi) = localize((N, p))$ and $\mathcal{C}^\Phi = execute(order((N^\Phi, p^\Phi)), N^\Phi, H, H_p, \mathcal{C}^\Phi_0)$. The set of matches from $Q$ into $H$ given by $\resultsstripped{p^\Phi}{\mathcal{C}^\Phi}$ is then complete under $H_p$.
\end{corollary}

\begin{proof}
Follows directly from Theorem \ref{the:completeness_consistent_configuration} and Theorem \ref{the:execution_order_consistency}.
\end{proof}

\subsubsection{Localized RETE Configuration Size}


\begin{theorem} \label{the:upper_bound_configuration_size_appendix}
Let $H$ be an edge-dominated graph, $H_p \subseteq H$, $(N, p)$ a well-formed RETE net with $Q$ the associated query graph of $p$, $\mathcal{C}$ a consistent configuration for $(N, p)$ for host graph $H$, and $\mathcal{C}^\Phi$ a consistent configuration for the localized RETE net $(N^\Phi, p^\Phi) = localize((N, p))$ for host graph $H$ and relevant subgraph $H_p$. It then holds that $\sum_{n^\Phi \in V^{N^\Phi}} \sum_{(m, \phi) \in \mathcal{C}(n^\Phi)} |m| \leq 7 \cdot |\mathcal{C}|_e$.
\end{theorem}

\begin{proof}
For each edge input node $[v \rightarrow w]$ in $N$ with associated query subgraph $Q_e$, $N^\Phi$ contains the seven nodes of the corresponding local navigation structure $LNS([v \rightarrow w])$: $[v]^\Phi$, $[w]^\Phi$, $[\cup]^\Phi_{v}$, $[\cup]^\Phi_{w}$, $[v \rightarrow_n w]^\Phi$, $[w \leftarrow_n v]^\Phi$, and $[\cup]^\Phi$. By the semantics of $[v \rightarrow w]$, it must hold that $\mathcal{C}([v \rightarrow w]) = \allmatches{Q_e}{H}$.

By the semantics of the forward and backward local navigation nodes, we know that $\resultsstripped{[v \rightarrow_n w]^\Phi}{\mathcal{C}^\Phi} \subseteq \allmatches{Q_e}{H}$ and $\resultsstripped{[w \leftarrow_n v]^\Phi}{\mathcal{C}^\Phi} \subseteq \allmatches{Q_e}{H}$ and thus, by the semantics of the marking sensitive union node, $\resultsstripped{[\cup]^\Phi}{\mathcal{C}^\Phi} \subseteq \allmatches{Q_e}{H}$.

Furthermore, by the semantics of the marking sensitive input nodes, it must hold that $\resultsstripped{[v]^\Phi}{\mathcal{C}^\Phi} \subseteq \allmatches{Q_{v}}{H}$ and $\resultsstripped{[w]^\Phi}{\mathcal{C}^\Phi} \subseteq \allmatches{Q_{w}}{H}$, where $Q_{v}$ and $Q_{w}$ are the query subgraphs associated with $v$ respectively $w$. By the semantics of the marking sensitive union nodes and the assumption regarding the common query subgraph associated with their dependencies, it must also hold that $\resultsstripped{[\cup]_{v}^\Phi}{\mathcal{C}^\Phi} \subseteq \allmatches{Q_{v}}{H}$ and $\resultsstripped{[\cup]_{w}^\Phi}{\mathcal{C}^\Phi} \subseteq \allmatches{Q_{w}}{H}$.

By the semantics of the involved marking-sensitive RETE nodes, for each node in $LNS([n_1 \rightarrow n_2])$, each match can only be associated with at most one marking in the node's current result set in $\mathcal{C}^\Phi$. Furthermore, by assumption, it holds that $\allmatches{Q_v}{H} \leq \allmatches{Q_e}{H}$ and $\allmatches{Q_w}{H} \leq \allmatches{Q_e}{H}$. Lastly, we know that $|Q_v| = |Q_w| \leq |Q_e|$.

It thus holds that $\sum_{n^\Phi \in V^{LNS([v \rightarrow w])}} \sum_{(m, \phi) \in \mathcal{C}(n^\Phi)} |m| \leq 7 \cdot \sum_{m \in \mathcal{C}([v \rightarrow w])} |m|$.

For each join node $[\bowtie]$ in $N$ with associated query subgraph $Q_{\bowtie}$ and dependencies $n_l$ and $n_r$, $N^\Phi$ contains the corresponding marking-sensitive join node $[\bowtie]^\Phi$ with dependencies $n^\Phi_l$ and $n^\Phi_r$ corresponding to $n_l$ respectively $n_r$, as well as the six nodes from the related request projection structures: $[\phi > h]^\Phi_l$, $[\pi_{Q_v}]^\Phi_l$, $[\phi > h]^\Phi_l$, $[\phi > h]^\Phi_r$, $[\pi_{Q_v}]^\Phi_r$, and $[\phi > h]^\Phi_r$.

By the semantics of join and marking sensitive join, we know that $\mathcal{C}(n_l) \subseteq \resultsstripped{n^\Phi_l}{\mathcal{C}^\Phi} \wedge \mathcal{C}(n_r) \subseteq \resultsstripped{n^\Phi_r}{\mathcal{C}^\Phi} \Rightarrow \mathcal{C}([\bowtie]) \subseteq \resultsstripped{[\bowtie]^\Phi}{\mathcal{C}^\Phi}$. For every edge input node $[v \rightarrow w]$ in $N$, we know that $\mathcal{C}([v \rightarrow w]) \subseteq \resultsstripped{[\cup]^\Phi}{\mathcal{C}^\Phi}$, where $[\cup]^\Phi$ is the root node of $LNS([v \rightarrow w])$. It is easy to see that thereby, it must hold that $\mathcal{C}([\bowtie]) \subseteq \resultsstripped{[\bowtie]^\Phi}{\mathcal{C}^\Phi}$. By Lemma \ref{lem:match_uniqueness}, it follows that $\sum_{(m, \phi) \in \mathcal{C}([\bowtie]^\Phi)} |m| \leq \sum_{m \in \mathcal{C}([\bowtie])} |m|$.

All nodes in $RPS_{\bowtie} =\{[\phi > h]^\Phi_l, [\pi_{Q_v}]^\Phi_l, [\phi > h]^\Phi_l, [\phi > h]^\Phi_r, [\pi_{Q_v}]^\Phi_r, [\phi > h]^\Phi_r\}$ have the same associated query subgraph $Q_v \subseteq Q_{\bowtie}$. Thus, we know that for any pair of matches $m_v \in \resultsstripped{n^\Phi_v}{\mathcal{C}^\Phi}$ and $m_{\bowtie} \in \mathcal{C}([\bowtie])$, with $n^\Phi_v \in RPS_{\bowtie}$, it must hold that $|m_v| \leq |m_{\bowtie}|$. Furthermore, by the semantics of the marking filter, marking-sensitive projection and marking assignment nodes, we know that for all $n^\Phi_v \in RPS_{\bowtie}$, it holds that $|\mathcal{C}^\Phi((n^\Phi_v)| \leq |\mathcal{C}^\Phi([\bowtie]^\Phi)| \leq |\mathcal{C}([\bowtie])|$. It thus follows for all six $n^\Phi_v \in RPS_{\bowtie}$ that $\sum_{(m, \phi) \in \mathcal{C}(n^\Phi_v)} |m| \leq \sum_{m \in \mathcal{C}([\bowtie])} |m|$. Hence, it must hold that $\sum_{n^\Phi_{\bowtie} \in RPS_{\bowtie} \cup \{[\bowtie]^\Phi\}} \sum_{(m, \phi) \in \mathcal{C}(n^\Phi_{\bowtie})} |m| \leq 7 \cdot \sum_{m \in \mathcal{C}([\bowtie])} |m|$.

Since $N^\Phi$ only consists of the local navigation structures corresponding to edge input nodes in $N$ and marking-sensitive joins corresponding to joins in $N$ along with the accordingly constructed request projection structures, it follows that $\sum_{n^\Phi \in V^{N^\Phi}} \sum_{(m, \phi) \in \mathcal{C}(n^\Phi)} |m| \leq 7 \cdot |\mathcal{C}|$.
\end{proof}

\subsubsection{Localized RETE Memory Consumption}


\begin{corollary}
Let $H$ be an edge-dominated graph, $H_p \subseteq H$, $(N, p)$ a well-formed RETE net, $\mathcal{C}$ a consistent configuration for $(N, p)$ for host graph $H$, and $\mathcal{C}^\Phi$ a consistent configuration for the localized RETE net $(N^\Phi, p^\Phi) = localize((N, p))$ for host graph $H$ and relevant subgraph $H_p$. Assuming that storing a match $m$ requires an amount of memory in $O(|m|)$ and storing an element from $\overline{\mathbb{N}}$ requires an amount of memory in $O(1)$, storing $\mathcal{C}^\Phi$ requires memory in $O(|\mathcal{C}|_e)$.
\end{corollary}

\begin{proof}
Follows directly from Theorem \ref{the:upper_bound_configuration_size} and the fact that no node in a well-formed or localized RETE net is associated with an empty query subgraph, and a consistent configuration for either kind of RETE net thus contains no empty matches.
\end{proof}

\subsubsection{Standard RETE Time Complexity}


\begin{theorem} \label{the:complexity_time_original}
Let $H$ be a graph, $(N, p)$ a well-formed RETE net for query graph $Q$, and $\mathcal{C}_0$ the empty configuration for $(N, p)$. Executing $(N, p)$ via a reverse topological sorting $\mathcal{C}_1 = execute(toposort^{-1}(N), N, H, \mathcal{C}_0)$ then takes $O(|\mathcal{C}_1|)$ steps.
\end{theorem}

\begin{proof}
Given a suitable representation of $H$ that allows efficient retrieval of matching edges, an edge input node $n_e \in V^N$ can be executed by enumerating the matching edges, creating the corresponding trivial matches, and adding them to $\mathcal{C}_0(n_e)$. Assuming indexing structures that allow insertion times linear in the size of the inserted match, this takes effort in $O(S)$, where $S = \sum_{m \in \mathcal{C}_1(n_e)} |m|$.

A join node $n_{\bowtie} \in V^N$ with dependencies $n_l$ and $n_r$ can be executed for configuration $\mathcal{C}_{0.x}$ by enumerating the matches from $\mathcal{C}_{0.x}(n_l)$, retrieving the complementary matches from $\mathcal{C}_{0.x}(n_r)$, constructing the corresponding matches for $n_{\bowtie}$ and adding them to $\mathcal{C}_0(n_{\bowtie})$. Notably, each enumerated combination of a match from $\mathcal{C}_{0.x}(n_l)$ and a complementary match from $\mathcal{C}_{0.x}(n_r)$ results in a distinct match for $n_{\bowtie}$. Thus, assuming ideal indexing structures that allow match enumeration in time linear in the number of matches, retrieval of complementary matches in time linear in the size of the opposite match, and insertion times linear in the size of the inserted match, execution of $n_{\bowtie}$ takes effort in $O(S_l + S_{\bowtie})$, where $S_l = \sum_{m_l \in \mathcal{C}_{0.x}(n_l)} |m_l|$ and $S_{\bowtie} = \sum_{m \in \mathcal{C}_1(n_{\bowtie})} |m|$.

Because each node in $V^N$ is only contained once in $toposort^{-1}(N)$ and has at most one dependent node due to the well-formedness criterion, execution of $(N, p)$ thus takes $O(|\mathcal{C}_1|)$ steps.
\end{proof}

\subsubsection{Localized RETE Time Complexity}


\begin{theorem} \label{the:complexity_time_localized_appendix}
Let $H$ be an edge-dominated graph, $H_p \subseteq H$, $(N, p)$ a well-formed RETE net for query graph $Q$, $\mathcal{C}$ a consistent configuration for $(N, p)$, and $\mathcal{C}^\Phi_0$ the empty configuration for the localized RETE net $(N^\Phi, p^\Phi) = localize((N, p))$. Executing $(N^\Phi, p^\Phi)$ via $execute(order((N^\Phi, p^\Phi)), N^\Phi, H, H_p, \mathcal{C}^\Phi_0)$ then takes $O(T \cdot (|Q_a| + |Q|))$ steps.
\end{theorem}

\begin{proof}
From Lemmata \ref{lem:monotonicity_forward_navigation}-\ref{lem:monotonicity_join} and $\mathcal{C}^\Phi_0$ being empty it follows that matches are never removed from a current result set during the execution of $(N^\Phi, p^\Phi)$ and their marking in such a current result set is never decreased. $\mathcal{C}^\Phi_0$ being empty also entails that $\mathcal{C}^\Phi_0$ is consistent for all nodes in $N^\Phi$ except for marking-sensitive vertex input nodes.

By Lemmata \ref{lem:execution_time_forward_navigation}-\ref{lem:execution_time_vertex_input_consistent}, we then know that whenever a marking-sensitive node $n^\Phi$ that is not a marking-sensitive join node is executed, execution takes $O(S^+_D + S^+_R + N^\uparrow_D + N^\uparrow_R)$ steps, where, compared to the last execution of $n^\Phi$, $S^+_D$ is the total size of matches added to the current result sets of dependencies of $n^\Phi$, $S^+_R$ is the total size of matches added to the current result set of $n^\Phi$, $N^\uparrow_D$ is the number of matches whose marking has been increased in the result sets of depedencies of $n^\Phi$, and $N^\uparrow_R$ is the number of matches whose marking has been increased in the result set of $n^\Phi$.

Each match is only added to a result set once and by the construction of $(N^\Phi, p^\Phi)$, at most $|Q|$ different marking values are possible for matches in a configuration for $(N^\Phi, p^\Phi)$. Furthermore, by the construction of $(N^\Phi, p^\Phi)$, each node in $N^\Phi$ can be a dependency of at most two other RETE nodes. It thus follows that the total effort for executing all non-join nodes in $N^\Phi$ must be in $O(|\mathcal{C}| + |Q| \cdot T)$.

For marking-sensitive join nodes, a slightly different argument has to be made, since the number of steps required for an individual execution is not necessarily in $O(S^+_D + S^+_R + N^\uparrow_D + N^\uparrow_R)$. This is due to the fact that the increase of a marking of a match in the result set of one of the join node's dependencies may make recomputation of the marking of a number of matches in the join's current result set necessary without actually resulting in any change.

However, we can observe that the effort for processing matches newly added to the result sets of the dependencies of a marking-sensitive join node $[\bowtie]^\Phi$ is in $O(S^+_D + S^+_R)$. This is due to the fact that, assuming consistency of the previous configuration for $[\bowtie]^\Phi$, these changes can be handled by constructing the new intermediate results spawned by the newly added dependency matches. For each such newly added dependency match, all complementary matches from the opposite dependency have to be retrieved and combined into result matches.

This takes effort in the size of the newly added dependency match and, since the query subgraphs associated with dependencies of $[\bowtie]^\Phi$ are subgraphs of the query subgraph associated with $[\bowtie]^\Phi$ and each complementary match for the opposite dependency results in a distinct intermediate result for the join, the combined size of newly created result matches.

Assuming changes to dependency result sets are cached and ideal indexing structures, this entails an overall effort in $O(S^+_D + S^+_R)$ for handling matches newly added to the result sets of dependencies of $[\bowtie]^\Phi$. An increase of the marking of a match in the result set of a dependency of $[\bowtie]^\Phi$ can be handled by recomputing the marking of all associated matches in the result set of $[\bowtie]^\Phi$. By explicitly storing the associated result matches of each dependency match, this only requires effort in the number of such associated result matches.

The number of associated result matches may not correspond to the number of result matches whose marking actually increases, since the result match marking also depends on the marking of the complementary match in the opposite dependency. However, since each result match is associated with exactly two dependency matches, we know that overall, the effort for recomputing the marking of result matches must still be linear in the number of result tuples and the number of different possible markings. Combined with the baseline effort for processing marking increases for dependency matches and the effort for handling newly added dependency matches, this means that the total effort for executing all marking-sensitive join nodes in $N^\Phi$ must also be in $O(|\mathcal{C}| + |Q| \cdot T)$.

Assuming that redundant execution of nodes whose dependency result sets have not changed since the last execution is avoided via a notification mechanism, this yields an overall effort in $O(|\mathcal{C}| + |Q| \cdot T)$ for executing $(N^\Phi, p^\Phi)$ via $execute(order((N^\Phi, p^\Phi)), N^\Phi, H, H_p, \mathcal{C}^\Phi_0)$.
\end{proof}

\subsubsection{Supplementary Lemmata}


\begin{lemma} \label{lem:completeness_extension_points}
Let $H$ be a graph, $(N, p)$ a well-formed RETE net with join tree height $h$ and $Q$ the query graph associated with $p$, $(X^\Phi, p_X^\Phi)$ a modular extension of the localized RETE net $(N^\Phi, p^\Phi) = localize((N, p))$, and $\mathcal{C}$ a configuration for $(X^\Phi, p_X^\Phi)$ that is consistent for all nodes in $V^{N^\Phi}$. Furthermore, let $x \in \chi(N^\Phi)$, $Q_v = (\{v\}, \emptyset, \emptyset, \emptyset)$ the query subgraph associated with $x$, and $(m', \psi) \in \mathcal{C}(x)$ such that $\psi > h$ for some match $m' : Q_v \rightarrow H$. It then holds that $\forall m \in \allmatches{Q}{H}: m(v) = m'(v) \Rightarrow \exists (m, \phi) \in \mathcal{C}(p^\Phi) : \phi \geq \psi$.
\end{lemma}

\begin{proof}
We prove the lemma by induction over the height $h$ of the join tree $N$.

Let $m \in \allmatches{Q}{H}$ such that $m(v) = m'(v)$. 

In the base case of $h = 0$, $N$ consists of a single edge input node $p$ and $(N^\Phi, p^\Phi)$ is hence given by $(N^\Phi, p^\Phi) = (LNS(p), [\cup]^\Phi)$. We thus also know that $\chi(N^\Phi) = \{[\cup]^\Phi_v, [\cup]^\Phi_w\}$. We can then distinguish two cases:

Case 1: $x = [\cup]^\Phi_v$. By the semantics of the forward navigation node $[v \rightarrow_n w]^\Phi$, which must have $Q$ as its associated query subgraph, and by $\mathcal{C}$ being consistent for all nodes in $V^{N^\Phi}$, it follows that $(m, \psi) \in \mathcal{C}([v \rightarrow_n w]^\Phi)$. By the semantics of the marking sensitive union node $p^\Phi = [\cup]^\Phi$, it then follows that $(m, \phi) \in \mathcal{C}(p^\Phi)$ for some $\phi \geq \psi$.

Case 2: $x = [\cup]^\Phi_w$. By the semantics of the backward navigation node $[w \leftarrow_n v]^\Phi$, which must have $Q$ as its associated query subgraph, and by $\mathcal{C}$ being consistent for all nodes in $V^{N^\Phi}$, it follows that $(m, \psi) \in \mathcal{C}([w \leftarrow_n v]^\Phi)$. By the semantics of the marking sensitive union node $p^\Phi = [\cup]^\Phi$, it then follows that $(m, \phi) \in \mathcal{C}(p^\Phi)$ for some $\phi \geq \psi$.

Thus, the lemma holds in the base case.

For the induction step, we show that if the lemma holds for all well-formed RETE nets with join tree height $h$, it must also hold for any well-formed RETE net $(N, p)$ with join tree height $h + 1$.

Since $h + 1 \geq 1$, we know that $p$ must be a join node. $(N^\Phi, p^\Phi)$ is hence given by $(N^\Phi, p^\Phi) = (N^\Phi_{\bowtie} \cup N^\Phi_l \cup N^\Phi_r \cup RPS_l \cup RPS_r, p^\Phi)$, with $(N^\Phi_l, p^\Phi_l) = localize((N_l, p_l))$, where $(N_l, p_l)$ is the RETE subnet under the left dependency of $p^\Phi$, and $(N^\Phi_r, p^\Phi_r) = localize((N_r, p_r))$, where $(N_r, p_r)$ is the RETE subnet under the right dependency of $p^\Phi$. We thus also know that $\chi(N^\Phi) = \chi(N^\Phi_l) \cup \chi(N^\Phi_r)$.

We can thus again distinguish two cases:

Case 1: $x \in \chi(N^\Phi_l)$. Since $N_l$ must have a join tree height of $h' \leq h$ and because $(N^\Phi, p^\Phi)$ and hence also $(X^\Phi, p_X^\Phi)$ is a modular extension of $(N^\Phi_l, p^\Phi_l)$, by the induction hypothesis, we know that $\forall m_l \in \allmatches{Q_l}{H}: m_l(v) = m'(v) \Rightarrow \exists (m_l, \phi_l) \in \mathcal{C}(p_l^\Phi) : \phi_l \geq \psi$, where $Q_l$ is the query subgraph associated with $p_l$.

For any such $(m_l, \phi_l) \in \mathcal{C}(p_l^\Phi)$, by the construction of $N^\Phi$, the semantics of the RETE nodes in $RPS_l$, $\mathcal{C}$ being consistent for all nodes in $V^{N^\Phi}$ and hence also for all nodes in $V^{N^\Phi_l}$, and since $\phi_l \geq \psi > h + 1$, it then follows that $(m_l|_{Q_{v'}}, h + 1) \in \mathcal{C}([\phi := h + 1]^\Phi_l)$, with $[\phi := h + 1]^\Phi_l$ the marking assignment node in $RPS_l$ and $Q_{v'} = (\{v'\}, \emptyset, \emptyset, \emptyset)$ the query subgraph associated with $[\phi := h + 1]^\Phi_l$.

Because $[\phi := h + 1]^\Phi_l$ is a dependency of some extension point $x_r \in \chi(N^\Phi_r)$ and $x_r$ is a marking-sensitive union node, it follows that there must be some tuple $(m_l|_{Q_{v'}}, \psi_r) \in \mathcal{C}(x_r)$, with $\psi_r \geq h + 1$. By the induction hypothesis and since $v' \in V^{Q_l} \cap V^{Q_r}$, where $Q_r$ is the query subgraph associated with $p_r$, it then follows that $\forall m_r \in \allmatches{Q_r}{H}: m_r(v') = m_l(v') \Rightarrow \exists (m_r, \phi_r) \in \mathcal{C}(p^\Phi_r) : \phi_r \geq \psi_r$.

Because $v' \in Q_\cap$ with $p^\Phi = [\bowtie_{Q_\cap}]^\Phi$, it must then hold that $\forall m_r \in \allmatches{Q_r}{H} : m_r|_{Q_\cap} = m_l|_{Q_\cap} \Rightarrow \exists (m_r, \phi_r) \in \mathcal{C}(p^\Phi_r) : \phi_r \geq \psi_r$. By the semantics of the marking-sensitive join node and Lemma \ref{lem:complementarity_completeness}, and because $\phi_l \geq \psi$, it then follows that $\forall m \in \allmatches{Q}{H} : m|_{Q_l} = m_l \Rightarrow \exists (m, \phi) \in \mathcal{C}(p^\Phi) : \phi \geq \psi$. Since $Q_l \subseteq Q$, which entails that $\forall m \in \allmatches{Q}{H}: \exists m_l \in \allmatches{Q_l}{H} : m|_{Q_l} = m_l$, and because $\forall m_l \in \allmatches{Q_l}{H}: m_l(v) = m'(v) \Rightarrow \exists (m_l, \phi_l) \in \mathcal{C}(p_l^\Phi) : \phi_l \geq \psi$, it must finally follow that $\forall m \in \allmatches{Q}{H}: m(v) = m'(v) \Rightarrow \exists (m, \phi) \in \mathcal{C}(p^\Phi) : \phi \geq \psi$.

Case 2: $x \in \chi(N^\Phi_r)$. The proof works analogously to the proof for Case 1.

Thus, the induction step holds.

From the correctness of the base case and the induction step then follows the correctness of the lemma.
\end{proof}


\begin{lemma} \label{lem:complementarity_completeness}
Let $n^\Phi = [\bowtie_{Q_\cap}]^\Phi \in V^{N^\Phi}$ be a marking-sensitive join with associated subgraph $Q$ and dependencies $n^\Phi_l$ and $n^\Phi_r$ with associated subgraphs $Q_l$ respectively $Q_r$, $H$ a graph, $H_p \subseteq H$, $\mathcal{C}$ a configuration, and $(m_l, \phi_l) \in \mathcal{C}(n^\Phi_l)$. If $\forall m_r \in \allmatches{Q_r}{H} : m_r|_{Q_\cap} = m_l|_{Q_\cap} \Rightarrow \exists \phi_r \in \overline{\mathbb{N}}: (m_r, \phi_r) \in \mathcal{C}(n^\Phi_r)$, it then holds that $\forall m \in \allmatches{Q}{H} : m|_{Q_l} = m_l \Rightarrow \exists (m, \phi) \in \resultslocal{n^\Phi}{N^\Phi}{H}{H_p}{\mathcal{C}}: \phi \geq \phi_l$.
\end{lemma}

\begin{proof}
Since $Q = Q_l \cup Q_r$, any match $m \in \allmatches{Q}{H}$ with $m|_{Q_l} = m_l$ can be represented as a union of two matches $m = m_l \cup m_r$, with $m_r \in \allmatches{Q_r}{H}$. For $m_r$, it must then hold that $m_l|_{Q_\cap} = m_r|_{Q_\cap}$.

By assumption, there exist a tuple $(m_l, \phi_l) \in \mathcal{C}(n^\Phi_l)$ and a tuple $(m_r, \phi_r) \in \mathcal{C}(n^\Phi_r)$. From the definition of the marking-sensitive join node, it follows that $(m, max(\phi_l, \phi_r)) \in \resultslocal{n^\Phi}{N^\Phi}{H}{H_p}{\mathcal{C}}$. Thus, it must hold that $\forall m \in \allmatches{Q}{H} : m|_{Q_l} = m_l \Rightarrow \exists (m, \phi) \in \resultslocal{n^\Phi}{N^\Phi}{H}{H_p}{\mathcal{C}}: \phi \geq \phi_l$.
\end{proof}


\begin{lemma} \label{lem:join_complementarity}
Let $H$ be a graph, $H_p \subseteq H$, $(N, p)$ a well-formed RETE net with join tree height $h > 0$, and $\mathcal{C}$ a consistent configuration for the localized RETE net $(N^\Phi, [\bowtie_{Q_\cap}]^\Phi) = localize((N, p))$. Furthermore, let $(N_l, p_l)$ and $(N_r, p_r)$ be the RETE subnets under the dependencies of the join node $p$, $Q_l$ and $Q_r$ the associated query subgraphs of $p_l$ respectively $p_r$, and $(N^\Phi_l, p^\Phi_l) = localize((N_l, p_l))$ and $(N^\Phi_r, p^\Phi_r) = localize((N_r, p_r))$ the associated RETE subnets in $(N^\Phi, [\bowtie_{Q_\cap}]^\Phi)$. Finally, let $(m_l, \phi_l) \in \mathcal{C}(p^\Phi_l)$ such that $\phi_l > h$. It then holds that $\forall m_r \in \allmatches{Q_r}{H} : m_l|_{Q_\cap} = m_r|_{Q_\cap} \Rightarrow \exists (m_r, \phi_r) \in \mathcal{C}(p^\Phi_r): \phi_r \geq h$.
\end{lemma}

\begin{proof}
By the construction of $(N^\Phi, [\bowtie_{Q_\cap}]^\Phi)$, there must be an marking-sensitive union node extension point $[\cup]_v^\Phi \in \chi(N^\Phi_r)$ that has a dependency on the marking assignment node $[\phi := h]^\Phi_l$ in the request projection structure $RPS_l$ in $N^\Phi$. According to the semantics of the RETE nodes in $RPS_l$, it must also hold that $\forall m_v \in \allmatches{Q_v}{H} : m_v(v) = m_l(v) \Rightarrow (m_v, h) \in \mathcal{C}([\phi := h]^\Phi_l)$, where $Q_v = (\{v\}, \emptyset, \emptyset, \emptyset)$ is the query subgraph associated with $[\phi := h]^\Phi_l$. By the semantics of the marking-sensitive union node, it must then also hold that $\forall m_v \in \allmatches{Q_v}{H} : m_v(v) = m_l(v) \Rightarrow \exists \psi \in \overline{\mathbb{N}}: (m_v, \phi_\cup) \in \mathcal{C}([\cup]_v^\Phi) \wedge \phi_\cup \geq h$.

Because $h' < h$, where $h'$ is the height of the join tree $N_r$, by Lemma \ref{lem:completeness_extension_points}, we then know that $\forall m_r \in \allmatches{Q_r}{H}: m_r(v) = m_l(v) \Rightarrow \exists (m_r, \phi_r) \in \mathcal{C}(p^\Phi) : \phi_r \geq h$. Since $Q_v \subseteq Q_\cap$, it then follows that $\forall m_r \in \allmatches{Q_r}{H} : m_l|_{Q_\cap} = m_r|_{Q_\cap} \Rightarrow \exists (m_r, \phi_r) \in \mathcal{C}(p^\Phi_r) : \phi_r \geq h$.
\end{proof}


\begin{lemma} \label{lem:localized_rete_max_marking}
Let $H$ a graph, $H_p \subseteq H$, $(N, p)$ a well-formed RETE net with join-tree height $h$, $\psi \in \overline{\mathbb{N}}$ such that $\psi > h$, $(N^\Phi, p^\Phi) = localize((N, p))$, $(X^\Phi, p_X^\Phi)$ a modular extension of $(N^\Phi, p^\Phi)$, $\mathcal{C}_0$ a configuration for $(X^\Phi, p^\Phi)$ that is consistent for all nodes in $V^{N^\Phi}$ with respect to $H$ and $H_p$, and $\mathcal{C}_1$ a configuration with $\forall n \in V^{N^\Phi}: \mathcal{C}_1(n) = \mathcal{C}_0(n)$ and $\forall n_X \in V^{X^\Phi}: (\exists e \in E^{X^\Phi} : t^{X^\Phi} = n_X \wedge s^{X^\Phi} \in V^{N^\Phi}) \Rightarrow \forall (m, \phi) \in (\mathcal{C}_1(n_X) \setminus \mathcal{C}_0(n_X)) \cup (\mathcal{C}_0(n_X) \setminus \mathcal{C}_1(n_X)) : \phi \leq \psi$. Under the assumption that the execution order $order(localize((N', p')))$ is robust under modular extension for any well-formed RETE net $(N', p')$ with join tree height $h' < h$, it then holds for the configuration $\mathcal{C}_2 = execute(order(N^\Phi, p^\Phi), X^\Phi, H, H_p, \mathcal{C}_1)$ that $\forall (m, \phi) \in (\mathcal{C}_2(p^\Phi) \setminus \mathcal{C}_1(p^\Phi)) \cup (\mathcal{C}_1(p^\Phi) \setminus \mathcal{C}_2(p^\Phi)) : \phi \leq \psi$.
\end{lemma}

\begin{proof}
We show the correctness of the lemma by induction over the height $h$ of the join tree $N$.

In the base case of $h$ = 0, $N$ consists of a single edge input $p$ and the localized RETE net is thus given by $(N^\Phi, p^\Phi) = (LNS(p), [\cup]^\Phi)$, with $LNS(p)$ consisting of nodes $[v]^\Phi$, $[w]^\Phi$, $[\cup]^\Phi_v$, $[\cup]^\Phi_w$, $[v \rightarrow_n w]^\Phi$, $[w \leftarrow_n v]^\Phi$, and $[\cup]^\Phi$. Without loss of generality, we assume the execution order $order(LNS(p)) = [v]^\Phi, [w]^\Phi, [\cup]^\Phi_v, [\cup]^\Phi_w, [v \rightarrow_n w]^\Phi, [w \leftarrow_n v]^\Phi, [\cup]^\Phi$.

By the semantics of the marking-sensitive vertex input node, we know that $\mathcal{C}_{1.1} = execute(O_1, X^\Phi, H, H_p, \mathcal{C}_1) = \mathcal{C}_1$, with $O_1 = [v]^\Phi, [w]^\Phi$.

By Lemma \ref{lem:union_max_marking} and the assumption regarding the difference between current result sets in $\mathcal{C}_0$ and $\mathcal{C}_1$ for dependencies from $LNS(p)$ into $X^\Phi$, it follows for $\mathcal{C}_{1.2} = execute(O_2, X^\Phi, H, H_p, \mathcal{C}_{1.1})$ and $O_2 = [\cup]^\Phi_v, [\cup]^\Phi_w$ that $\forall (m, \phi) \in (\mathcal{C}_{1.2}([\cup]^\Phi_v) \setminus \mathcal{C}_{1}([\cup]^\Phi_v)) \cup (\mathcal{C}_1([\cup]^\Phi_v) \setminus \mathcal{C}_{1.2}([\cup]^\Phi_v)) : \phi \leq \psi$, which holds analogously for $[\cup]^\Phi_w$.

By Lemmata \ref{lem:forward_navigation_max_marking} and \ref{lem:backward_navigation_max_marking}, the same holds for $\mathcal{C}_{1.3} = execute(O_3, X^\Phi, H, H_p, \mathcal{C}_{1.2})$ and the nodes $[v \rightarrow_n w]^\Phi$ and $[w \leftarrow_n v]^\Phi$, with $O_3 = [v \rightarrow_n w]^\Phi, [w \leftarrow_n v]^\Phi$.

Finally, by Lemma \ref{lem:union_max_marking} and the assumption regarding the difference between current result sets in $\mathcal{C}_0$ and $\mathcal{C}_1$ for dependencies from $LNS(p)$ into $X^\Phi$ then follows the satisfaction of the base case for $\mathcal{C}_2 = execute([\cup]^\Phi, X^\Phi, H, H_p, \mathcal{C}_{1.3})$.

For the induction step, we show that if the lemma holds for all join trees of height less than or equal to $h$, it follows that the lemma holds for any join tree $N$ of height $h + 1$.

Since $h + 1 \geq 1$, $p$ is a join node with dependencies $p_l$ and $p_r$ and the localized RETE net is thus given by $(N^\Phi, p^\Phi) = (N^\Phi_{\bowtie} \cup N^\Phi_l \cup N^\Phi_r \cup RPS_l \cup RPS_r, p^\Phi)$, with $(N^\Phi_l, p^\Phi_l) = localize((N_l, p_l))$ and $(N^\Phi_r, p^\Phi_r) = localize((N_r, p_r))$. Consequently, the execution order of $N^\Phi$ is given by $order(N^\Phi) = order(RPS_r) \circ order(N^\Phi_l) \circ order(RPS_l) \circ order(N^\Phi_r) \circ order(RPS_r) \circ order(N^\Phi_l) \circ N^\Phi_{\bowtie}$.

By definition, the marking filter nodes, marking-sensitive projection nodes, and marking assignment nodes in $RPS_l$ and $RPS_r$ must have exactly one dependency and the marking-sensitive join node in $N^\Phi_{\bowtie}$ must have exactly two dependencies. By the construction of $N^\Phi$, in the modular extension $X^\Phi$, all nodes in $RPS_l$, $RPS_r$ and $N^\Phi_{\bowtie}$ may thus only depend on nodes in $N^\Phi$.

Since $\forall n \in V^{N^\Phi}: \mathcal{C}_1(n) = \mathcal{C}_0(n)$ and because $\mathcal{C}_0$ is consistent for all nodes in $V^{N^\Phi}$, $\mathcal{C}_1$ is then consistent for all nodes in $V^{RPS_r}$. Thus, $\mathcal{C}_{1.1} = execute(order(RPS_r), X^\Phi, H, H_p, \mathcal{C}_1) = \mathcal{C_1}$.

Since $(N^\Phi, p^\Phi)$ is a modular extension of $(N^\Phi_l, p^\Phi_l)$ and $(N^\Phi_r, p^\Phi_r)$, $(X^\Phi, p_X^\Phi)$ must be a modular extension of $(N^\Phi_l, p^\Phi_l)$ and $(N^\Phi_r, p^\Phi_r)$. Furthermore, since $\forall n \in V^{N^\Phi}: \mathcal{C}_1(n) = \mathcal{C}_0(n)$ and $\forall n_X \in V^{X^\Phi}: (\exists e \in E^{X^\Phi} : t^{X^\Phi} = n_X \wedge s^{X^\Phi} \in V^{N^\Phi}) \Rightarrow \forall (m, \phi) \in (\mathcal{C}_1(n_X) \setminus \mathcal{C}_0(n_X)) \cup (\mathcal{C}_0(n_X) \setminus \mathcal{C}_1(n_X)) : \phi \leq \psi$, it also holds for $N^\Phi_l$ and analogously for $N^\Phi_r$ that $\forall n \in V^{N^\Phi_l}: \mathcal{C}_1(n) = \mathcal{C}_0(n)$ and $\forall n_X \in V^{X^\Phi}: (\exists e \in E^{X^\Phi} : t^{X^\Phi} = n_X \wedge s^{X^\Phi} \in V^{N^\Phi_l}) \Rightarrow \forall (m, \phi) \in (\mathcal{C}_1(n_X) \setminus \mathcal{C}_0(n_X)) \cup (\mathcal{C}_0(n_X) \setminus \mathcal{C}_1(n_X)) : \phi \leq \psi$.

Thus, for configuration $\mathcal{C}_{1.2} = execute(order(N^\Phi_l), X^\Phi, H, H_p, \mathcal{C}_{1.1})$, by the induction hypothesis it holds that $\forall (m, \phi) \in (\mathcal{C}_{1.2}(p^\Phi_l) \setminus \mathcal{C}_0(p^\Phi_l)) \cup (\mathcal{C}_0(p^\Phi_l) \setminus \mathcal{C}_{1.2}(p^\Phi_l)) : \phi \leq \psi$. Furthermore, since $N_l$ must have join tree height $h' < h + 1$ and by the entailed assumption of robustness of $order(N^\Phi_l)$ under modular extension, $\mathcal{C}_{1.2}$ must be consistent for all nodes in $V^{N^\Phi_l}$.

Since $\psi \geq h + 1$, by the semantics of the marking assignment node $[\phi := h + 1]_l^\Phi$, it follows that $\forall (m, \phi) \in (\mathcal{C}_{1.3}([\phi := h + 1]_l^\Phi) \setminus \mathcal{C}_0([\phi := h + 1]_l^\Phi)) \cup (\mathcal{C}_0([\phi := h + 1]_l^\Phi) \setminus \mathcal{C}_{1.3}([\phi := h + 1]_l^\Phi)) : \phi \leq \psi$.

By the construction of $N^\Phi$, it then holds that $\forall n \in V^{N^\Phi_r}: \mathcal{C}_{1.3}(n) = \mathcal{C}_0(n)$ and $\forall n_X \in V^{X^\Phi}: (\exists e \in E^{X^\Phi} : t^{X^\Phi} = n_X \wedge s^{X^\Phi} \in V^{N^\Phi_r}) \Rightarrow \forall (m, \phi) \in (\mathcal{C}_{1.3}(n_X) \setminus \mathcal{C}_0(n_X)) \cup (\mathcal{C}_0(n_X) \setminus \mathcal{C}_{1.3}(n_X)) : \phi \leq \psi$. Hence, by the induction hypothesis, it holds for $\mathcal{C}_{1.4} = execute(order(N^\Phi_r), X^\Phi, H, H_p, \mathcal{C}_{1.3})$ that $\forall (m, \phi) \in (\mathcal{C}_{1.4}(p^\Phi_r) \setminus \mathcal{C}_0(p^\Phi_r)) \cup (\mathcal{C}_0(p^\Phi_r) \setminus \mathcal{C}_{1.4}(p^\Phi_r)) : \phi \leq \psi$. Furthermore, since $N_r$ must have join tree height $h' < h + 1$ and by the entailed assumption of robustness of $order(N^\Phi_r)$ under modular extension, $\mathcal{C}_{1.4}$ must be consistent for all nodes in $V^{N^\Phi_r}$.

By the semantics of the marking assignment node $[\phi := h + 1]_r^\Phi$ in $RPS_r$, for $\mathcal{C}_{1.5} = execute(order(RPS_r), X^\Phi, H, H_p, \mathcal{C}_{1.4})$ it must then hold that $\forall (m, \phi) \in (\mathcal{C}_{1.5}([\phi := h + 1]_r^\Phi) \setminus \mathcal{C}_0([\phi := h + 1]_r^\Phi)) \cup (\mathcal{C}_{1.4}([\phi := h + 1]_r^\Phi) \setminus \mathcal{C}_{1.5}([\phi := h + 1]_r^\Phi)) : \phi = h + 1$.

Since $V^{N^\Phi_l} \cap (V^{N^\Phi_r} \cup V^{RPS_l} \cup V^{RPS_r}) = \emptyset$, it then follows that $\forall n \in V^{N^\Phi_l}: \mathcal{C}_{1.5}(n) = \mathcal{C}_{1.2}(n)$. Because no node in $V^{N^\Phi_l}$ may depend on any node in the sequence $order(RPS_l) \circ order(N^\Phi_r) \circ order(RPS_r)$ other than $[\phi := h]_r^\Phi$, it also follows that $\forall n_X \in V^{X^\Phi}: (\exists e \in E^{X^\Phi} : t^{X^\Phi} = n_X \wedge s^{X^\Phi} \in V^{N^\Phi_l}) \Rightarrow \forall (m, \phi) \in (\mathcal{C}_{1.5}(n_X) \setminus \mathcal{C}_{1.2}(n_X)) \cup (\mathcal{C}_{1.2}(n_X) \setminus \mathcal{C}_{1.5}(n_X)) : \phi = h + 1$. Because $\mathcal{C}_{1.2}$ is consistent for all nodes in $V^{N^\Phi_l}$ and $N_l$ must have a join tree height $h' \leq h + 1$, for $\mathcal{C}_{1.6} = execute(order(N^\Phi_l), X^\Phi, H, H_p, \mathcal{C}_{1.5})$, it then follows by the induction hypothesis that $\forall (m, \phi) \in (\mathcal{C}_{1.6}(p^\Phi_l) \setminus \mathcal{C}_{1.2}(p^\Phi_l)) \cup (\mathcal{C}_{1.2}(p^\Phi_l) \setminus \mathcal{C}_{1.6}(p^\Phi_l)) : \phi \leq h + 1$.

Thus, it holds that $\forall (m, \phi) \in (\mathcal{C}_{1.6}(p^\Phi_l) \setminus \mathcal{C}_0(p^\Phi_l)) \cup (\mathcal{C}_0(p^\Phi_l) \setminus \mathcal{C}_{1.6}(p^\Phi_l)) : \phi \leq \psi$. Furthermore, since $p^\Phi_r \notin (V^{N^\Phi_l} \cup V^{RPS_r})$, it also holds that $\forall (m, \phi) \in (\mathcal{C}_{1.6}(p^\Phi_r) \setminus \mathcal{C}_0(p^\Phi_r)) \cup (\mathcal{C}_0(p^\Phi_r) \setminus \mathcal{C}_{1.6}(p^\Phi_r)) : \phi \leq \psi$.

By the assumption regarding the robustness under modular extension of $order(N^\Phi_l)$, $\mathcal{C}_{1.6}$ must be consistent with all nodes in $V^{N^\Phi_l}$. Since $V^{N^\Phi_l}$, $V^{N^\Phi_r}$, $V^{RPS_l}$, and $V^{RPS_r}$ are pairwise disjoint and since the modular extension $X^\Phi$ cannot introduce additional dependencies between nodes in $N^\Phi$, it follows by the construction of $N^\Phi$ that $\mathcal{C}_{1.6}$ must also be consistent for all nodes in $V^{N^\Phi_r}$ and $V^{RPS_r}$. By the semantics of $[\phi > h + 1]_l^\Phi$ and because $\forall (m, \phi) \in (\mathcal{C}_{1.6}(p^\Phi_l) \setminus \mathcal{C}_{1.3}(p^\Phi_l)) \cup (\mathcal{C}_{1.3}(p^\Phi_l) \setminus \mathcal{C}_{1.6}(p^\Phi_l)) : \phi \leq h + 1$, we also know that $\mathcal{C}_{1.6}$ is consistent for all nodes in $V^{RPS_l}$.

Because $V^{N^\Phi_{\bowtie}} = \{p^\Phi\}$ with $p^\Phi$ a marking sensitive join with dependencies $p^\Phi_l$ and $p^\Phi_r$, by Lemma \ref{lem:join_max_marking} it follows for $\mathcal{C}_2 = execute(order(N^\Phi_{\bowtie}), X^\Phi, H, H_p, \mathcal{C}_{1.6})$ that $\forall (m, \phi) \in (\mathcal{C}_2(p^\Phi) \setminus \mathcal{C}_1(p^\Phi)) \cup (\mathcal{C}_1(p^\Phi) \setminus \mathcal{C}_2(p^\Phi)) : \phi \leq \psi$.

From the correctness of the base case and the induction step then follows the correctness of the lemma.
\end{proof}


\begin{lemma} \label{lem:union_max_marking}
Let $n^\Phi$ be a marking-sensitive union node with dependencies $N^\Phi_\alpha$ in some marking-sensitive RETE net $(N^\Phi, p^\Phi)$. Let $\mathcal{C}_0$ and $\mathcal{C}$ furthermore be configurations for $(N^\Phi, p^\Phi)$ such that $\mathcal{C}_0$ is consistent for $n^\Phi$ and it holds for all $n^\Phi_\alpha \in N^\Phi_\alpha$ that $\forall (m, \phi) \in \mathcal{C}_0(n^\Phi_\alpha) \setminus \mathcal{C}(n^\Phi_\alpha) : \phi \leq h$ and $\forall (m, \phi) \in \mathcal{C}(n^\Phi_\alpha) \setminus \mathcal{C}_0(n^\Phi_\alpha) : \phi \leq h$. It then holds for host graph $H$ and partial host graph $H_p$ that $\forall (m, \phi) \in \mathcal{C}_0(n^\Phi) \setminus \resultslocal{n^\Phi}{N^\Phi}{H}{H_p}{\mathcal{C}} : \phi \leq h$ and $\forall (m, \phi) \in \resultslocal{n^\Phi}{N^\Phi}{H}{H_p}{\mathcal{C}} \setminus \mathcal{C}_0(n^\Phi) : \phi \leq h$.
\end{lemma}

\begin{proof}
Let $(m, \phi) \in \resultslocal{n^\Phi}{N^\Phi}{H}{H_p}{\mathcal{C}} \setminus \mathcal{C}_0(n^\Phi)$. We can then distinguish two cases:

Case 1: It may hold that $\nexists (m', \phi') \in \mathcal{C}_0(n^\Phi) : m = m'$. In this case, by the semantics of the marking-sensitive union node, there must be some $n^\Phi_\alpha \in N^\Phi_\alpha$ such that $(m, \phi) \in \mathcal{C}(n^\Phi_\alpha) \setminus \mathcal{C}_0(n^\Phi_\alpha)$. Since $\forall (m, \phi) \in \mathcal{C}(n^\Phi_\alpha) \setminus \mathcal{C}_0(n^\Phi_\alpha) : \phi \leq h$, it thus holds that $\phi \leq h$.

Case 2: Otherwise, there must be a tuple $(m', \phi') \in \mathcal{C}_0(n^\Phi) : m = m'$. If $\phi' < \phi$, there again must be some $n^\Phi_\alpha \in N^\Phi_\alpha$ such that $(m, \phi) \in \mathcal{C}(n^\Phi_\alpha) \setminus \mathcal{C}_0(n^\Phi_\alpha)$, and thus $\phi \leq h$. If $\phi' > \phi$, there must instead be some $n^\Phi_\alpha \in N^\Phi_\alpha$ such that $(m, \phi') \in \mathcal{C}_0(n^\Phi_\alpha) \setminus \mathcal{C}(n^\Phi_\alpha)$. Because $\forall (m, \phi') \in \mathcal{C}_0(n^\Phi_\alpha) \setminus \mathcal{C}(n^\Phi_\alpha) : \phi' \leq h$, it then also follows that $\phi \leq h$.

It thus follows that $\phi \leq h$.


Let $(m, \phi) \in \mathcal{C}_0(n^\Phi) \setminus \resultslocal{n^\Phi}{N^\Phi}{H}{H_p}{\mathcal{C}}$. We can then distinguish two cases:

Case 1: It may hold that $\nexists (m', \phi') \in \resultslocal{n^\Phi}{N^\Phi}{H}{H_p}{\mathcal{C}} : m = m'$. In this case, by the semantics of the marking-sensitive union node, there must be some $n^\Phi_\alpha \in N^\Phi_\alpha$ such that $(m, \phi) \in \mathcal{C}_0(n^\Phi_\alpha) \setminus \mathcal{C}(n^\Phi_\alpha)$. Since $\forall (m, \phi) \in \mathcal{C}_0(n^\Phi_\alpha) \setminus \mathcal{C}(n^\Phi_\alpha) : \phi \leq h$, it thus holds that $\phi \leq h$.

Case 2: Otherwise, there must be a tuple $(m', \phi') \in \resultslocal{n^\Phi}{N^\Phi}{H}{H_p}{\mathcal{C}} : m = m'$. If $\phi' < \phi$, there again must be some $n^\Phi_\alpha \in N^\Phi_\alpha$ such that $(m, \phi) \in \mathcal{C}_0(n^\Phi_\alpha) \setminus \mathcal{C}(n^\Phi_\alpha)$, and thus $\phi \leq h$. If $\phi' > \phi$, there must instead be some $n^\Phi_\alpha \in N^\Phi_\alpha$ such that $(m, \phi') \in \mathcal{C}(n^\Phi_\alpha) \setminus \mathcal{C}_0(n^\Phi_\alpha)$. Because $\forall (m, \phi') \in \mathcal{C}(n^\Phi_\alpha) \setminus \mathcal{C}_0(n^\Phi_\alpha) : \phi' \leq h$, it then also follows that $\phi \leq h$.

It thus follows that $\phi \leq h$.
\end{proof}

\begin{lemma} \label{lem:projection_max_marking}
Let $n^\Phi$ be a marking-sensitive projection node with dependency $n^\Phi_\alpha$ in some marking-sensitive RETE net $(N^\Phi, p^\Phi)$. Let $\mathcal{C}_0$ and $\mathcal{C}$ furthermore be configurations for $(N^\Phi, p^\Phi)$ such that $\mathcal{C}_0$ is consistent for $n^\Phi$ and it holds that $\forall (m, \phi) \in \mathcal{C}_0(n^\Phi_\alpha) \setminus \mathcal{C}(n^\Phi_\alpha) : \phi \leq h$ and $\forall (m, \phi) \in \mathcal{C}(n^\Phi_\alpha) \setminus \mathcal{C}_0(n^\Phi_\alpha) : \phi \leq h$. It then holds for host graph $H$ and partial host graph $H_p$ that $\forall (m, \phi) \in \mathcal{C}_0(n^\Phi) \setminus \resultslocal{n^\Phi}{N^\Phi}{H}{H_p}{\mathcal{C}} : \phi \leq h$ and $\forall (m, \phi) \in \resultslocal{n^\Phi}{N^\Phi}{H}{H_p}{\mathcal{C}} \setminus \mathcal{C}_0(n^\Phi) : \phi \leq h$.
\end{lemma}

\begin{proof}
Let $(m, \phi) \in \resultslocal{n^\Phi}{N^\Phi}{H}{H_p}{\mathcal{C}} \setminus \mathcal{C}_0(n^\Phi)$. We can then distinguish two cases:

Case 1: It may hold that $\nexists (m', \phi') \in \mathcal{C}_0(n^\Phi) : m = m'$. In this case, by the semantics of the marking-sensitive projection node, there must be some match $m''$ such that $(m'', \phi) \in \mathcal{C}(n^\Phi_\alpha) \setminus \mathcal{C}_0(n^\Phi_\alpha)$. Since $\forall (m'', \phi) \in \mathcal{C}(n^\Phi_\alpha) \setminus \mathcal{C}_0(n^\Phi_\alpha) : \phi \leq h$, it thus holds that $\phi \leq h$.

Case 2: Otherwise, there must be a tuple $(m', \phi') \in \mathcal{C}_0(n^\Phi) : m = m'$. If $\phi' < \phi$, there again must be some match $m''$ such that $(m'', \phi) \in \mathcal{C}(n^\Phi_\alpha) \setminus \mathcal{C}_0(n^\Phi_\alpha)$, and thus $\phi \leq h$. If $\phi' > \phi$, there must instead be a match $m''$ such that $(m'', \phi') \in \mathcal{C}_0(n^\Phi_\alpha) \setminus \mathcal{C}(n^\Phi_\alpha)$. Because $\forall (m'', \phi') \in \mathcal{C}_0(n^\Phi_\alpha) \setminus \mathcal{C}(n^\Phi_\alpha) : \phi' \leq h$, it then also follows that $\phi \leq h$.

It thus follows that $\phi \leq h$.


Let $(m, \phi) \in \mathcal{C}_0(n^\Phi) \setminus \resultslocal{n^\Phi}{N^\Phi}{H}{H_p}{\mathcal{C}}$. We can then distinguish two cases:

Case 1: It may hold that $\nexists (m', \phi') \in \resultslocal{n^\Phi}{N^\Phi}{H}{H_p}{\mathcal{C}} : m = m'$. In this case, by the semantics of the marking-sensitive projection node, there must be some match $m''$ such that $(m'', \phi) \in \mathcal{C}_0(n^\Phi_\alpha) \setminus \mathcal{C}(n^\Phi_\alpha)$. Since $\forall (m'', \phi) \in \mathcal{C}_0(n^\Phi_\alpha) \setminus \mathcal{C}(n^\Phi_\alpha) : \phi \leq h$, it thus holds that $\phi \leq h$.

Case 2: Otherwise, there must be a tuple $(m', \phi') \in \resultslocal{n^\Phi}{N^\Phi}{H}{H_p}{\mathcal{C}} : m = m'$. If $\phi' < \phi$, there again must be some  match $m''$ such that $(m'', \phi) \in \mathcal{C}_0(n^\Phi_\alpha) \setminus \mathcal{C}(n^\Phi_\alpha)$, and thus $\phi \leq h$. If $\phi' > \phi$, there must instead be a match $m''$ such that $(m, \phi') \in \mathcal{C}(n^\Phi_\alpha) \setminus \mathcal{C}_0(n^\Phi_\alpha)$. Because $\forall (m, \phi') \in \mathcal{C}(n^\Phi_\alpha) \setminus \mathcal{C}_0(n^\Phi_\alpha) : \phi' \leq h$, it then also follows that $\phi \leq h$.

It thus follows that $\phi \leq h$.
\end{proof}

\begin{lemma} \label{lem:forward_navigation_max_marking}
Let $n^\Phi$ be a forward navigation node with dependency $n^\Phi_\alpha$ in some marking-sensitive RETE net $(N^\Phi, p^\Phi)$. Let $\mathcal{C}_0$ and $\mathcal{C}$ furthermore be configurations for $(N^\Phi, p^\Phi)$ such that $\mathcal{C}_0$ is consistent for $n^\Phi$ and it holds that $\forall (m, \phi) \in \mathcal{C}_0(n^\Phi_\alpha) \setminus \mathcal{C}(n^\Phi_\alpha) : \phi \leq h$ and $\forall (m, \phi) \in \mathcal{C}(n^\Phi_\alpha) \setminus \mathcal{C}_0(n^\Phi_\alpha) : \phi \leq h$. It then holds for host graph $H$ and partial host graph $H_p$ that $\forall (m, \phi) \in \mathcal{C}_0(n^\Phi) \setminus \resultslocal{n^\Phi}{N^\Phi}{H}{H_p}{\mathcal{C}} : \phi \leq h$ and $\forall (m, \phi) \in \resultslocal{n^\Phi}{N^\Phi}{H}{H_p}{\mathcal{C}} \setminus \mathcal{C}_0(n^\Phi) : \phi \leq h$.
\end{lemma}

\begin{proof}
Let $(m, \phi) \in \resultslocal{n^\Phi}{N^\Phi}{H}{H_p}{\mathcal{C}} \setminus \mathcal{C}_0(n^\Phi)$. It must then hold that $\exists (m', \phi) \in \mathcal{C}(n^\Phi_\alpha) \setminus \mathcal{C}_0(n^\Phi_\alpha)$. Since $\forall (m', \phi) \in \mathcal{C}(n^\Phi_\alpha) \setminus \mathcal{C}_0(n^\Phi_\alpha) : \phi \leq h$, it follows that $\phi \leq h$.


Let $(m, \phi) \in \mathcal{C}_0(n^\Phi) \setminus \resultslocal{n^\Phi}{N^\Phi}{H}{H_p}{\mathcal{C}}$. It must then hold that $\exists (m', \phi) \in \mathcal{C}_0(n^\Phi_\alpha) \setminus \mathcal{C}(n^\Phi_\alpha)$. Since $\forall (m', \phi) \in \mathcal{C}_0(n^\Phi_\alpha) \setminus \mathcal{C}(n^\Phi_\alpha) : \phi \leq h$, it follows that $\phi \leq h$.
\end{proof}

\begin{lemma} \label{lem:backward_navigation_max_marking}
Let $n^\Phi$ be a backward navigation node with dependency $n^\Phi_\alpha$ in some marking-sensitive RETE net $(N^\Phi, p^\Phi)$. Let $\mathcal{C}_0$ and $\mathcal{C}$ furthermore be configurations for $(N^\Phi, p^\Phi)$ such that $\mathcal{C}_0$ is consistent for $n^\Phi$ and it holds that $\forall (m, \phi) \in \mathcal{C}_0(n^\Phi_\alpha) \setminus \mathcal{C}(n^\Phi_\alpha) : \phi \leq h$ and $\forall (m, \phi) \in \mathcal{C}(n^\Phi_\alpha) \setminus \mathcal{C}_0(n^\Phi_\alpha) : \phi \leq h$. It then holds for host graph $H$ and partial host graph $H_p$ that $\forall (m, \phi) \in \mathcal{C}_0(n^\Phi) \setminus \resultslocal{n^\Phi}{N^\Phi}{H}{H_p}{\mathcal{C}} : \phi \leq h$ and $\forall (m, \phi) \in \resultslocal{n^\Phi}{N^\Phi}{H}{H_p}{\mathcal{C}} \setminus \mathcal{C}_0(n^\Phi) : \phi \leq h$.
\end{lemma}

\begin{proof}
Let $(m, \phi) \in \resultslocal{n^\Phi}{N^\Phi}{H}{H_p}{\mathcal{C}} \setminus \mathcal{C}_0(n^\Phi)$. It must then hold that $\exists (m', \phi) \in \mathcal{C}(n^\Phi_\alpha) \setminus \mathcal{C}_0(n^\Phi_\alpha)$. Since $\forall (m', \phi) \in \mathcal{C}(n^\Phi_\alpha) \setminus \mathcal{C}_0(n^\Phi_\alpha) : \phi \leq h$, it follows that $\phi \leq h$.


Let $(m, \phi) \in \mathcal{C}_0(n^\Phi) \setminus \resultslocal{n^\Phi}{N^\Phi}{H}{H_p}{\mathcal{C}}$. It must then hold that $\exists (m', \phi) \in \mathcal{C}_0(n^\Phi_\alpha) \setminus \mathcal{C}(n^\Phi_\alpha)$. Since $\forall (m', \phi) \in \mathcal{C}_0(n^\Phi_\alpha) \setminus \mathcal{C}(n^\Phi_\alpha) : \phi \leq h$, it follows that $\phi \leq h$.
\end{proof}

\begin{lemma} \label{lem:vertex_input_max_marking}
Let $n^\Phi$ be a marking-sensitive vertex input node in some marking-sensitive RETE net $(N^\Phi, p^\Phi)$. Let $\mathcal{C}_0$ and $\mathcal{C}$ furthermore be configurations for $(N^\Phi, p^\Phi)$ such that $\mathcal{C}_0$ is consistent for $n^\Phi$. It then holds for host graph $H$ and partial host graph $H_p$ that $\resultslocal{n^\Phi}{N^\Phi}{H}{H_p}{\mathcal{C}} = \mathcal{C}_0(n^\Phi)$.
\end{lemma}

\begin{proof}
Follows trivially from the semantics of the marking-sensitive vertex input node.
\end{proof}

\begin{lemma} \label{lem:filter_max_marking}
Let $n^\Phi$ be a marking filter node with dependency $n^\Phi_\alpha$ in a marking-sensitive RETE net $(N^\Phi, p^\Phi)$. Let $\mathcal{C}_0$ and $\mathcal{C}$ furthermore be configurations for $(N^\Phi, p^\Phi)$ such that $\mathcal{C}_0$ is consistent for $n^\Phi$ and it holds that $\forall (m, \phi) \in \mathcal{C}_0(n^\Phi_\alpha) \setminus \mathcal{C}(n^\Phi_\alpha) : \phi \leq h$ and $\forall (m, \phi) \in \mathcal{C}(n^\Phi_\alpha) \setminus \mathcal{C}_0(n^\Phi_\alpha) : \phi \leq h$. It then holds for host graph $H$ and partial host graph $H_p$ that $\forall (m, \phi) \in \mathcal{C}_0(n^\Phi) \setminus \resultslocal{n^\Phi}{N^\Phi}{H}{H_p}{\mathcal{C}} : \phi \leq h$ and $\forall (m, \phi) \in \resultslocal{n^\Phi}{N^\Phi}{H}{H_p}{\mathcal{C}} \setminus \mathcal{C}_0(n^\Phi) : \phi \leq h$.
\end{lemma}

\begin{proof}
Let $(m, \phi) \in \resultslocal{n^\Phi}{N^\Phi}{H}{H_p}{\mathcal{C}} \setminus \mathcal{C}_0(n^\Phi)$. It must then hold that $\exists (m', \phi) \in \mathcal{C}(n^\Phi_\alpha) \setminus \mathcal{C}_0(n^\Phi_\alpha)$. Since $\forall (m', \phi) \in \mathcal{C}(n^\Phi_\alpha) \setminus \mathcal{C}_0(n^\Phi_\alpha) : \phi \leq h$, it follows that $\phi \leq h$.


Let $(m, \phi) \in \mathcal{C}_0(n^\Phi) \setminus \resultslocal{n^\Phi}{N^\Phi}{H}{H_p}{\mathcal{C}}$. It must then hold that $\exists (m', \phi) \in \mathcal{C}_0(n^\Phi_\alpha) \setminus \mathcal{C}(n^\Phi_\alpha)$. Since $\forall (m', \phi) \in \mathcal{C}_0(n^\Phi_\alpha) \setminus \mathcal{C}(n^\Phi_\alpha) : \phi \leq h$, it follows that $\phi \leq h$.
\end{proof}

\begin{lemma} \label{lem:join_max_marking}
Let $H$ a graph, $H_p \subseteq H$, $(N, p)$ a well-formed RETE net with join-tree height $h \geq 1$, $\psi \in \overline{\mathbb{N}}$ such that $\psi \geq h$, and $(X^\Phi, p^\Phi_X)$ a modular extension of $(N^\Phi, p^\Phi) = localize((N, p)) = (N^\Phi_{\bowtie} \cup N^\Phi_l \cup N^\Phi_r \cup RPS_l \cup RPS_r, p^\Phi)$, with $(N^\Phi_l, p^\Phi_l) = localize((N_l, p_l))$ and $(N^\Phi_r, p^\Phi_r) = localize((N_r, p_r))$, where $p_l$ is the left dependency of the join node $p$ with associated subtree $N_l$ and $p_r$ is the right dependency of $p$ with associated subtree $N_r$. Let $\mathcal{C}_0$ and $\mathcal{C}_1$ furthermore be configurations for $(X^\Phi, p^\Phi_X)$ such that with respect to $H$, $\mathcal{C}_0$ is consistent for all nodes in $V^{N^\Phi}$, $\mathcal{C}_1$ is consistent for all nodes in $V^{N^\Phi} \setminus \{p^\Phi\}$, and it holds that $\forall (m, \phi) \in (\mathcal{C}_0(n^\Phi_l) \setminus \mathcal{C}_1(n^\Phi_l)) \cup (\mathcal{C}_1(n^\Phi_l) \setminus \mathcal{C}_0(n^\Phi_l)) : \phi \leq \psi$ and $\forall (m, \phi) \in (\mathcal{C}_0(n^\Phi_r) \setminus \mathcal{C}_1(n^\Phi_r)) \cup (\mathcal{C}_1(n^\Phi_r) \setminus \mathcal{C}_0(n^\Phi_r)) : \phi \leq h$. It then holds for $\mathcal{C}_2 = execute(p^\Phi, X^\Phi, H, H_p, \mathcal{C}_1)$ that $\forall (m, \phi) \in (\mathcal{C}_0(p^\Phi) \setminus \mathcal{C}_2(p^\Phi)) \cup (\mathcal{C}_2(p^\Phi) \setminus \mathcal{C}_0(p^\Phi)) : \phi \leq \psi$.
\end{lemma}

\begin{proof}
Let $(m, \phi)$ be a tuple in $\mathcal{C}_0(p^\Phi) \setminus \mathcal{C}_2(p^\Phi)$. Because $\mathcal{C}_0$ is consistent for $p^\Phi$, there must be tuples $(m_l, \phi_l) \in \mathcal{C}_0(p^\Phi_l)$ and $(m_r, \phi_r) \in \mathcal{C}_0(p^\Phi_r)$ with $m_l|_{Q_\cap} = m_r|_{Q_\cap}$, where $[\bowtie_{Q_\cap}]^\Phi = p^\Phi$ and $max(\phi_l, \phi_r) = \phi$. Since $(m, \phi) \notin \mathcal{C}_2(p^\Phi)$, it must also hold that $(m_l, \phi_l) \notin \mathcal{C}_1(p^\Phi_l)$ or $(m_r, \phi_r) \notin \mathcal{C}_1(p^\Phi_r)$.

\emph{Case 1}: We assume that $(m_l, \phi_l) \notin \mathcal{C}_1(p^\Phi_l)$ and $(m_r, \phi_r) \in \mathcal{C}_1(p^\Phi_r)$. If $\phi_l \geq \phi_r$, it immediately follows that $\phi \leq \psi$. If $\phi_l < \phi_r$, by Lemma \ref{lem:join_complementarity} and $\mathcal{C}_1$ being consistent for all nodes in $V^{N^\Phi} \setminus \{p^\Phi\}$, it follows that $\exists (m_l', \phi_l') \in \mathcal{C}_1(p^\Phi_l) : m_l'|_{Q_\cap} = m_r|_{Q_\cap}$. Due to Lemma \ref{lem:match_uniqueness} and since $\forall (m', \phi') \in \mathcal{C}_1(n^\Phi_l) \setminus \mathcal{C}_0(n^\Phi_l) : \phi' \leq \psi$, it must then follow that $\phi_r \leq \psi$, since otherwise, it would follow that $(m, \phi) \in \mathcal{C}_2(p^\Phi)$, which is a contradiction.

Case 2: We assume that $(m_l, \phi_l) \in \mathcal{C}_1(p^\Phi_l)$ and $(m_r, \phi_r) \notin \mathcal{C}_1(p^\Phi_r)$. Analogously to Case 1, it then follows that $\phi \leq \psi$.

Case 3: We assume that $(m_l, \phi_l) \notin \mathcal{C}_1(p^\Phi_l)$ and $(m_r, \phi_r) \notin \mathcal{C}_1(p^\Phi_r)$. In this case, since $\forall (m', \phi') \in \mathcal{C}_0(n^\Phi_l) \setminus \mathcal{C}_1(n^\Phi_l) : \phi' \leq \psi$ and $\forall (m', \phi') \in \mathcal{C}_0(n^\Phi_r) \setminus \mathcal{C}_1(n^\Phi_r) : \phi' \leq \psi$, it immediately follows that $\phi \leq \psi$.

It thus follows that $\phi \leq \psi$.


Let $(m, \phi)$ be a tuple in $\mathcal{C}_2(p^\Phi) \setminus \mathcal{C}_0(p^\Phi)$. It then follows that there must be tuples $(m_l, \phi_l) \in \mathcal{C}_1(p^\Phi_l)$ and $(m_r, \phi_r) \in \mathcal{C}_1(p^\Phi_r)$ with $m_l|_{Q_\cap} = m_r|_{Q_\cap}$, where $[\bowtie_{Q_\cap}]^\Phi = p^\Phi$ and $max(\phi_l, \phi_r) = \phi$. Since $(m, \phi) \notin \mathcal{C}_0(p^\Phi)$, it must also hold that $(m_l, \phi_l) \notin \mathcal{C}_0(p^\Phi_l)$ or $(m_r, \phi_r) \notin \mathcal{C}_0(p^\Phi_r)$..

Case 1: We assume that $(m_l, \phi_l) \notin \mathcal{C}_0(p^\Phi_l)$ and $(m_r, \phi_r) \in \mathcal{C}_0(p^\Phi_r)$. If $\phi_l \geq \phi_r$, it immediately follows that $\phi \leq \psi$. If $\phi_l < \phi_r$, by Lemma \ref{lem:join_complementarity} and $\mathcal{C}_1$ being consistent for all nodes in $V^{N^\Phi} \setminus \{p^\Phi\}$, it follows that $\exists (m_l', \phi_l') \in \mathcal{C}_1(p^\Phi_l) : m_l'|_{Q_\cap} = m_r|_{Q_\cap}$. Due to Lemma \ref{lem:match_uniqueness} and since $\forall (m', \phi') \in \mathcal{C}_0(n^\Phi_l) \setminus \mathcal{C}(n^\Phi_l) : \phi' \leq \psi$, it must then follow that $\phi \leq \psi$, since otherwise, it would follow that $(m, \phi) \in \mathcal{C}_0(p^\Phi)$, which is a contradiction.

Case 2: We assume that $(m_l, \phi_l) \in \mathcal{C}_0(p^\Phi_l)$ and $(m_r, \phi_r) \notin \mathcal{C}_0(p^\Phi_r)$. Analogously to Case 1, it then follows that $\phi \leq \psi$.

Case 3: We assume that $(m_l, \phi_l) \in \mathcal{C}_0(p^\Phi_l)$ and $(m_r, \phi_r) \in \mathcal{C}_0(p^\Phi_r)$. In this case, since $\forall (m', \phi') \in \mathcal{C}_1(n^\Phi_l) \setminus \mathcal{C}_0(n^\Phi_l) : \phi' \leq \psi$ and $\forall (m', \phi') \in \mathcal{C}_1(n^\Phi_r) \setminus \mathcal{C}_0(n^\Phi_r) : \phi' \leq \psi$, it immediately follows that $\phi \leq \psi$

It thus follows that $\phi \leq \psi$.
\end{proof}


\begin{lemma} \label{lem:match_uniqueness}
Let $(N, p)$ be a well-formed RETE net and $(X^\Phi, p^\Phi_X)$ a modular extension of $(N^\Phi, p^\Phi) = localize((N, p))$. It then holds for any configuration $\mathcal{C}$ of $(X^\Phi, p^\Phi_X)$ that is consistent for all nodes in $V^{N^\Phi}$ that $\forall (m, \phi), (m', \phi') \in \mathcal{C}(p^\Phi) : m = m' \Rightarrow \phi = \phi'$.
\end{lemma}

\begin{proof}
We show the correctness of the lemma via induction over the height $h$ of the join tree $N$.

In the base case of $h = 0$, $(N^\Phi, p^\Phi) = (LNS(p), [\cup]^\Phi)$. By the semantics of $[\cup]^\Phi$ and $\mathcal{C}$ being consistent for $[\cup]^\Phi$, it follows that $\forall (m, \phi), (m', \phi') \in \mathcal{C}(p^\Phi) : m = m' \Rightarrow \phi = \phi'$.

We now show that if the lemma holds for any well-formed RETE net $(N_h, p_h)$ with join-tree height $h$, it also holds for any well-formed RETE net $(N, p)$ with join-tree height $h+1$.

Since $h + 1 \geq 1$, it holds that $(N^\Phi, p^\Phi) = localize((N, p)) = (N^\Phi_{\bowtie} \cup N^\Phi_l \cup N^\Phi_r \cup RPS_l \cup RPS_r, p^\Phi)$, with $(N^\Phi_l, p^\Phi_l) = localize((N_l, p_l))$ and $(N^\Phi_r, p^\Phi_r) = localize((N_r, p_r))$, where $p_l$ is the left dependency of the join node $p$ with associated subtree $N_l$ and $p_r$ is the right dependency of $p$ with associated subtree $N_r$. Since both $N_l$ and $N_r$ must have a join-tree height less than $h + 1$, by the induction hypothesis it follows that $\forall (m, \phi), (m', \phi') \in \mathcal{C}(p^\Phi_l) : m = m' \Rightarrow \phi = \phi'$ and $\forall (m, \phi), (m', \phi') \in \mathcal{C}(p^\Phi_r) : m = m' \Rightarrow \phi = \phi'$. By the semantics of the marking-sensitive join node $p^\Phi$, which has $p^\Phi_l$ and $p^\Phi_r$ as left respectively right dependency, it must follow that $\forall (m, \phi), (m', \phi') \in \mathcal{C}(p^\Phi) : m = m' \Rightarrow \phi = \phi'$.

From the correctness of the base case and the induction step then follows the correctness of the lemma.
\end{proof}


We say that a marking-sensitive RETE net $(X^\Phi, p^\Phi_X)$ is a \emph{modular extension} of $(N^\Phi, p^\Phi)$ if $N^\Phi \subseteq X^\Phi$ and $\forall e \in E^{X^\Phi} : s^{X^\Phi}(e) \in V^{N^\Phi} \wedge t^{X^\Phi}(e) \in V^{N^\Phi} \Rightarrow e \in E^{N^\Phi}$.

We say that an execution order $O$ for $(N^\Phi, p^\Phi)$ is \emph{robust under modular extension} if for any modular extension $(X^\Phi, p^\Phi_X)$, host graph $H$, and starting configuration $\mathcal{C}_0$ for $(X^\Phi, p^\Phi_X)$, it holds that $\mathcal{C} = execute(O, X^\Phi, H, H_p, \mathcal{C_0})$ is consistent for all RETE nodes $n \in V^{N^\Phi}$.

The execution order $O = order((N^\Phi, p^\Phi))$ is indeed robust under modular extension:

\begin{lemma} \label{lem:order_robustness}
Given a well-formed RETE net $(N, p)$, the execution order $O = order((N^\Phi, p^\Phi))$ for $(N^\Phi, p^\Phi) = localize((N, p))$ is robust under modular extension.
\end{lemma}

\begin{proof}
We prove the lemma by induction over the height $h$ of the join tree $N$.

We begin by showing that the lemma holds for the base case of $h = 0$, where $N$ consists of a single edge input node. Thus, $(N^\Phi, p^\Phi) = localize((N, p)) = (LNS(p), p^\Phi)$, that is, the marking-sensitive RETE graph is given by a single local navigation structure. This graph contains no cyclic dependencies on current results. A modular extension $(X^\Phi, p^\Phi_X)$ of $(N^\Phi, p^\Phi)$ may not add any dependencies between nodes in $V^{N^\Phi}$ and hence cannot create any cycles. Since $O$ contains all nodes in $V^{N^\Phi}$ and no node execution in $O$ can impact the target result set of a previous node, after executing $O = toposort(N^\Phi)^{-1}$ with some starting configuration $\mathcal{C}_0$ for $(X^\Phi, p^\Phi_X)$, the resulting configuration $\mathcal{C}_1 = execute(O, X^\Phi, H, H_p, \mathcal{C}_0)$ must thus be consistent for all nodes $n \in V^{N^\Phi}$.

We now show that, if it holds that $O_h = order((N^\Phi_h, p^\Phi_h))$ with $(N^\Phi_h, p^\Phi_h) = localize((N_h, p_h))$ is robust under modular extension for any well-formed RETE net $(N_h, p_h)$ with join-tree height less than or equal to $h$, it follows that $O = order((N^\Phi, p^\Phi))$ with $(N^\Phi, p^\Phi) = localize(N, p)$ is robust under modular extension for any well-formed RETE net $(N, p)$ with join-tree height $h+1$.

Since $h + 1 \geq 1$, it follows that $p$ is a join node, which has a left dependency $p_l$, which is the root of a RETE subtree $N_l$, and a right dependency $p_r$, which is the root of a RETE subtree $N_r$. Hence, $(N^\Phi, p^\Phi) = (N^\Phi_{\bowtie} \cup N^\Phi_l \cup N^\Phi_r \cup RPS_l \cup RPS_r, p^\Phi)$, with $(N^\Phi_l, p^\Phi_l) = localize((N_l, p_l))$, $(N^\Phi_r, p^\Phi_r) = localize((N_r, p_r))$, $N^\Phi_{\bowtie}$ a RETE net fragment consisting of a marking-sensitive join $p^\Phi$ with left dependency $p^\Phi_l$ and right dependency $p^\Phi_r$, $RPS_l = RPS(N^\Phi_r, p^\Phi_l)$, and $RPS_r = RPS(N^\Phi_l, p^\Phi_r))$. The execution order is thus given by $O = order((N^\Phi, p^\Phi)) = order(RPS_r) \circ order(N^\Phi_l) \circ order(RPS_l) \circ order(N^\Phi_r) \circ order(RPS_r) \circ order(N^\Phi_l) \circ order(N^\Phi_{\bowtie})$.

Given some modular extension $(X^\Phi, p^\Phi_X)$ of $(N^\Phi, p^\Phi)$, we can observe that $(X^\Phi, p^\Phi_X)$ is also a modular extension of both $(N^\Phi_l, p^\Phi_l)$ and $(N^\Phi_r, p^\Phi_r)$ by definition. Furthermore, because $N$ has a join-tree height of $h + 1$, it follows that $N_l$ and $N_r$ must have join-tree height $h_l \leq h$ respectively $h_r \leq h$.

Some initial configuration $\mathcal{C}_0$ for $(X^\Phi, p^\Phi_X)$ may be inconsistent for all of the nodes in $N^\Phi$. Because $RPS_l$ just consists of three consecutive RETE nodes and the modular extension $X^\Phi$ may not add additional dependencies between nodes in $N^\Phi$, it follows that $\mathcal{C}_{0.1} = execute(order(RPS_r), X^\Phi, H, H_p, \mathcal{C}_0)$ is consistent for all nodes in $V^{RPS_r}$.

Since, by the induction hypothesis, we know that $O_l = order((N^\Phi_l, p^\Phi_l))$ is robust under modular extension and $(X^\Phi, p^\Phi_X)$ is a modular extension of $(N^\Phi_l, p^\Phi_l)$ and since no node in $V^{RPS_r}$ depends on a node in $N^\Phi_l$, it follows that after executing $O_l$, the resulting configuration $\mathcal{C}_{0.2} = execute(O_l, X^\Phi, H, H_p, \mathcal{C}_{0.1})$ is consistent for all nodes in $V^{N^\Phi_l} \cup V^{RPS_r}$.

Since no node in $V^{N^\Phi_l} \cup V^{RPS_r}$ depends on any node in $RPS_l$, by the same logic as for $RPS_r$, we know that $\mathcal{C}_{0.3} = execute(order(RPS_l), X^\Phi, H, H_p, \mathcal{C}_{0.2})$ must be consistent for all nodes in $V^{N^\Phi_l} \cup V^{RPS_l} \cup V^{RPS_r}$.

Due to the robustness under modular extension of $O_r = order((N^\Phi_r, p^\Phi_r))$ and because among nodes in $V^{N^\Phi_l} \cup V^{RPS_l} \cup V^{RPS_r}$, only the marking filter node $[\phi > h + 1]_l^\Phi$ in $V^{RPS_r}$ depends on a node in $V^{N^\Phi_r}$, $\mathcal{C}_{0.4} = execute(O_r, X^\Phi, H, H_p, \mathcal{C}_{0.3})$ then must be consistent for all nodes in $V^{N^\Phi} \setminus (\{[\phi > h + 1]_l^\Phi\} \cup V^{N^\Phi_{\bowtie}})$.

By similar logic, $\mathcal{C}_{0.5} = execute(order(RPS_r), X^\Phi, H, H_p, \mathcal{C}_{0.4})$ is then consistent for all nodes in $V^{N^\Phi} \setminus (\{[\cup]^\Phi_v\} \cup V^{N^\Phi_{\bowtie}})$, where $[\cup]^\Phi_v$ is the union node in $N^\Phi_l$ that depends on the marking assignment node $[\phi := h + 1]^\Phi$ in $RPS_r$.

By the semantics of the marking assignment node $[\phi := h + 1]$ in $RPS_r$ and no node in $V^{N^\Phi_l}$ depending on any other node in $order(RPS_l) \circ order(N^\Phi_r) \circ order(RPS_r)$, we know that $\forall n_X \in V^{X^\Phi}: (\exists e \in E^{X^\Phi} : t^{X^\Phi} = n_X \wedge s^{X^\Phi} \in V^{N^\Phi_l}) \Rightarrow \forall (m, \phi) \in (\mathcal{C}_{0.5}(n_X) \setminus \mathcal{C}_{0.4}(n_X)) \cup (\mathcal{C}_{0.4}(n_X) \setminus \mathcal{C}_{0.5}(n_X)) : \phi \leq h + 1$. Furthermore, $\mathcal{C}_{0.4}$ is consistent for all nodes in $V^{N^\Phi_l}$ and it holds that $\forall n \in V^{N^\Phi_l}: \mathcal{C}_{0.5}(n) = \mathcal{C}_{0.4}(n)$. By the induction hypothesis, we know that the execution order $order(N^\Phi_{h-1})$ is robust under modular extension for any localized RETE net $(N^\Phi_{h-1}, p^\Phi_{h-1}) = localize((N_{h-1}, p_{h-1}))$, if $(N_{h-1}, p_{h-1})$ is a well-formed RETE net with join tree height $h' < h$. Finally, $N_l$ must have a join tree height of $h' < h + 1$.

Thus, by Lemma \ref{lem:localized_rete_max_marking}, it follows for $\mathcal{C}_{0.6} = execute(O_l, X^\Phi, H, H_p, \mathcal{C}_{0.5})$ that $\forall (m, \phi) \in (\mathcal{C}_{0.6}(p^\Phi_l) \setminus \mathcal{C}_{0.5}(p^\Phi_l)) \cup (\mathcal{C}_{0.5}(p^\Phi_l) \setminus \mathcal{C}_{0.6}(p^\Phi_l)) : \phi \leq h + 1$. Because no node in $V^{N^\Phi} \setminus V^{N^\Phi_l}$ other than the marking filter node $[\phi > h + 1]_r^\Phi$ can depend on any node in $V^{N^\Phi_l}$, because $O_l$ is robust under modular extension, and by the semantics of $[\phi > h + 1]_r^\Phi$, it then follows that $\mathcal{C}_{0.6}$ is consistent for all nodes in $V^{N^\Phi} \setminus V^{N^\Phi_{\bowtie}}$.

Since $N^\Phi_{\bowtie}$ consists only of a marking-sensitive join node which no other node in $V^{N^\Phi}$ depends on, it follows that $\mathcal{C}_1 = execute(order(N^\Phi_{\bowtie}, X^\Phi, H, H_p, \mathcal{C}_{0.6}) = execute(order((N^\Phi, p^\Phi)), H, \mathcal{C}_0)$ must be consistent for all nodes in $V^{N^\Phi}$. It thus follows that $order((N^\Phi, p^\Phi))$ must be robust under modular extension for any well-formed RETE net $(N, p)$ with join-tree height $h+1$.

From the correctness of the base case and the induction step then follows the correctness of the lemma.
\end{proof}


\begin{lemma} \label{lem:execution_time_forward_navigation}
Let $H$ be a graph, $H_p \subseteq H$, $[v \rightarrow_n w]^\Phi \in V^{N^\Phi}$ a forward navigation node with dependency $n^\Phi_v$, and $\mathcal{C}^\Phi_0$ and $\mathcal{C}^\Phi_1$ configurations such that
$\mathcal{C}^\Phi_0$ is consistent for $[v \rightarrow_n w]^\Phi$,
$\forall (m, \phi_0) \in \mathcal{C}^\Phi_0(n^\Phi_v) : \exists (m, \phi_1) \in \mathcal{C}^\Phi_1(n^\Phi_v) : \phi_0 \leq \phi_1$, and
$\mathcal{C}^\Phi_0([v \rightarrow_n w]^\Phi) = \mathcal{C}^\Phi_1([v \rightarrow_n w]^\Phi)$.
Executing $[v \rightarrow_n w]^\Phi$ via $\mathcal{C}^\Phi_2 = execute([v \rightarrow_n w]^\Phi, N^\Phi, H, H_p, \mathcal{C}^\Phi_1)$ then takes $O(S^+_D + S^+_R + N^\uparrow_D + N^\uparrow_R)$ steps, where
$S^+_D = \sum_{m \in M^+_D} |m|$ with $M^+_D = \{m | \exists (m, \phi_1) \in \mathcal{C}^\Phi_1(n^\Phi_v) \wedge \nexists (m, \phi_0) \in \mathcal{C}^\Phi_0(n^\Phi_v) \}$,
$S^+_R = \sum_{m \in M^+_R} |m|$ with $M^+_R = \{m | \exists (m, \phi_2) \in \mathcal{C}^\Phi_2([v \rightarrow_n w]^\Phi) \wedge \nexists (m, \phi_0) \in \mathcal{C}^\Phi_0([v \rightarrow_n w]^\Phi) \}$,
$N^\uparrow_D = |\mathcal{C}^\Phi_0(n^\Phi_v) \setminus \mathcal{C}^\Phi_1(n^\Phi_v)|$, and
$N^\uparrow_R = |\mathcal{C}^\Phi_0([v \rightarrow_n w]^\Phi) \setminus \mathcal{C}^\Phi_2([v \rightarrow_n w]^\Phi)|$.
\end{lemma}

\begin{proof}
Assuming changes to the result set of $n^\Phi_v$ are cached and $\mathcal{C}^\Phi_1$ is available for modification into $\mathcal{C}^\Phi_2$, $[v \rightarrow_n w]^\Phi$ can be executed in two phases.

First execution iterates over the intermediate results in $\mathcal{C}^\Phi_1(n^\Phi_v)$ that represent newly added matches compared to $\mathcal{C}^\Phi_0$, constructs the corresponding intermediate results for $[v \rightarrow_n w]^\Phi$, and adds them to $\mathcal{C}^\Phi_1([v \rightarrow_n w]^\Phi)$. Given a suitable representation of $H$ that allows efficient retrieval of edges adjacent to a node and assuming indexing structures that allow insertion times into result sets linear in match size, this takes $O(S^+_D + S^+_R)$ steps.

Second, execution iterates over the intermediate results in $\mathcal{C}^\Phi_1(n^\Phi_v)$ that represent matches whose marking has been increased compared to $\mathcal{C}^\Phi_0$ and adjusts the marking of the corresponding intermediate results in $\mathcal{C}^\Phi_1([v \rightarrow_n w]^\Phi)$. By modifying intermediate results for $n^\Phi_v$ on a change of marking instead of creating new tuples and by maintaining appropriate references between the intermediate results for $n^\Phi_v$ and for $[v \rightarrow_n w]^\Phi$, this can be achieved in $O(N^\uparrow_D + N^\uparrow_R)$ steps.

This yields the stated complexity.
\end{proof}

\begin{lemma} \label{lem:execution_time_backward_navigation}
Let $H$ be a graph, $H_p \subseteq H$, $[w \leftarrow_n v]^\Phi \in V^{N^\Phi}$ a backward navigation node with dependency $n^\Phi_w$, and $\mathcal{C}^\Phi_0$ and $\mathcal{C}^\Phi_1$ configurations such that
$\mathcal{C}^\Phi_0$ is consistent for $[w \leftarrow_n v]^\Phi$,
$\forall (m, \phi_0) \in \mathcal{C}^\Phi_0(n^\Phi_w) : \exists (m, \phi_1) \in \mathcal{C}^\Phi_1(n^\Phi_w) : \phi_0 \leq \phi_1$, and
$\mathcal{C}^\Phi_0([w \leftarrow_n v]^\Phi) = \mathcal{C}^\Phi_1([w \leftarrow_n v]^\Phi)$.
Executing $[w \leftarrow_n v]^\Phi$ via $\mathcal{C}^\Phi_2 = execute([w \leftarrow_n v]^\Phi, N^\Phi, H, H_p, \mathcal{C}^\Phi_1)$ then takes $O(S^+_D + S^+_R + N^\uparrow_D + N^\uparrow_R)$ steps, where
$S^+_D = \sum_{m \in M^+_D} |m|$ with $M^+_D = \{m | \exists (m, \phi_1) \in \mathcal{C}^\Phi_1(n^\Phi_w) \wedge \nexists (m, \phi_0) \in \mathcal{C}^\Phi_0(n^\Phi_w) \}$,
$S^+_R = \sum_{m \in M^+_R} |m|$ with $M^+_R = \{m | \exists (m, \phi_2) \in \mathcal{C}^\Phi_2([w \leftarrow_n v]^\Phi) \wedge \nexists (m, \phi_0) \in \mathcal{C}^\Phi_0([w \leftarrow_n v]^\Phi) \}$,
$N^\uparrow_D = |\mathcal{C}^\Phi_0(n^\Phi_w) \setminus \mathcal{C}^\Phi_1(n^\Phi_w)|$, and
$N^\uparrow_R = |\mathcal{C}^\Phi_0([w \leftarrow_n v]^\Phi) \setminus \mathcal{C}^\Phi_2([w \leftarrow_n v]^\Phi)|$.
\end{lemma}

\begin{proof}
Assuming changes to the result set of $n^\Phi_w$ are cached and $\mathcal{C}^\Phi_1$ is available for modification into $\mathcal{C}^\Phi_2$, $[w \leftarrow_n v]^\Phi$ can be executed in two phases.

First execution iterates over the intermediate results in $\mathcal{C}^\Phi_1(n^\Phi_w)$ that represent newly added matches compared to $\mathcal{C}^\Phi_0$, constructs the corresponding intermediate results for $[w \leftarrow_n v]^\Phi$, and adds them to $\mathcal{C}^\Phi_1([w \leftarrow_n v]^\Phi)$. Given a suitable representation of $H$ that allows efficient retrieval of edges adjacent to a node and assuming indexing structures that allow insertion times into result sets linear in match size, this takes $O(S^+_D + S^+_R)$ steps.

Second, execution iterates over the intermediate results in $\mathcal{C}^\Phi_1(n^\Phi_w)$ that represent matches whose marking has been increased compared to $\mathcal{C}^\Phi_0$ and adjusts the marking of the corresponding intermediate results in $\mathcal{C}^\Phi_1([w \leftarrow_n v]^\Phi)$. By modifying intermediate results for $n^\Phi_w$ on a change of marking instead of creating new tuples and by maintaining appropriate references between the intermediate results for $n^\Phi_w$ and for $[w \leftarrow_n v]^\Phi$, this can be achieved in $O(N^\uparrow_D + N^\uparrow_R)$ steps.

This yields the stated complexity.
\end{proof}

\begin{lemma} \label{lem:execution_time_projection}
Let $H$ be a graph, $H_p \subseteq H$, $[\pi_Q]^\Phi \in V^{N^\Phi}$ a marking-sensitive projection node with dependency $n^\Phi_\alpha$, and $\mathcal{C}^\Phi_0$ and $\mathcal{C}^\Phi_1$ configurations such that
$\mathcal{C}^\Phi_0$ is consistent for $[\pi_Q]^\Phi$,
$\forall (m, \phi_0) \in \mathcal{C}^\Phi_0(n^\Phi_\alpha) : \exists (m, \phi_1) \in \mathcal{C}^\Phi_1(n^\Phi_\alpha) : \phi_0 \leq \phi_1$, and
$\mathcal{C}^\Phi_0([\pi_Q]^\Phi) = \mathcal{C}^\Phi_1([\pi_Q]^\Phi)$.
Executing the marking-sensitive projection node $[\pi_Q]^\Phi$ via $\mathcal{C}^\Phi_2 = execute([\pi_Q]^\Phi, N^\Phi, H, H_p, \mathcal{C}^\Phi_1)$ then takes $O(S^+_D + S^+_R + N^\uparrow_D + N^\uparrow_R)$ steps, where
$S^+_D = \sum_{m \in M^+_D} |m|$ with $M^+_D = \{m | \exists (m, \phi_1) \in \mathcal{C}^\Phi_1(n^\Phi_\alpha) \wedge \nexists (m, \phi_0) \in \mathcal{C}^\Phi_0(n^\Phi_\alpha) \}$,
$S^+_R = \sum_{m \in M^+_R} |m|$ with $M^+_R = \{m | \exists (m, \phi_2) \in \mathcal{C}^\Phi_2([\pi_Q]^\Phi) \wedge \nexists (m, \phi_0) \in \mathcal{C}^\Phi_0([\pi_Q]^\Phi) \}$,
$N^\uparrow_D = |\mathcal{C}^\Phi_0(n^\Phi_\alpha) \setminus \mathcal{C}^\Phi_1(n^\Phi_\alpha)|$, and
$N^\uparrow_R = |\mathcal{C}^\Phi_0([\pi_Q]^\Phi) \setminus \mathcal{C}^\Phi_2([\pi_Q]^\Phi)|$.
\end{lemma}

\begin{proof}
Assuming changes to the result set of $n^\Phi_\alpha$ are cached and $\mathcal{C}^\Phi_1$ is available for modification into $\mathcal{C}^\Phi_2$, $[\pi_Q]^\Phi$ can be executed in two phases.

First execution iterates over the intermediate results in $\mathcal{C}^\Phi_1(n^\Phi_\alpha)$ that represent newly added matches compared to $\mathcal{C}^\Phi_0$, constructs the corresponding intermediate results for $[\pi_Q]^\Phi$, and adds them to $\mathcal{C}^\Phi_1([\pi_Q]^\Phi)$ if they are not yet present. Assuming indexing structures that allow presence check and insertion times linear in match size, this takes $O(S^+_D + S^+_R)$ steps.

Second, execution iterates over the intermediate results in $\mathcal{C}^\Phi_1(n^\Phi_\alpha)$ that represent matches whose marking has been increased compared to $\mathcal{C}^\Phi_0$ and adjusts the marking of the corresponding intermediate results in $\mathcal{C}^\Phi_1([\pi_Q]^\Phi)$, if its current marking is lower. By modifying intermediate results for $n^\Phi_\alpha$ on a change of marking instead of creating new tuples and by maintaining appropriate references between the intermediate results for $n^\Phi_\alpha$ and for $[\pi_Q]^\Phi$, this can be achieved in $O(N^\uparrow_D + N^\uparrow_R)$ steps.

This yields the stated complexity.
\end{proof}

\begin{lemma} \label{lem:execution_time_union}
Let $H$ be a graph, $H_p \subseteq H$, $[\cup]^\Phi \in V^{N^\Phi}$ a marking-sensitive union node with dependencies $N^\Phi_\alpha$, and $\mathcal{C}^\Phi_0$ and $\mathcal{C}^\Phi_1$ configurations such that
$\mathcal{C}^\Phi_0$ is consistent for $[\cup]^\Phi$,
$\forall n^\Phi_\alpha \in N^\Phi_\alpha : \forall (m, \phi_0) \in \mathcal{C}^\Phi_0(n^\Phi_\alpha) : \exists (m, \phi_1) \in \mathcal{C}^\Phi_1(n^\Phi_\alpha) : \phi_0 \leq \phi_1$, and
$\mathcal{C}^\Phi_0([\cup]^\Phi) = \mathcal{C}^\Phi_1([\cup]^\Phi)$.
Executing $[\cup]^\Phi$ via $\mathcal{C}^\Phi_2 = execute([\cup]^\Phi, N^\Phi, H, H_p, \mathcal{C}^\Phi_1)$ then takes $O(S^+_D + S^+_R + N^\uparrow_D + N^\uparrow_R)$ steps, where
$S^+_D = \sum_{n^\Phi_\alpha \in N^\Phi_\alpha}\sum_{m \in M^+_{n^\Phi_\alpha}} |m|$ with $M^+_{n^\Phi_\alpha} = \{m | \exists (m, \phi_1) \in \mathcal{C}^\Phi_1(n^\Phi_\alpha) \wedge \nexists (m, \phi_0) \in \mathcal{C}^\Phi_0(n^\Phi_\alpha) \}$,
$S^+_R = \sum_{m \in M^+_R} |m|$ with $M^+_R = \{m | \exists (m, \phi_2) \in \mathcal{C}^\Phi_2([\cup]^\Phi) \wedge \nexists (m, \phi_0) \in \mathcal{C}^\Phi_0([\cup]^\Phi) \}$,
$N^\uparrow_D = \sum_{n^\Phi_\alpha \in N^\Phi_\alpha}|\mathcal{C}^\Phi_0(n^\Phi_\alpha) \setminus \mathcal{C}^\Phi_1(n^\Phi_\alpha)|$, and
$N^\uparrow_R = |\mathcal{C}^\Phi_0([\cup]^\Phi) \setminus \mathcal{C}^\Phi_2([\cup]^\Phi)|$.
\end{lemma}

\begin{proof}
Assuming changes to the result sets of nodes in $N^\Phi_\alpha$ are cached and $\mathcal{C}^\Phi_1$ is available for modification into $\mathcal{C}^\Phi_2$, $[\cup]^\Phi$ can be executed in two phases.

First, for each $n^\Phi_\alpha \in N^\Phi_\alpha$, execution iterates over the intermediate results in $\mathcal{C}^\Phi_1(n^\Phi_\alpha)$ that represent newly added matches compared to $\mathcal{C}^\Phi_0$, and retrieves the corresponding intermediate result from $\mathcal{C}^\Phi_1([\cup]^\Phi)$ if there is any. If there is one, execution then checks whether the newly added dependency match has a higher marking and potentially modifies the marking in $\mathcal{C}^\Phi_1([\cup]^\Phi)$. Otherwise, the intermediate result for $[\cup]^\Phi$ corresponding to the newly added match is constructed and added to $\mathcal{C}^\Phi_1([\cup]^\Phi)$. Assuming indexing structures that allow retrieval and insertion times linear in match size, this takes $O(S^+_D + S^+_R)$ steps for all $n^\Phi_\alpha$.

Second, for each $n^\Phi_\alpha \in N^\Phi_\alpha$, execution iterates over the intermediate results in $\mathcal{C}^\Phi_1(n^\Phi_\alpha)$ that represent matches whose marking has been increased compared to $\mathcal{C}^\Phi_0$ and adjusts the marking of the corresponding intermediate results in $\mathcal{C}^\Phi_1([\cup]^\Phi)$, if its current marking is lower. By modifying intermediate results for $n^\Phi_\alpha$ on a change of marking instead of creating new tuples and by maintaining appropriate references between the intermediate results for $n^\Phi_\alpha$ and for $[\cup]^\Phi$, this can be achieved in $O(N^\uparrow_D + N^\uparrow_R)$ steps for all $n^\Phi_\alpha$.

This yields the stated complexity.
\end{proof}

\begin{lemma} \label{lem:execution_time_marking_filter}
Let $H$ be a graph, $H_p \subseteq H$, $[\phi > x]^\Phi \in V^{N^\Phi}$ a marking filter node with dependency $n^\Phi_\alpha$, and $\mathcal{C}^\Phi_0$ and $\mathcal{C}^\Phi_1$ configurations such that
$\mathcal{C}^\Phi_0$ is consistent for $[\phi > x]^\Phi$,
$\forall (m, \phi_0) \in \mathcal{C}^\Phi_0(n^\Phi_\alpha) : \exists (m, \phi_1) \in \mathcal{C}^\Phi_1(n^\Phi_\alpha) : \phi_0 \leq \phi_1$, and
$\mathcal{C}^\Phi_0([\phi > x]^\Phi) = \mathcal{C}^\Phi_1([\phi > x]^\Phi)$.
Executing $[\phi > x]^\Phi$ via $\mathcal{C}^\Phi_2 = execute([\phi > x]^\Phi, N^\Phi, H, H_p, \mathcal{C}^\Phi_1)$ then takes $O(S^+_D + S^+_R + N^\uparrow_D + N^\uparrow_R)$ steps, where
$S^+_D = \sum_{m \in M^+_D} |m|$ with $M^+_D = \{m | \exists (m, \phi_1) \in \mathcal{C}^\Phi_1(n^\Phi_\alpha) \wedge \nexists (m, \phi_0) \in \mathcal{C}^\Phi_0(n^\Phi_\alpha) \}$,
$S^+_R = \sum_{m \in M^+_R} |m|$ with $M^+_R = \{m | \exists (m, \phi_2) \in \mathcal{C}^\Phi_2([\phi > x]^\Phi) \wedge \nexists (m, \phi_0) \in \mathcal{C}^\Phi_0([\phi > x]^\Phi) \}$,
$N^\uparrow_D = |\mathcal{C}^\Phi_0(n^\Phi_\alpha) \setminus \mathcal{C}^\Phi_1(n^\Phi_\alpha)|$, and
$N^\uparrow_R = |\mathcal{C}^\Phi_0([\phi > x]^\Phi) \setminus \mathcal{C}^\Phi_2([\phi > x]^\Phi)|$.
\end{lemma}

\begin{proof}
Assuming changes to the result set of $n^\Phi_\alpha$ are cached and $\mathcal{C}^\Phi_1$ is available for modification into $\mathcal{C}^\Phi_2$, $[\phi > x]^\Phi$ can be executed in two phases.

First execution iterates over the intermediate results in $\mathcal{C}^\Phi_1(n^\Phi_\alpha)$ that represent newly added matches compared to $\mathcal{C}^\Phi_0$ and constructs the corresponding intermediate results for $[\phi > x]^\Phi$. Each such constructed result with a marking higher than $x$ is then added to $\mathcal{C}^\Phi_1([\phi > x]^\Phi)$. Constructed tuples with lower marking are stored separatly to potentially be added later, already locating their insertion position into $\mathcal{C}^\Phi_1([\phi > x]^\Phi)$. Assuming indexing structures that allow insertion times linear in match size, this takes $O(S^+_D + S^+_R)$ steps.

Second, execution iterates over the intermediate results in $\mathcal{C}^\Phi_1(n^\Phi_\alpha)$ that represent matches whose marking has been increased compared to $\mathcal{C}^\Phi_0$ and adjusts the marking of the corresponding intermediate results in $\mathcal{C}^\Phi_1([\phi > x]^\Phi)$ or in the separate storage. If an intermediate result in the separate storage now has a marking higher than $x$, it is added into its previously located position in $\mathcal{C}^\Phi_1([\phi > x]^\Phi)$. By modifying intermediate results for $n^\Phi_\alpha$ on a change of marking instead of creating new tuples and by maintaining appropriate references between the intermediate results for $n^\Phi_\alpha$ and for $[\phi > x]^\Phi$, this can be achieved in $O(N^\uparrow_D + N^\uparrow_R)$ steps.

This yields the stated complexity.
\end{proof}

\begin{lemma} \label{lem:execution_time_marking_assignment}
Let $H$ be a graph, $H_p \subseteq H$, $[\phi := x]^\Phi \in V^{N^\Phi}$ a marking filter node with dependency $n^\Phi_\alpha$, and $\mathcal{C}^\Phi_0$ and $\mathcal{C}^\Phi_1$ configurations such that
$\mathcal{C}^\Phi_0$ is consistent for $[\phi := x]^\Phi$,
$\forall (m, \phi_0) \in \mathcal{C}^\Phi_0(n^\Phi_\alpha) : \exists (m, \phi_1) \in \mathcal{C}^\Phi_1(n^\Phi_\alpha) : \phi_0 \leq \phi_1$, and
$\mathcal{C}^\Phi_0([\phi := x]^\Phi) = \mathcal{C}^\Phi_1([\phi := x]^\Phi)$.
Executing $[\phi := x]^\Phi$ via $\mathcal{C}^\Phi_2 = execute([\phi := x]^\Phi, N^\Phi, H, H_p, \mathcal{C}^\Phi_1)$ then takes $O(S^+_D + S^+_R)$ steps, where
$S^+_D = \sum_{m \in M^+_D} |m|$ with $M^+_D = \{m | \exists (m, \phi_1) \in \mathcal{C}^\Phi_1(n^\Phi_\alpha) \wedge \nexists (m, \phi_0) \in \mathcal{C}^\Phi_0(n^\Phi_\alpha) \}$ and
$S^+_R = \sum_{m \in M^+_R} |m|$ with $M^+_R = \{m | \exists (m, \phi_2) \in \mathcal{C}^\Phi_2([\phi > x]^\Phi) \wedge \nexists (m, \phi_0) \in \mathcal{C}^\Phi_0([\phi > x]^\Phi) \}$.
\end{lemma}

\begin{proof}
Assuming changes to the result set of $n^\Phi_\alpha$ are cached and $\mathcal{C}^\Phi_1$ is available for modification into $\mathcal{C}^\Phi_2$, $[\phi := x]^\Phi$ can be executed as follows:

Execution simply iterates over the intermediate results in $\mathcal{C}^\Phi_1(n^\Phi_\alpha)$ that represent newly added matches compared to $\mathcal{C}^\Phi_0$ and constructs the corresponding intermediate results for $[\phi := x]^\Phi$. Each such constructed result is then added to $\mathcal{C}^\Phi_1([\phi := x]^\Phi)$. Assuming indexing structures that allow insertion times linear in match size, this takes $O(S^+_D + S^+_R)$ steps.

Modifications to $\mathcal{C}^\Phi_1(n^\Phi_\alpha)$ that add no new matches but only change the marking of already present matches do not affect the target result set of $[\phi := x]^\Phi$ and can hence be ignored.

This yields the stated complexity.
\end{proof}

\begin{lemma} \label{lem:execution_time_vertex_input}
Let $H$ be a graph, $H_p \subseteq H$, $[v]^\Phi \in V^{N^\Phi}$ a marking-sensitive vertex input node, and $\mathcal{C}^\Phi_0$ a configuration such that
$\mathcal{C}^\Phi_0([v]^\Phi) = \emptyset$.
Executing $[v]^\Phi$ via $\mathcal{C}^\Phi_1 = execute([v]^\Phi, N^\Phi, H, H_p, \mathcal{C}^\Phi_0)$ then takes $O(S^+_R)$ steps, where
$S^+_R = \sum_{m \in M^+_R} |m|$ with $M^+_R = \{m | \exists (m, \phi_1) \in \mathcal{C}^\Phi_1([v]^\Phi)\}$.
\end{lemma}

\begin{proof}
Assuming $\mathcal{C}^\Phi_0$ is available for modification into $\mathcal{C}^\Phi_1$, $[v]^\Phi$ can be executed by simply retrieving the matching vertices from $H_p$, constructing the trivial matches, and adding them to $\mathcal{C}^\Phi_0([v]^\Phi)$. Assuming a represnetation of $H_p$ that allows efficient retrieval of the matching vertices and indexing structures that allow insertion times linear in match size, this takes $O(S^+_R)$ steps.
\end{proof}

\begin{lemma} \label{lem:execution_time_vertex_input_consistent}
Let $H$ be a graph, $H_p \subseteq H$, $[v]^\Phi \in V^{N^\Phi}$ a marking-sensitive vertex input node, and $\mathcal{C}^\Phi_0$ a configuration such that
$\mathcal{C}^\Phi_0([v]^\Phi)$ is consistent for $[v]^\Phi$.
Executing $[v]^\Phi$ via $\mathcal{C}^\Phi_1 = execute([v]^\Phi, N^\Phi, H, H_p, \mathcal{C}^\Phi_0)$ then takes $O(1)$ steps.
\end{lemma}

\begin{proof}
Assuming $\mathcal{C}^\Phi_0$ can be reused as $\mathcal{C}^\Phi_1$, no changes have to be made to $\mathcal{C}^\Phi_0$, yielding an execution time in $O(1)$.
\end{proof}


\begin{lemma} \label{lem:monotonicity_forward_navigation}
Let $H$ be a graph, $H_p \subseteq H$, $[v \rightarrow_n w]^\Phi \in V^{N^\Phi}$ a forward navigation node with dependency $n^\Phi_v$, and $\mathcal{C}^\Phi_0$ and $\mathcal{C}^\Phi_1$ configurations such that
$\mathcal{C}^\Phi_0$ is consistent for $[v \rightarrow_n w]^\Phi$ and
$\forall (m, \phi_0) \in \mathcal{C}^\Phi_0(n^\Phi_v) : \exists (m, \phi_1) \in \mathcal{C}^\Phi_1(n^\Phi_v) : \phi_0 \leq \phi_1$.
It then holds for the configuration $\mathcal{C}^\Phi_2 = execute([v \rightarrow_n w]^\Phi, N^\Phi, H, H_p, \mathcal{C}^\Phi_1)$ that
$\forall (m, \phi_0) \in \mathcal{C}^\Phi_0([v \rightarrow_n w]^\Phi) : \exists (m, \phi_2) \in \mathcal{C}^\Phi_2([v \rightarrow_n w]^\Phi) : \phi_0 \leq \phi_2$.
\end{lemma}

\begin{proof}
Let $(m, \phi_0) \in \mathcal{C}^\Phi_0([v \rightarrow_n w]^\Phi)$. Since $\mathcal{C}^\Phi_0$ is consistent for $[v \rightarrow_n w]^\Phi$, it must hold that there exists some tuple $(m_v, \phi_0) \in \mathcal{C}^\Phi_0(n^\Phi_v)$ such that $m(v) = m_v(v)$.

By the assumption regarding $\mathcal{C}^\Phi_1$, there must then be a tuple $(m_v, \phi_1)$ in $\mathcal{C}^\Phi_1(n^\Phi_v)$ such that $\phi_0 \leq \phi_1$.
By the semantics of the forward navigation node, it then follows that $(m, \phi_1) \in \mathcal{C}^\Phi_2([v \rightarrow_n w]^\Phi)$.

It thus follows that $\forall (m, \phi_0) \in \mathcal{C}^\Phi_0([v \rightarrow_n w]^\Phi) : \exists (m, \phi_2) \in \mathcal{C}^\Phi_2([v \rightarrow_n w]^\Phi) : \phi_0 \leq \phi_2$.
\end{proof}

\begin{lemma} \label{lem:monotonicity_backward_navigation}
Let $H$ be a graph, $H_p \subseteq H$, $[w \leftarrow_n v]^\Phi \in V^{N^\Phi}$ a backward navigation node with dependency $n^\Phi_w$, and $\mathcal{C}^\Phi_0$ and $\mathcal{C}^\Phi_1$ configurations such that
$\mathcal{C}^\Phi_0$ is consistent for $[w \leftarrow_n v]^\Phi$ and
$\forall (m, \phi_0) \in \mathcal{C}^\Phi_0(n^\Phi_w) : \exists (m, \phi_1) \in \mathcal{C}^\Phi_1(n^\Phi_w) : \phi_0 \leq \phi_1$.
It then holds for the configuration $\mathcal{C}^\Phi_2 = execute([w \leftarrow_n v]^\Phi, N^\Phi, H, H_p, \mathcal{C}^\Phi_1)$ that
$\forall (m, \phi_0) \in \mathcal{C}^\Phi_0([w \leftarrow_n v]^\Phi) : \exists (m, \phi_2) \in \mathcal{C}^\Phi_2([w \leftarrow_n v]^\Phi) : \phi_0 \leq \phi_2$.
\end{lemma}

\begin{proof}
Let $(m, \phi_0) \in \mathcal{C}^\Phi_0([w \leftarrow_n v]^\Phi)$. Since $\mathcal{C}^\Phi_0$ is consistent for $[w \leftarrow_n v]^\Phi$, it must hold that there exists some tuple $(m_w, \phi_0) \in \mathcal{C}^\Phi_0(n^\Phi_w)$ such that $m(v) = m_w(w)$.

By the assumption regarding $\mathcal{C}^\Phi_1$, there must then be a tuple $(m_w, \phi_1)$ in $\mathcal{C}^\Phi_1(n^\Phi_w)$ such that $\phi_0 \leq \phi_1$.
By the semantics of the backward navigation node, it then follows that $(m, \phi_1) \in \mathcal{C}^\Phi_2([w \leftarrow_n v]^\Phi)$.

It thus follows that $\forall (m, \phi_0) \in \mathcal{C}^\Phi_0([w \leftarrow_n v]^\Phi) : \exists (m, \phi_2) \in \mathcal{C}^\Phi_2([w \leftarrow_n v]^\Phi) : \phi_0 \leq \phi_2$.
\end{proof}

\begin{lemma} \label{lem:monotonicity_projection}
Let $H$ be a graph, $H_p \subseteq H$, $[\pi_Q]^\Phi \in V^{N^\Phi}$ a marking-sensitive projection node with dependency $n^\Phi_\alpha$, and $\mathcal{C}^\Phi_0$ and $\mathcal{C}^\Phi_1$ configurations such that
$\mathcal{C}^\Phi_0$ is consistent for $[\pi_Q]^\Phi$ and
$\forall (m, \phi_0) \in \mathcal{C}^\Phi_0(n^\Phi_\alpha) : \exists (m, \phi_1) \in \mathcal{C}^\Phi_1(n^\Phi_\alpha) : \phi_0 \leq \phi_1$.
It then holds for the configuration $\mathcal{C}^\Phi_2 = execute([\pi_Q]^\Phi, N^\Phi, H, H_p, \mathcal{C}^\Phi_1)$ that
$\forall (m, \phi_0) \in \mathcal{C}^\Phi_0([\pi_Q]^\Phi) : \exists (m, \phi_2) \in \mathcal{C}^\Phi_2([\pi_Q]^\Phi) : \phi_0 \leq \phi_2$.
\end{lemma}

\begin{proof}
Let $(m, \phi_0) \in \mathcal{C}^\Phi_0([\pi_Q]^\Phi)$. Since $\mathcal{C}^\Phi_0$ is consistent for $[\pi_Q]^\Phi$, it must hold that there exists some tuple $(m_\alpha, \phi_0) \in \mathcal{C}^\Phi_0(n^\Phi_\alpha)$ such that $m|_{Q} = m_\alpha|_Q$.

By the assumption regarding $\mathcal{C}^\Phi_1$, there must then be a tuple $(m_\alpha, \phi_1)$ in $\mathcal{C}^\Phi_1(n^\Phi_\alpha)$ such that $\phi_0 \leq \phi_1$.
By the semantics of the marking-sensitive projection node, it then follows that $(m, \phi_2) \in \mathcal{C}^\Phi_2([\pi_Q]^\Phi)$ for some $\phi_2 \geq \phi_1$.

It thus follows that $\forall (m, \phi_0) \in \mathcal{C}^\Phi_0([\pi_Q]^\Phi) : \exists (m, \phi_2) \in \mathcal{C}^\Phi_2([\pi_Q]^\Phi) : \phi_0 \leq \phi_2$.
\end{proof}

\begin{lemma} \label{lem:monotonicity_union}
Let $H$ be a graph, $H_p \subseteq H$, $[\cup]^\Phi \in V^{N^\Phi}$ a marking-sensitive union node with a set of dependencies $N^\Phi_\alpha$, and $\mathcal{C}^\Phi_0$ and $\mathcal{C}^\Phi_1$ configurations such that
$\mathcal{C}^\Phi_0$ is consistent for $[\cup]^\Phi$ and
$\forall n^\Phi_\alpha \in N^\Phi_\alpha : \forall (m, \phi_0) \in \mathcal{C}^\Phi_0(n^\Phi_\alpha) : \exists (m, \phi_1) \in \mathcal{C}^\Phi_1(n^\Phi_\alpha) : \phi_0 \leq \phi_1$.
It then holds for the configuration $\mathcal{C}^\Phi_2 = execute([\cup]^\Phi, N^\Phi, H, H_p, \mathcal{C}^\Phi_1)$ that
$\forall (m, \phi_0) \in \mathcal{C}^\Phi_0([\cup]^\Phi) : \exists (m, \phi_2) \in \mathcal{C}^\Phi_2([\cup]^\Phi) : \phi_0 \leq \phi_2$.
\end{lemma}

\begin{proof}
Let $(m, \phi_0) \in \mathcal{C}^\Phi_0([\cup]^\Phi)$. Since $\mathcal{C}^\Phi_0$ is consistent for $[\cup]^\Phi$, it must hold that $(m, \phi_0) \in \mathcal{C}^\Phi_0(n^\Phi_\alpha)$ for some $n^\Phi_\alpha \in N^\Phi_\alpha$.

By the assumption regarding $\mathcal{C}^\Phi_1$, there must then be a tuple $(m_\alpha, \phi_1)$ in $\mathcal{C}^\Phi_1(n^\Phi_\alpha)$ such that $\phi_0 \leq \phi_1$.
By the semantics of the marking-sensitive union node, it then follows that $(m, \phi_2) \in \mathcal{C}^\Phi_2([\cup]^\Phi)$ for some $\phi_2 \geq \phi_1$.

It thus follows that $\forall (m, \phi_0) \in \mathcal{C}^\Phi_0([\cup]^\Phi) : \exists (m, \phi_2) \in \mathcal{C}^\Phi_2([\cup]^\Phi) : \phi_0 \leq \phi_2$.
\end{proof}

\begin{lemma} \label{lem:monotonicity_marking_filter}
Let $H$ be a graph, $H_p \subseteq H$, $[\phi > x]^\Phi \in V^{N^\Phi}$ a marking filter node with dependency $n^\Phi_\alpha$, and $\mathcal{C}^\Phi_0$ and $\mathcal{C}^\Phi_1$ configurations such that
$\mathcal{C}^\Phi_0$ is consistent for $[\phi > x]^\Phi$ and
$\forall (m, \phi_0) \in \mathcal{C}^\Phi_0(n^\Phi_\alpha) : \exists (m, \phi_1) \in \mathcal{C}^\Phi_1(n^\Phi_\alpha) : \phi_0 \leq \phi_1$.
It then holds for the configuration $\mathcal{C}^\Phi_2 = execute([\phi > x]^\Phi, N^\Phi, H, H_p, \mathcal{C}^\Phi_1)$ that
$\forall (m, \phi_0) \in \mathcal{C}^\Phi_0([\phi > x]^\Phi) : \exists (m, \phi_2) \in \mathcal{C}^\Phi_2([\phi > x]^\Phi) : \phi_0 \leq \phi_2$.
\end{lemma}

\begin{proof}
Let $(m, \phi_0) \in \mathcal{C}^\Phi_0([\phi > x]^\Phi)$. Since $\mathcal{C}^\Phi_0$ is consistent for $[\phi > x]^\Phi$, it must hold that $(m, \phi_0) \in \mathcal{C}^\Phi_0(n^\Phi_\alpha)$.

By the assumption regarding $\mathcal{C}^\Phi_1$, there must then be a tuple $(m_\alpha, \phi_1)$ in $\mathcal{C}^\Phi_1(n^\Phi_\alpha)$ such that $\phi_0 \leq \phi_1$.
By the semantics of the marking filter node, it then follows that $(m, \phi_1) \in \mathcal{C}^\Phi_2([\phi > x]^\Phi)$.

It thus follows that $\forall (m, \phi_0) \in \mathcal{C}^\Phi_0([\phi > x]^\Phi) : \exists (m, \phi_2) \in \mathcal{C}^\Phi_2([\phi > x]^\Phi) : \phi_0 \leq \phi_2$.
\end{proof}

\begin{lemma} \label{lem:monotonicity_marking_assignment}
Let $H$ be a graph, $H_p \subseteq H$, $[\phi := x]^\Phi \in V^{N^\Phi}$ a marking assignment node with dependency $n^\Phi_\alpha$, and $\mathcal{C}^\Phi_0$ and $\mathcal{C}^\Phi_1$ configurations such that
$\mathcal{C}^\Phi_0$ is consistent for $[\phi := x]^\Phi$ and
$\forall (m, \phi_0) \in \mathcal{C}^\Phi_0(n^\Phi_\alpha) : \exists (m, \phi_1) \in \mathcal{C}^\Phi_1(n^\Phi_\alpha) : \phi_0 \leq \phi_1$.
It then holds for the configuration $\mathcal{C}^\Phi_2 = execute([\phi := x]^\Phi, N^\Phi, H, H_p, \mathcal{C}^\Phi_1)$ that
$\forall (m, \phi_0) \in \mathcal{C}^\Phi_0([\phi := x]^\Phi) : \exists (m, \phi_2) \in \mathcal{C}^\Phi_2([\phi := x]^\Phi) : \phi_0 \leq \phi_2$.
\end{lemma}

\begin{proof}
Let $(m, \phi_0) \in \mathcal{C}^\Phi_0([\phi := x]^\Phi)$. Since $\mathcal{C}^\Phi_0$ is consistent for $[\phi := x]^\Phi$, it must hold that $\phi_0 = x$ and there exists some tuple $(m, \psi_0) \in \mathcal{C}^\Phi_0(n^\Phi_\alpha)$.

By the assumption regarding $\mathcal{C}^\Phi_1$, there must then be a tuple $(m_\alpha, \psi_1)$ in $\mathcal{C}^\Phi_1(n^\Phi_\alpha)$.
By the semantics of the marking assignment node, it then follows that $(m, x) \in \mathcal{C}^\Phi_2([\phi := x]^\Phi)$.

It thus follows that $\forall (m, \phi_0) \in \mathcal{C}^\Phi_0([\phi := x]^\Phi) : \exists (m, \phi_2) \in \mathcal{C}^\Phi_2([\phi := x]^\Phi) : \phi_0 \leq \phi_2$.
\end{proof}

\begin{lemma} \label{lem:monotonicity_vertex_input}
Let $H$ be a graph, $H_p \subseteq H$, $[v]^\Phi \in V^{N^\Phi}$ a marking-sensitive vertex input node, and $\mathcal{C}^\Phi_0$ and $\mathcal{C}^\Phi_1$ configurations such that
$\mathcal{C}^\Phi_0$ is consistent for $[v]^\Phi$.
It then holds for the configuration $\mathcal{C}^\Phi_2 = execute([v]^\Phi, N^\Phi, H, H_p, \mathcal{C}^\Phi_1)$ that
$\forall (m, \phi_0) \in \mathcal{C}^\Phi_0([v]^\Phi) : \exists (m, \phi_2) \in \mathcal{C}^\Phi_2([v]^\Phi) : \phi_0 \leq \phi_2$.
\end{lemma}

\begin{proof}
By the semantics of the marking-sensitive vertex input node and $\mathcal{C}^\Phi_0$ being consistent for $[v]^\Phi$ and $H$, it follows that $\mathcal{C}^\Phi_0([v]^\Phi) = \mathcal{C}^\Phi_1([v]^\Phi)$ and thus
$\forall (m, \phi_0) \in \mathcal{C}^\Phi_0([v]^\Phi) : \exists (m, \phi_2) \in \mathcal{C}^\Phi_2([v]^\Phi) : \phi_0 \leq \phi_2$.
\end{proof}

\begin{lemma} \label{lem:monotonicity_join}
Let $H$ be a graph, $H_p \subseteq H$, $[\bowtie]^\Phi \in V^{N^\Phi}$ a marking-sensitive join node with dependencies $n^\Phi_l$ and $n^\Phi_r$, and $\mathcal{C}^\Phi_0$ and $\mathcal{C}^\Phi_1$ configurations such that
$\mathcal{C}^\Phi_0$ is consistent for $[\bowtie]^\Phi$,
$\forall (m, \phi_0) \in \mathcal{C}^\Phi_0(n^\Phi_l) : \exists (m, \phi_1) \in \mathcal{C}^\Phi_1(n^\Phi_l) : \phi_0 \leq \phi_1$, and
$\forall (m, \phi_0) \in \mathcal{C}^\Phi_0(n^\Phi_r) : \exists (m, \phi_1) \in \mathcal{C}^\Phi_1(n^\Phi_r) : \phi_0 \leq \phi_1$.
It then holds for the configuration $\mathcal{C}^\Phi_2 = execute([\bowtie]^\Phi, N^\Phi, H, H_p, \mathcal{C}^\Phi_1)$ that
$\forall (m, \phi_0) \in \mathcal{C}^\Phi_0([\bowtie]^\Phi) : \exists (m, \phi_2) \in \mathcal{C}^\Phi_2([\bowtie]^\Phi) : \phi_0 \leq \phi_2$.
\end{lemma}

\begin{proof}
Let $(m, \phi_0) \in \mathcal{C}^\Phi_0([\bowtie]^\Phi)$. Since $\mathcal{C}^\Phi_0$ is consistent for $[\bowtie]^\Phi$, there must exist some tuples $(m_l, \psi_0^l) \in \mathcal{C}^\Phi_0(n^\Phi_l)$ and $(m_r, \psi_0^r) \in \mathcal{C}^\Phi_0(n^\Phi_r)$
such that $m_l \cup m_r = m$ and $\phi_0 = max(\psi_0^l, \psi_0^r)$.

By the assumption regarding $\mathcal{C}^\Phi_1$, there must then be tuples $(m_l, \psi_1^l) \in \mathcal{C}^\Phi_1(n^\Phi_l)$ and $(m_r, \psi_1^r) \in \mathcal{C}^\Phi_1(n^\Phi_r)$ such that $\psi_0^l \leq \psi_1^l$ and $\psi_0^r \leq \psi_1^r$.
By the semantics of the marking assignment node, it then follows that $(m, max(\psi_1^l, \psi_1^r)) \in \mathcal{C}^\Phi_2([\bowtie]^\Phi)$, with $max(\psi_1^l \psi_1^r) \geq max(\psi_0^l \psi_0^r)$.

It thus follows that $\forall (m, \phi_0) \in \mathcal{C}^\Phi_0([\bowtie]^\Phi) : \exists (m, \phi_2) \in \mathcal{C}^\Phi_2([\bowtie]^\Phi) : \phi_0 \leq \phi_2$.
\end{proof}

\end{document}